\documentclass[12pt, margin = 1in]{article}

\usepackage{amsmath,amsthm}
\usepackage{fullpage,hyperref}

\usepackage{bm}
\usepackage{natbib}

    \usepackage{paralist}
    \usepackage[compact]{titlesec}
   \titlespacing{\section}{0pt}{2ex}{1ex}
    \titlespacing{\subsection}{0pt}{1ex}{0ex}
    \titlespacing{\subsubsection}{0pt}{0.5ex}{0ex}

\usepackage[ruled,vlined]{algorithm2e}

\usepackage{chngcntr}

\usepackage{amssymb}
\usepackage{graphicx,caption}
\usepackage{tabularx}       
\usepackage[table]{xcolor}
\definecolor{shade}{HTML}{F8F4FF} 

\theoremstyle{plain} \newtheorem{theorem}{Theorem}  \newtheorem{lemma}{Lemma} \newtheorem{corollary}{Corollary}

\theoremstyle{definition} \newtheorem{definition}{Definition}  \newtheorem{example}{Example} 

\theoremstyle{remark}   

\newcommand{\R}{\mathbb{R}}
\newcommand{\M}{\mathcal{M}}
\newcommand{\B}{\mathbb{B}}
\newcommand{\Z}{\mathbb{Z}}

\newcommand{\F}{\mathbb{F}}
\newcommand{\U}{\mathcal{U}}

\newcommand{\E}{\mathbb{E}}

\begin{document}

\newif\ifblinded

\title{Semiparametric discrete data regression with \\ Monte Carlo inference and prediction} 
\date{}

\ifblinded
\author{}
\else

\author{Daniel R. Kowal\thanks{Dobelman Family Assistant Professor, Department of Statistics, Rice University (\href{mailto:Daniel.Kowal@rice.edu}{daniel.kowal@rice.edu}). This material is based upon work supported by the National Science Foundation (SES-2214726).} \ and Bohan Wu\thanks{PhD student, Department of Statistics, Columbia University
(\href{mailto:bw2766@columbia.edu}{bw2766@columbia.edu})}}

\fi

\maketitle

\begin{abstract}%
Discrete data are abundant and often arise as counts or rounded data. These data commonly exhibit complex distributional features such as zero-inflation, over-/under-dispersion, boundedness, and heaping, which render many parametric models inadequate. Yet even for parametric regression models, approximations such as MCMC typically are needed for posterior inference. This paper introduces a Bayesian modeling and algorithmic framework that enables \emph{semiparametric} regression analysis for discrete data with \emph{Monte Carlo} (not MCMC) sampling. The proposed approach pairs a nonparametric marginal model with a latent linear regression model to encourage both flexibility and interpretability, and delivers posterior consistency even under model misspecification. For a parametric or large-sample approximation of this model, we identify a class of conjugate priors with (pseudo) closed-form posteriors. All posterior and predictive distributions are available analytically or via direct Monte Carlo sampling. These tools are broadly useful for linear regression, nonlinear models via basis expansions, and variable selection with discrete data. Simulation studies demonstrate significant advantages in computing, prediction, estimation, and selection relative to existing alternatives. This novel approach is applied successfully to self-reported mental health data that exhibit zero-inflation, overdispersion, boundedness, and heaping. 
\end{abstract}



\section{Introduction}
Discrete data commonly arise either as natural counts or via data coarsening such as rounding. Bayesian regression models must account for this discreteness in order to provide a coherent data-generating process. However, discrete data often present challenging distributional features, including zero-inflation, over-/under-dispersion, boundedness or censoring, and heaping (Figure~\ref{fig:dmhng}). Customized Bayesian models that target these features are too complex for analytical posterior and predictive inference or Monte Carlo sampling, and can incur significant computational burdens from the requisite approximation algorithms such as Markov chain Monte Carlo (MCMC). The contribution of this paper is to develop a new Bayesian semiparametric discrete data regression model that features (i) substantial modeling flexibility, (ii) 
efficient Monte Carlo (not MCMC) computing for posterior and predictive inference, (iii) broad theoretical guarantees, and (iv) for important special cases, conjugate priors with (pseudo) closed-form posteriors. These features advance the state-of-the-art for Bayesian modeling and computing for regression analysis of integer-valued data.


\begin{figure}[h]
\centering
\includegraphics[width=.6\textwidth]{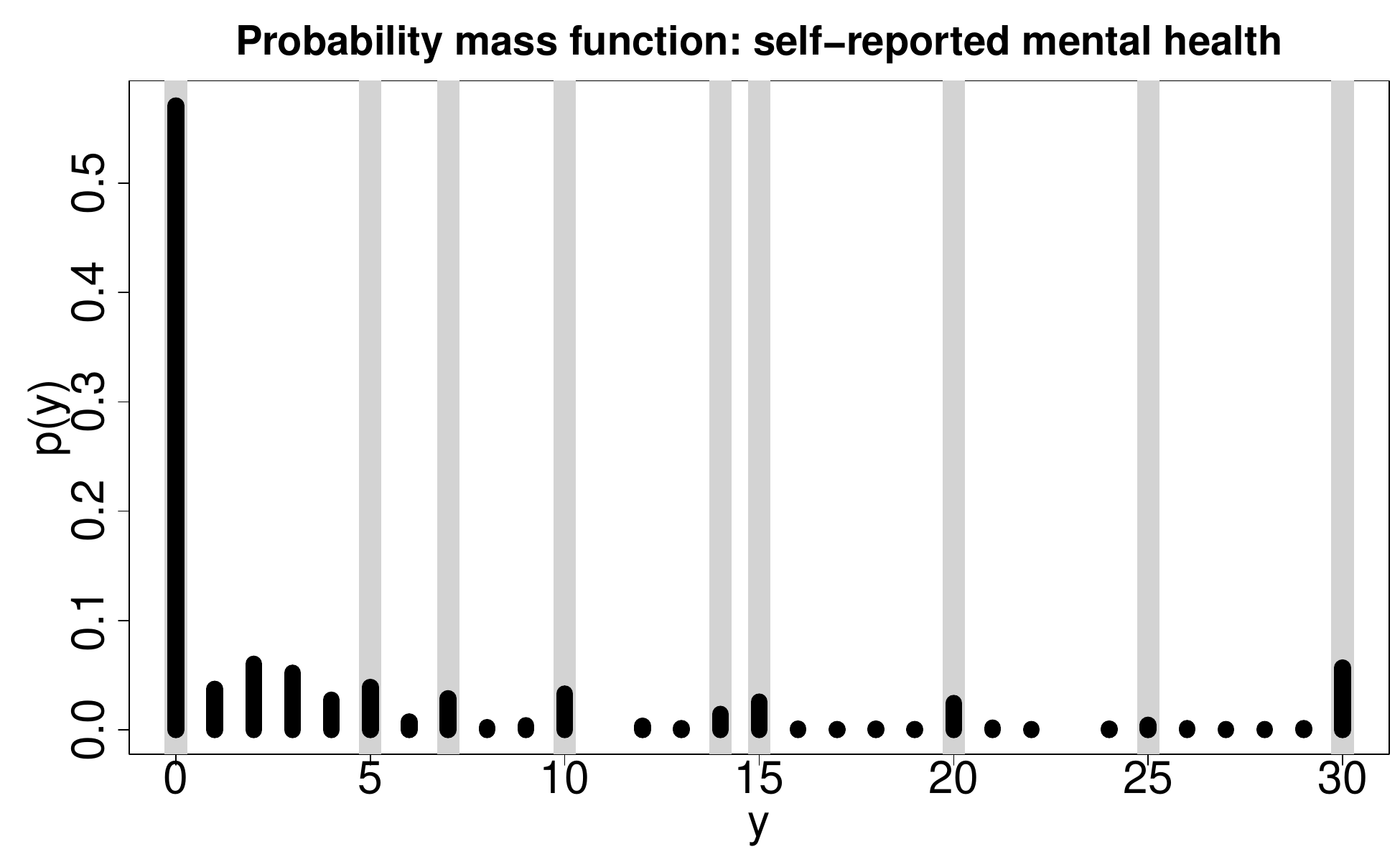}
\caption{\small Empirical probability mass function for self-reported  mental health data (Section~\ref{app}): ``For how many days during the past 30 days was your mental health not good?" These integer-valued data exhibit zero-inflation, overdispersion, boundedness   $(y_{max} = 30)$,  and heaping (gray lines), which present significant challenges for regression modeling (see also Table~\ref{tab:data}).
\label{fig:dmhng}}
\end{figure}

Consider paired data $\{(x_i, y_i)\}_{i=1}^n$ with covariates $x_i \in \mathbb{R}^p$ and integer-valued responses $y_i \in \mathbb{Z}$.  We study the following class of Bayesian discrete 
 data regression models:
\begin{align}
    \label{tr}
    y_i &= h\circ g^{-1}(z_i)  \\
    \label{mod}
    z_i &= \mu_\theta(x_i)  + \epsilon_i
\end{align}
where $h:\mathbb{R}\to \mathbb{Z}$ is a known \emph{rounding operator}, $g:\mathbb{R}\to\mathbb{R}$ is a (monotone increasing) \emph{transformation},  $\mu_\theta$ is a regression function with unknown parameters $\theta$, and $\epsilon_i$ are independent and identically distributed (iid) errors. Informally,  \eqref{tr}--\eqref{mod} is designed to adapt a {continuous data} model for coherent modeling of {discrete data}.  We focus primarily on the Gaussian linear model 
\begin{equation}
\label{lm}
z_i = x_i'\theta + \epsilon_i, \quad \epsilon_i \stackrel{iid}{\sim} N(0,\sigma^2)
\end{equation}
with generalizations for nonlinear versions (Section~\ref{sec-predict}) and variable selection (Section~\ref{sec-sp}). However, our theoretical results apply more broadly for generic continuous data models (Section~\ref{sec-bnp}), such as Gaussian process  regression \citep{rasmussen2006gauss}, BART \citep{chipman2010bart}, and quantile regression \citep{Yu2001}, among many others.

We emphasize that  \eqref{tr}--\eqref{mod}  is \emph{not} equivalent to the common practice of transforming count data (e.g., logarithmically) prior to modeling, which fails to produce a  valid discrete data-generating  process, nor is  it the same as  rounding the predictions from a continuous  data model, which introduces an incongruity between the inferential model and the prediction model. Instead,  \eqref{tr}--\eqref{mod} simultaneously  (i) provides a coherent discrete data-generating process, (ii)  incorporates a transformation for greater representational ability, and (iii) builds upon interpretable regression models.

It is useful to decouple the (known) rounding operator $h$, which defines the support of $y$, and the (typically unknown) transformation $g$, which is a smooth and monotone increasing function (see Figure~\ref{fig:scheme}). First, $h$  is defined via a partition $\{\mathcal{A}_j\}$ of $\mathbb{Z}$ such that $y = j$  whenever $z \in g(\mathcal{A}_j)$. Specifically, we  use intervals of the form $\mathcal{A}_j = [a_j, a_{j+1})$, which implies  that
\begin{equation}
\label{def-h}
y =j  \iff z \in [g(a_j), g(a_{j+1})).
\end{equation}
The rounding operator determines the partition $\{a_j\}$, which is fixed based on the known support of the data. There are several important examples:
\begin{example}  
    \label{ex-round}
    Suppose we observe \emph{rounded data} instead of continuous data, so $y \in \mathbb{Z}$ without loss of generality. 
    When $g(t) = t$ is the identity, $z_i$ represents the unobserved continuous data version of $y_i$. 
    In this case, $y_i = j$ whenever $z_i \in [j-0.5, j+0.5)$, so $h(t) = \lfloor t + 0.5\rfloor$, $\lfloor \cdot \rfloor$ is the floor function, and $a_j = j -0.5$.  This example is revisited in Section~\ref{sims-sparse}.
\end{example}  

\begin{example}\label{ex-round-count}  
    Suppose $y \in \{0,\ldots,y_{max}\}$ is count-valued and bounded. The rounding operator  
    \[
    h(t) = 
    \begin{cases}
        0 & t < 1 \\
        \lfloor t \rfloor & t \in [1, y_{max}] \\
        y_{max} & t > y_{max}
   \end{cases}
   \quad \iff \quad 
       a_j = 
    \begin{cases}
        -\infty & j = 0 \\
        j & j=1,\ldots,y_{max}\\
      \infty & j =  y_{max} +  1
   \end{cases}
    \]
    ensures correct support. 
    The unbounded case simply sets $y_{max} =\infty$, while $y_{max} = 1$ reproduces probit regression for binary data. 
\end{example}   


Model \eqref{tr}--\eqref{mod} is also suitable for censored data, and in fact requires no distinction between boundedness and censoring:
\begin{example}\label{ex-censored-count}
Suppose only $y_i^c = \min\{y_i, C\}$ is observed and let  $c_i = \mathbb{I}\{y_i \ge C\}$ be a known (right) censoring indicator.  The likelihood  is
$p(y \mid \theta, x) = \prod_{i: c_i = 0} p(y_i \mid \theta, x_i) \prod_{i: c_i = 1} p(y_i \ge C \mid \theta, x_i)
=  \prod_{i: c_i = 0} p\{z_i \in [g(a_{y_i}), g(a_{y_i + 1})) \mid \theta, x_i\} \prod_{i: c_i = 1} p\{z_i \ge  g(a_C) \mid \theta, x_i\}$, which is \emph{identical} to the likelihood for $\{y_i^c\}_{i=1}^n$ subject to $y_{max} = C$ 
 in Example~\ref{ex-round-count}. 
\end{example}

The rounding operator provides the correct discrete support for the data-generating process, including (prior and posterior) predictive distributions. Crucially, the specification of the rounding operator does not alter the proposed modeling and computing strategies, which apply for rounded data, counts, and (discrete) bounded data. 
This modularity is highly advantageous and unparalleled among discrete data models. 



\begin{figure}[h]
\centering
\includegraphics[width=.49\textwidth]{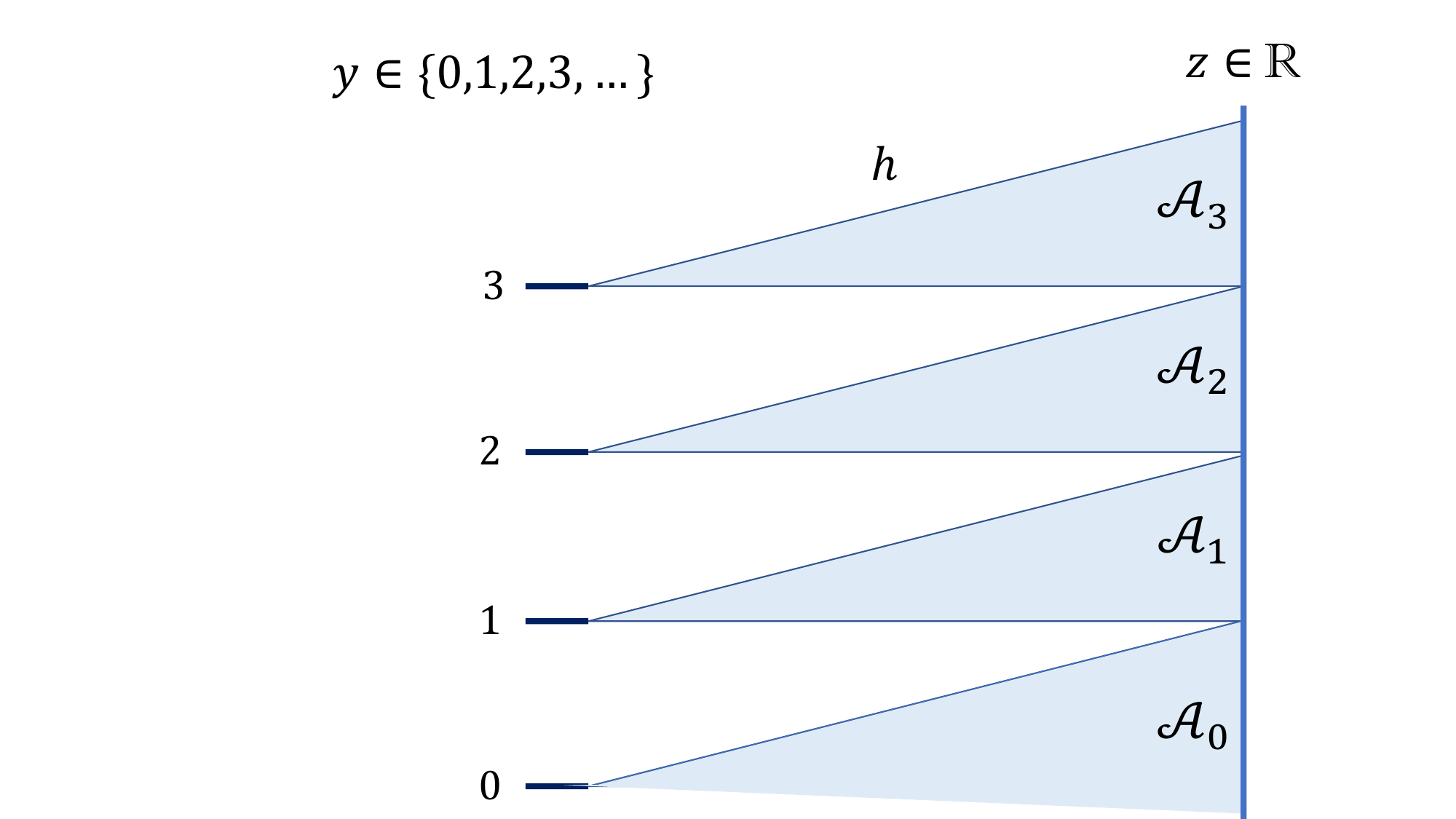}
\includegraphics[width=.49\textwidth]{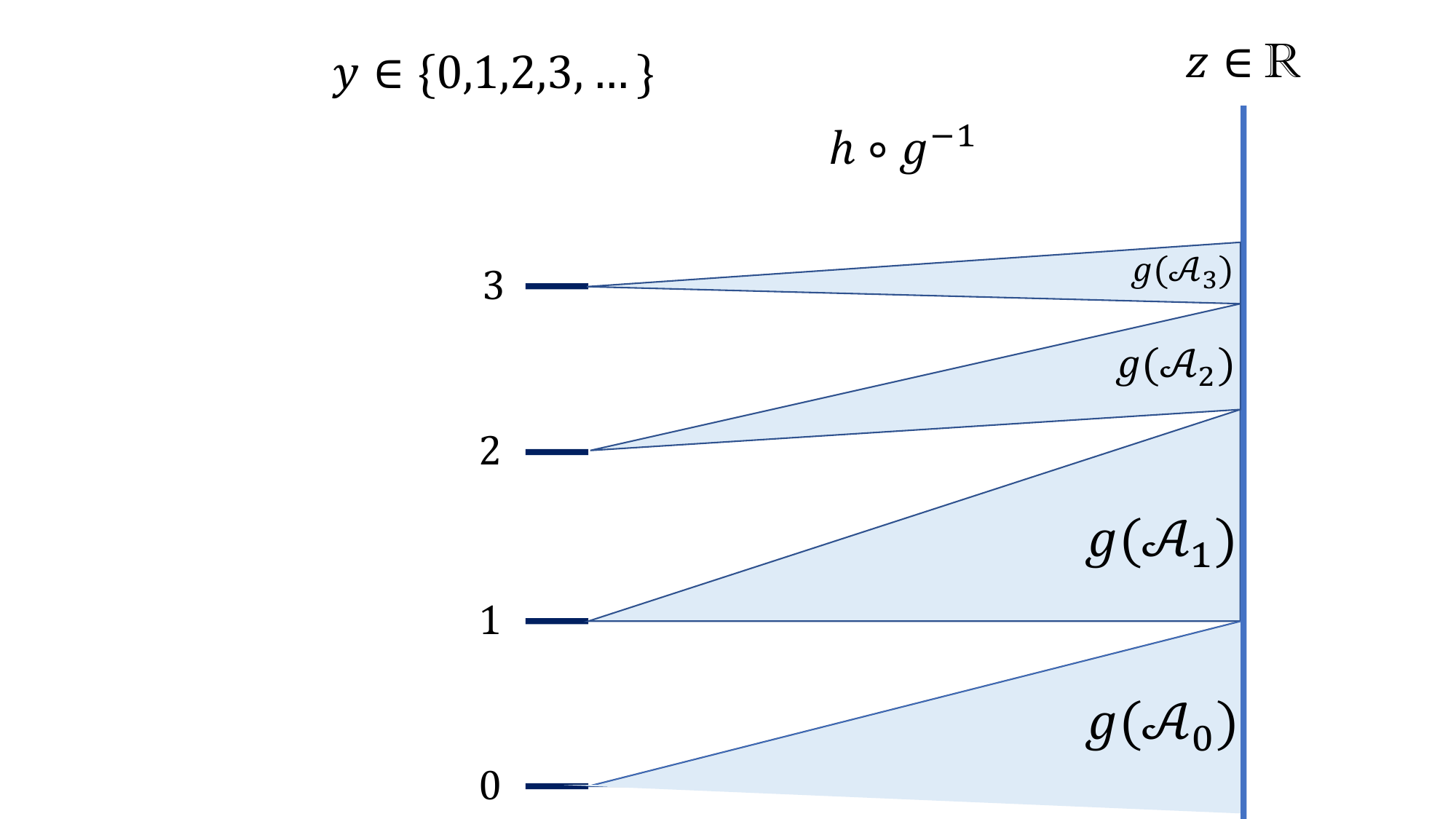}
\caption{\small Illustrating the relationship between the discrete space for $y$ and the continuous space for $z$ under \eqref{tr}--\eqref{mod}. While the rounding operator $h$ simply maps intervals to integers, the learned transformation $g$ warps the latent  data  space to favor certain regions of the domain (skewness, over-/under-dispersion, multimodality) or specific values (zero-inflation, heaping) 
for more accurate predictive inference with discrete data.
\label{fig:scheme}}
\end{figure}
 
We prioritize Bayesian modeling, computing, and theoretical analysis of (i) the transformation $g$ and (ii) the linear coefficients $\theta$.  The role of  $g$  is to provide sufficient marginal distributional flexibility to account for complex discrete data features, such as  zero-inflation, over-/under-dispersion, boundedness or censoring, and heaping  (see Figure~\ref{fig:scheme} and Section~\ref{app}).  
The proposed  \emph{semiparametric} discrete data regression model combines a nonparametric model for $g$ with a parametric model for \eqref{mod}. Our approach hinges on the fundamental decomposition for the joint posterior of $(g,\theta)$:
\begin{equation}
\label{decomp}
p(g, \theta \mid y) =  p(g  \mid y)  \ p(\theta  \mid y, g).
\end{equation}
For  $ p(g  \mid y) $, we establish a link between the marginal distribution of $y$---unconditional on $x$ and $\theta$---and the transformation   $g$, which enables direct application of established Bayesian nonparametric models for this term. Specifically, we design a smoothed Bayesian bootstrap procedure that simultaneously delivers posterior consistency and direct Monte Carlo  sampling for posterior inference. A notable feature of our framework is that both the theoretical analysis and the Monte Carlo algorithms are direct and straightforward. 

The second term, $p(\theta  \mid y, g)$, is equivalent to the posterior of $\theta$ under \eqref{tr}--\eqref{mod} with a \emph{known} transformation $g$. This case is intrinsically interesting, and corrects the popular (yet incoherent) use of parametric transformations (such as logarithmic  or square-root) to model count data using  a continuous model \eqref{mod}. In addition, this case applies under a large-sample approximation to  an unknown transformation  (see Section~\ref{sec-approx}). For the  linear model model \eqref{lm}, we derive conjugate priors, (pseudo) closed-form posteriors, and Monte Carlo inference  for $ p(\theta  \mid y, g) $. 

In conjunction, the resulting methods provide \emph{direct Monte Carlo} posterior and predictive inference when $g$ is modeled nonparametrically and (pseudo) \emph{closed-form} posterior and predictive inference when $g$ is fixed. In both cases, we establish posterior consistency for the unknown parameters, and do not require model correctness for this result.


%


The remainder of the paper is organized as follows. Section~\ref{sec-related}  reviews related work. Section~\ref{sec-bnp} develops the model, algorithms, and theory for inferring the unknown transformation $g$. Section~\ref{sec-post} derives (pseudo) closed-form posterior and predictive distributions, Monte Carlo sampling, and posterior consistency for $\theta$ given  $g$. Section~\ref{sims} includes computational comparisons and simulation studies and Section~\ref{app} contains the application to self-reported mental health data. Section~\ref{disc} concludes. The Appendix includes proofs of all results (Appendix~\ref{sec-proofs}), analytic (Appendix~\ref{sec-pred-an}) and Monte Carlo (Appendix~\ref{a-mc-pred}) strategies for prediction and model selection (Appendix~\ref{sec-model-sel}), 
additional simulation comparisons (Appendix~\ref{sec-sims-add}), and additional details about the real data (Appendix~\ref{sec-app-add}). \texttt{R} code to reproduce all  results is available online  along with  an \texttt{R} package:  \url{www.GitHub.com/drkowal/rSTAR}. 

\section{Related work}\label{sec-related}


Most commonly, Bayesian regression models for count data build upon the Poisson distribution. These parametric regression models are often incapable of modeling  zero-inflation, over-/under-dispersion, boundedness or censoring, and heaping. Despite the simplicity of Poisson regression, the accompanying posterior and predictive distributions typically require approximations such as MCMC. Generalizations such as the Conway-Maxwell-Poisson distribution \citep{sellers2010flexible} or the discrete Weibull distribution \citep{peluso2019discrete} provide moderate improvements for over-/under-dispersion, yet are burdened by more intensive computations. \cite{polson2013bayesian} introduced a data augmentation strategy for Gibbs sampling with Negative Binomial likelihoods, while  \cite{Bradley2018} proposed log-gamma distributions that are conditionally conjugate to the Poisson distribution. These approaches employ MCMC for all posterior and predictive inference and, without significant modifications, cannot account for zero-inflation, bounded or censored data, and heaping. Similarly, copula models for multivariate count data typically specify parametric marginal distributions  \citep{Pitt2006,Genest2007,Nikoloulopoulos2010}, and thus suffer from the same limitations.

The coarsening of continuous data models is a widely useful alternative strategy for modeling count, rounded, and binary data. 
\cite{canale2011bayesian} introduced rounded Gaussian mixture models for count-valued Bayesian nonparametrics.   \cite{canale2013nonparametric} similarly applied rounded Gaussian processes for count regression modeling. \cite{Kowal2020a} extended these models to include both transformation and rounding for BART and additive regression with count data. This general modeling approach has been empirically successful, including for toxicity studies and marketing data  \citep{canale2011bayesian}, tumor counts and asthma inhaler usage \citep{canale2013nonparametric}, species counts and hospital utilization  \citep{Kowal2020a},  and drug overdose counts \citep{King2021}, 
among others. However, these existing methods are severely limited by  computationally intensive and complex MCMC algorithms, and do not provide a direct (closed-form) characterization of the posterior distribution---even with a known transformation---or any theoretical guarantees for inferring an unknown transformation $g$ and the regression parameters $\theta$. 



The many benefits of (pseudo) closed-form posteriors and Monte Carlo sampling have been demonstrated for \emph{binary} 
 data \citep{Durante2019}, including extensions for multinomial 
 data \citep{Fasano2020}. However, the present discrete data setting with $y \in \mathbb{Z}$ introduces unique and significant challenges relative to  binary (or multinomial) data. First, the complex discrete distributional features that we consider---zero-inflation, over-/under-dispersion, boundedness or censoring, and heaping---are simply not relevant (or comparatively trivial) for binary data. These challenging features motivate the inclusion of the transformation $g$, which requires careful consideration to enable flexible modeling, efficient computing, and theoretical analysis. By comparison, transformations are not relevant for binary data: they only change the sample space but not the distribution of the data. Naturally, the omission of $g$ greatly simplifies all phases of the analysis. 
 Further, our approach prioritizes discrete and correctly-supported data-generating processes, and in particular, posterior predictive distributions that may be accessed using Monte Carlo sampling or in (pseudo) closed-form. This task is highly challenging for count and rounded data, which typically require MCMC or other approximations, yet comparatively simple for binary data, for which the posterior predictive distribution is completely characterized by $p(\tilde y = 1 \mid y)$. 

Modeling or estimating an unknown transformation $g$ is widely popular for semiparametric regression \citep{Carroll1988}. Bayesian approaches commonly model $g$ using parametric families \citep{bean2016transformations,Lin2020} or nonparametrically as a monotone and smooth function \citep{lazaro2012bayesian,song2012,Mulgrave2018,Kowal2020a}. However, existing strategies fail to provide easy access to the joint posterior  \eqref{decomp} and instead use complex and often inefficient Metropolis-within-Gibbs sampling algorithms. Non-Bayesian approaches have similarly emphasized the importance of learning a transformation for count data regression \citep{Siegfried2020,Kowal2021b}, count time series analysis \citep{Jia2021}, and (continuous) graphical models \cite{Liu2009}, but do not provide uncertainty quantification that accounts for the unknown transformation.

\section{Bayesian nonparametrics to infer the transformation} \label{sec-bnp}

When the transformation $g$ is learned from the data, model \eqref{tr}--\eqref{lm} is a \emph{semiparametric} regression model with significantly enhanced capabilities to ensure adequacy of the linear representation in \eqref{lm}. A flexible $g$ is crucial for modeling challenging marginal distributional properties, such as zero-inflation, skewness, and heaping (Figure~\ref{fig:scheme}).   However, an unknown $g$ requires careful considerations  for both modeling and computing. Although it may appear straightforward to specify a Bayesian model for $g$ as an unknown and monotone function, existing approaches require complex and inefficient posterior computations \citep{lazaro2012bayesian,song2012,Mulgrave2018,Kowal2020a}. These methods typically apply Metropolis-within-Gibbs sampling steps for $[g \mid y, \theta]$ and $[\theta \mid y, g]$, often with sub-blocks for the coefficients that determine $g$, and must enforce the monotonicity constraint on $g$. 
Thus, it is a nontrivial task to provide modeling flexibility for $g$  that does not impede computational performance.


We introduce  a novel  Bayesian nonparametric (BNP) strategy for $g$ that leverages the vast and powerful array of BNP models for marginal distributions yet enables Monte Carlo  inference for the joint posterior of $g$ and $\theta$. The approach builds upon the fundamental link between the marginal cumulative distribution functions (CDFs) of $y$ and $z$ under \eqref{tr}--\eqref{mod}: $F_Y(j) = p_Y(y \le j) = p_Z\{z < g(a_{j+1})\} = F_Z\{g(a_{j+1})\}$,  which implies that 
\begin{equation} \label{trans-cdf}
g(a_{j+1}) = F_Z^{-1}\{F_Y(j)\}.
\end{equation}
The CDFs $F_Y$ and $F_Z$ of $y$ and $z$, respectively, are \emph{not} simply the distributions from \eqref{tr}--\eqref{mod}: 
 $g$ does not depend on the covariates $x$ or the parameters $\theta$, so the marginal CDFs $F_Y$ and $F_Z$  are unconditional on these terms. 
Specifically, the marginal CDF of the latent data is computed from the latent data model for $z$:  $p_Z(z) = \int \int p_{Z \mid \theta, x} (z) \ p(\theta) \ p_X(x) \, d\theta \, dx$, where $p(\theta)$ is the prior  and $p_X(x)$ denotes the unknown joint density function of the covariates. Under the linear model \eqref{lm}  and a Gaussian prior $\theta \sim N_p(\mu_\theta, \Sigma_\theta)$, the latent data marginal CDF is 
\begin{equation}\label{z-cdf}
F_Z(t) = \int \Phi(t; x'\mu_\theta, x'\Sigma_\theta x  +\sigma^2) \, d F_X(x)
\end{equation}
where $F_X$ is the marginal CDF of the covariates. 
Note that when the covariates are fixed (i.e., non-random), we may substitute the empirical CDF for $F_X$ throughout. Also, $\sigma^2$ is not identified when $g$ is fully nonparametric and may be fixed at $\sigma = 1$.  

Equations \eqref{trans-cdf} and \eqref{z-cdf} showcase the important facts that (i) $g$  is  linked to the marginal distributions of the data, $F_Y$ and $F_X$, and (ii)  the monotonicity constraint is enforced automatically. Thus, it is possible to target $F_Y$ and $F_X$ directly---and  separately---to learn the transformation $g$ with no further constraints needed.
Although a broad variety of BNP models can  be used, we prioritize modeling strategies that admit efficient Monte Carlo sampling for $g$.  We specify Bayesian bootstrap (BB)  models \citep{rubin1981bayesian} for the unknown marginal distributions. The BB may be constructed as a Dirichlet process prior with any well-defined  base measure, and then taking the limit as  the concentration parameter goes to zero. For the covariates $x$, the BB posterior distribution is $[F_X \mid x ] \sim \mbox{DP}(n \hat F_X)$, where $\hat F_X$ is the empirical CDF of $\{x_i\}_{i=1}^n$. Crucially, the BB requires no tuning parameters, applies for mixed data types, and suitably generalizes the empirical CDF to account for uncertainty in the unknown distribution. In the  context of \eqref{tr}--\eqref{lm}, the BB is exceptionally convenient for computing: posterior samples of $F_Z$ in \eqref{z-cdf} are  obtained by sampling $(\alpha_1^x,\ldots,\alpha_n^x) \sim \mbox{Dirichlet}(1,\ldots,1)$ and  computing $\tilde F_Z(t) = \sum_{i=1}^n \alpha_i^x  \Phi(t; x_i'\mu_\theta, x_i'\Sigma_\theta x_i  +\sigma^2)$. 

The BB is also an appealing  choice  for $F_Y$, since it is capable of capturing various marginal distributional features  like zero-inflation, boundedness, and heaping. Similarly, posterior draws of $F_Y$ are obtained  by computing  $\tilde F_Y (j) =  n(n+1)^{-1} \sum_{i=1}^n \alpha_i^y  \mathbb{I}\{y_i \le j\}$ with $(\alpha_1^y,\ldots,\alpha_n^y) \sim \mbox{Dirichlet}(1,\ldots,1)$, where the rescaling $n(n+1)^{-1}$ helps avoid  boundary issues. However, the  primary limitation of the BB is  that it is  only supported on the \emph{observed}  data values $y = \{y_i\}_{i=1}^n$. This undesirable property appears in \eqref{trans-cdf} and, if uncorrected, propagates for all predictive distributions. To resolve  this limitation, we augment \eqref{trans-cdf} with a monotone and smooth interpolation of $g$ at the observed values $y$   \citep{Fritsch1980},  
which appropriately expands the support of the data-generating  process without changing the implied transformation at the observed data points. We emphasize  that this operation is coherent within the Bayesian model: the smooth interpolation is part of the data-generating process, but does not change the models for $F_X$ and $F_Y$. 

The proposed joint sampling algorithm for $(\theta, g)$ is outlined in Algorithm~\ref{alg:joint}. This key algorithm has three important features. First, the computations for drawing $g$ in step 1 are simple, and include draws from a Dirichlet distribution and deterministic computations. Second, Algorithm~\ref{alg:joint} delivers a \emph{joint} posterior draw of $(\theta, g)$, and requires only a draw from $p(\theta \mid y, g)$---i.e., the posterior under a fixed and known transformation. Lastly,  if step 2 is a Monte Carlo draw of $\theta$ (e.g., Algorithm~\ref{alg:g-sim}), then Algorithm~\ref{alg:joint} is a direct Monte Carlo sampler for the joint posterior $p(\theta, g \mid y)$. 
Step 2 may correspond to any other valid sampling scheme for $p(\theta \mid y, g)$, such as a Gibbs sampler (e.g., Algorithm~\ref{alg:gibbs}); even in this case, Algorithm~\ref{alg:joint} provides a highly convenient and efficient blocking structure for $(\theta, g)$. This strategy is particularly useful for large $n$ in Sections~\ref{sims}~and~\ref{app}.

\begin{algorithm}[h]
\SetAlgoLined  
\begin{enumerate}
    \item Simulate $g^* \sim p(g \mid y) $ as follows:
    \begin{enumerate}
    \item Draw $(\alpha_1^x,\ldots,\alpha_n^x) \sim \mbox{Dirichlet}(1,\ldots,1)$. 
    \item Compute  $\tilde F_Z(t) = \sum_{i=1}^n \alpha_i^x  \Phi(t; x_i'\mu_\theta, x_i'\Sigma_\theta x_i  + \sigma^2)$.
        \item Draw  $(\alpha_1^y,\ldots,\alpha_n^y) \sim \mbox{Dirichlet}(1,\ldots,1)$.
\item        Compute $\tilde F_Y(j) =  n(n+1)^{-1} \sum_{i=1}^n \alpha_i^y  \mathbb{I}\{y_i \le j\}$.
   \item Set $\tilde g(a_{j+1}) = \tilde F_Z^{-1}\{\tilde F_Y(a_j)\} $.
\item Apply a  smooth monotonic interpolation of $\tilde g$: $ g^*({y_i}) = \tilde g({y_i})$ for $i=1,\ldots,n$.
    \end{enumerate}
    \item Simulate $\theta^* \sim  p(\theta \mid y, g = g^*)$  (e.g., Algorithm~\ref{alg:g-sim}).
\end{enumerate} 
 \caption{Joint sampler for $(\theta^*, g^*) \sim p(\theta, g \mid y)$.} \label{alg:joint}
\end{algorithm}
A subtle caveat is that Algorithm~\ref{alg:joint} implicitly uses a \emph{surrogate likelihood} for $F_Y$ and thus $g$. 
Since $p_Y(y_i = j) = F_Y(j) - F_Y(j-1) = p_Z\{z_i \in g(\mathcal{A}_j)\} = F_Z\{g(a_{j+1})\} - F_Z\{g(a_j)\}$ due to \eqref{def-h}, the likelihood for $g$, upon marginalizing over $x_i \stackrel{iid}{\sim} F_X$ and $\theta \sim p(\theta)$, is 
\begin{align*}
p(y_1,\ldots, y_n \mid g) &= \int p(\theta) \prod_{i=1}^n [ F_{Z \mid \theta}\{g(a_{y_i + 1})\} - F_{Z \mid \theta}\{g(a_{y_i})\}]  d\theta\\
&= 
  \left\{\prod_{i=1}^n p_Y(y_i)\right\} \int p(\theta) \prod_{i=1}^n \frac{F_{Z \mid \theta}[F_Z^{-1}\{F_Y(y_i)\}] - F_{Z \mid \theta}[F_Z^{-1}\{F_Y(y_i-1)\}]}{F_Y(y_i) - F_Y(y_i - 1)}  d\theta
\end{align*}
using the identification \eqref{trans-cdf}. Instead, the BB model for $F_Y$ uses only the leading term, $\prod_{i=1}^n p_Y(y_i)$. We argue that this term should dominate: the information about $F_Y$ contained in $y$ should far outweigh the information provided indirectly by the covariates $\{x_i\}$ and the prior $p(\theta)$ under the  linear model \eqref{lm}. Indeed, this is validated asymptotically in Section~\ref{sec-theory}: the implied posterior for $g$ is consistent, and thus the omitted term is asymptotically negligible and Algorithm~\ref{alg:joint} is asymptotically valid. 
Finite-sample corrections are available, such as importance sampling---the importance weights are simply the omitted term $p(y_1,\ldots, y_n \mid g)/\prod_{i=1}^n p_Y(y_i)$---but our empirical results (Sections~\ref{sims}~-~\ref{app}) suggest that such no corrections are needed to achieve excellent performance.



\subsection{Posterior consistency for the transformation}\label{sec-theory}

To ensure that the BB for $F_X$ indeed induces a  reasonable model for $F_Z$,  we establish consistency of $F_Z$, and allowing a general (possibly nonlinear or non-Gaussian) model for $[z \mid \theta, x]$:
\begin{theorem}
\label{thm-Fz-con}
Let $F_{Z  \mid \theta, X}$ be a CDF such that $F_{Z  \mid \theta, X=x}$ is continuous in $x$ and define $F_{Z \mid X}(t) = \int F_{Z  \mid \theta, X}(t) \ p(\theta) \, d\theta$, where $p(\theta)$ is the prior.  Let $F_{X,0}$ be the true CDF of $X$ and suppose that $\theta$ and $X$ are independent. Under the BB model for $\tilde F_X$, the marginal CDF   $\tilde F_Z(t) = \int F_{Z \mid X=x}(t) \ d \tilde F_X(x)$ converges 
weakly
to 
$F_{Z,0}(t) = \int F_{Z \mid X=x}(t) \ d F_{X,0}(x)$.
\end{theorem} 
Thus, the BB  model for $F_X$ is broadly applicable, free of tuning,  computationally convenient, and theoretically appealing. All proofs are in Appendix~\ref{sec-proofs}. 

Building upon this result, we establish posterior consistency for the transformation $g$ under the BB models for $F_X$ and $F_Y$, and thus using the proposed surrogate likelihood for $g$.
Specifically, we study the posterior asymptotics on $\M(\Z)$, i.e., the space of non-increasing functions mapping the set of integers $\Z$ to $\R$, and the topology of pointwise convergence  $\mathcal{T}$. 
\begin{theorem}\label{thm-joint-con}
Suppose that under the transformation model \eqref{trans-cdf} with BB models for $F_X$ and $F_Y$ (or equivalently, step 1 of  Algorithm~\ref{alg:joint}), there exists a  true $  \tilde g_0 \in  (\M(\Z), \mathcal{T})$. Let $\tilde g$ be the restriction of $g$ to the integers. Then the posterior distribution $\Pi_n(\tilde g \mid y)$ is weakly consistent at $\tilde g_0$ under $\mathcal{T}$.
\end{theorem}
This result is general: it does \emph{not} require the  linear model \eqref{lm}, and instead only needs the weak convergence of $\tilde F_Z$ established in Theorem~\ref{thm-Fz-con}. Thus, the proposed model for the transformation \eqref{trans-cdf} with BB models for $F_X$ and $F_Y$ is theoretically justified for any continuous (in $z$ and $x$) latent data model $F_{Z\mid\theta,X}$, and such generalizations only require modification of step 1(b) in Algorithm~\ref{alg:joint}.

Since Algorithm~\ref{alg:joint} targets the joint posterior $[\theta, g \mid y]$, we acknowledge the sufficient conditions for joint posterior consistency, which are revisited in Section~\ref{sec-consist}. In particular, the joint posterior is consistent whenever the posterior for $\theta$ under a \emph{fixed and known} transformation is consistent. Let $\mathcal{T}_E$ be the product of Euclidean topology and the topology of pointwise convergence. The following result is an immediate consequence of Theorem~\ref{thm-joint-con}:
\begin{corollary}\label{cor-joint-con}
Suppose that under the transformation model \eqref{trans-cdf} with BB models for $F_X$ and $F_Y$ (or equivalently, step 1 of  Algorithm~\ref{alg:joint}), there exists a true $(\theta_0, \tilde g_0) \in (\Theta \times \M(\Z), \mathcal{T}_E)$, and that the conditional posterior distribution $\Pi_n(\theta \mid y, \tilde g)$ is weakly consistent at $\theta_0$.  
Then the joint posterior distribution $\Pi_n(\theta, \tilde g \mid y)$ is weakly consistent at $(\theta_0, \tilde g_0)$ under $\mathcal{T}_E$. 
\end{corollary}
Corollary~\ref{cor-joint-con} implies that the proposed semiparametric modeling framework delivers posterior consistency in the case of an \emph{unknown} transformation whenever the parametric model is consistent in the case of a \emph{known} transformation---which is often much easier to verify (see Theorem~\ref{thm-theta-con}). 

\subsection{Approximations for efficient computing}\label{sec-approx}
When $n$ is large, the BBs for $F_Y$ and $F_X$ will concentrate at their respective empirical CDFs. This suggests the following large-sample point approximation for the transformation: 
\begin{equation}\label{trans-cdf-approx}
\hat g(a_{j+1}) = \hat F_Z^{-1}\{\hat F_Y(j)\},
\end{equation} 
where $\hat F_Z(t) = n^{-1} \sum_{i=1}^n  \Phi(t; x_i'\mu_\theta, x_i'\Sigma_\theta x_i  +\sigma^2)$ and $\hat F_Y(j) = (n+1)^{-1} \sum_{i=1}^n \mathbb{I}\{y_i \le j\}$, i.e., the Dirichlet weights $\alpha_i^x$ and $\alpha_i^y$  are replaced by $n^{-1}$. When the transformation is fixed at $\hat g$, Algorithm~\ref{alg:joint} instead uses only step 2. This approximation also motivates analysis of the posterior $p(\theta \mid y, g)$ with known transformation, which is considered in the next section.

\section{Posterior inference for the linear coefficients}\label{sec-post}
In this section, we study the posterior distribution $p(\theta \mid y, g)$, which is equivalently the posterior with a \emph{known} transformation. This distribution is interesting for two reasons. First, sampling from $p(\theta \mid y, g)$ is needed to complete Algorithm~\ref{alg:joint} and posterior consistency of $p(\theta \mid y, g)$ is needed to complete Corollary~\ref{cor-joint-con}. Second, model \eqref{tr}--\eqref{lm} with a known transformation is a useful parametric model:  regression analysis  commonly fixes a transformation in advance, such as a logarithmic, square-root, or identity transformation, to improve model adequacy without increasing the number of parameters. Even for unknown $g$, the point approximation \eqref{trans-cdf-approx} is accurate when $n$ is large, and the uncertainty from $[g \mid y]$ may be negligible. However, the  posterior distribution of $\theta$---even for known transformation---has not been derived previously. Thus, the goals are to (i) derive the  posterior distribution $p(\theta \mid y, g)$  in (pseudo) closed-form, (ii) provide Monte Carlo sampling for this distribution, and (iii) establish posterior consistency for this distribution. By design, Monte Carlo sampling and posterior consistency for $[\theta \mid y, g]$ are sufficient for Monte Carlo sampling and posterior consistency, respectively, of the \emph{joint} posterior $[\theta, g \mid y]$. 


We first consider a general (possibly nonlinear or non-Gaussian) model \eqref{mod} and then analyze \eqref{lm} specifically. 
Notational dependence on $g$ is omitted in the conditioning sets, since $g$ is treated as known in this section. Since $p(y_i = j \mid \theta, x) = p_Z\{z_i \in g(\mathcal{A}_j) \mid \theta, x\}$ due to \eqref{def-h},  the posterior distribution is 
\begin{equation}
    \label{posterior-gen}
    p( \theta \mid  y)  = \frac{p( \theta) \ p\{ z \in g(\mathcal{A}_{y}) \mid  \theta, x\}}{p\{ z \in g(\mathcal{A}_{y}) \mid x\}} = \frac{p(\theta) \ p\{g(\mathcal{A}_{y}) \mid \theta, x\}}{\int p(\theta) \ p\{g(\mathcal{A}_{y}) \mid \theta, x\} \ d\theta}
\end{equation}
where $\{ z \in g(\mathcal{A}_{y})\} = \{z_1 \in g(\mathcal{A}_{y_1}), \ldots, z_n \in g(\mathcal{A}_{y_n})\}$ is elementwise (the posterior conditions on $x$ as well). An obvious strategy for posterior inference   is to apply a Gibbs sampling algorithm for $p(\theta , z  \mid y)$. This strategy is outlined in Algorithm~\ref{alg:gibbs} and has appeared previously \citep{Albert1993,canale2011bayesian,canale2013nonparametric,Kowal2020a}. The main advantages of this approach are its relative simplicity and its scalability: the latent data sampler  in step 1 requires only linear computing time and is parallelizable, while step 2 is a standard draw from a Gaussian linear regression model. Indeed, when $n$ is large, we recommend Algorithm~\ref{alg:gibbs} for step 2 of Algorithm~\ref{alg:joint}. Additional sampling steps for other parameters (e.g., the variance components) are easily added to Algorithm~\ref{alg:gibbs}.

\begin{algorithm}[h]
\SetAlgoLined  
\begin{enumerate}
    \item Simulate $z_i^* \sim [z_i \mid y, \theta] = N(x_i'\theta, \sigma^2)$ truncated to $(g(a_{y_i}), g(a_{y_i + 1})]$ for $i=1,\ldots,n$.
    
    \item Simulate $\theta^* \sim [\theta \mid y, z] =  N(Q_\theta^{-1} \ell_\theta, Q_\theta^{-1})$ where $Q_\theta = \sigma^{-2} X'X + \Sigma_\theta^{-1}$ and $\ell_\theta = \sigma^{-2} X'z + \Sigma_\theta^{-1}\mu_\theta.$ 
\end{enumerate} 
 \caption{Gibbs sampler for $(\theta^*, z^*) \sim p(\theta, z \mid y)$ under \eqref{tr}--\eqref{lm} and $\theta \sim N_p(\mu_\theta, \Sigma_\theta).$} \label{alg:gibbs}
\end{algorithm}

However, MCMC algorithms  fail to provide general insights about the posterior or predictive distributions and require customized diagnostics and lengthy simulation runs. Compared to the proposed Monte Carlo (not MCMC) sampler introduced subsequently, the Gibbs sampler in Algorithm~\ref{alg:gibbs} is inefficient for small or moderate $n$, especially as $p$ grows (see Section~\ref{sims}).

Here, we characterize the posterior distribution explicitly: 
\begin{theorem}\label{thm-slct0}
    When $ \mathcal{C} = g(\mathcal{A}_{y})$ is a measurable subset of of $\mathbb{R}^n$, the posterior distribution \eqref{posterior-gen} is a \emph{selection distribution}.
\end{theorem}
The result is self-evident: a selection distribution is defined by random variables of the form $[\theta \mid z \in \mathcal{C}]$ and parametrized in terms of the joint distribution for $(\theta, z)$ and the measurable constraint region $\mathcal{C}$. So, the observation that $p( \theta \mid  y) = p\{ \theta \mid z \in g(\mathcal{A}_{y})\}$ is sufficient. 

General properties of selection distributions are provided by \cite{ArellanoValle2006}. When the prior $p(\theta)$ is closed under a set of transformations or closed under marginalization, these properties propagate to the posterior $p(\theta \mid y)$. When $(\theta, z)$ have a joint multivariate elliptically contoured distribution, the posterior density and distribution functions can be expressed in terms of the location parameters, dispersion parameters, and density generator function; see \cite{ArellanoValle2006} for the explicit forms. We provide the Gaussian case below.
 
\begin{theorem}\label{thm-slct}
    When $(\theta, z)$ are jointly Gaussian and parametrized by $\theta \sim N_p(\mu_\theta, \Sigma_\theta)$,  $z \sim N_n(\mu_z, \Sigma_z)$, and $\mbox{cov}(z, \theta) = \Sigma_{z \theta}$, the posterior distribution is \emph{selection normal}, denoted $[ \theta \mid  y] \sim \mbox{SLCT-N}_{n, p}( \mu_z,  \mu_\theta,  \Sigma_z,  \Sigma_\theta,  \Sigma_{z\theta},  \mathcal{C})$, with constraint region $\mathcal{C} = g(\mathcal{A}_y)$, density  
    \begin{equation}
    \label{density-slct-n}
    p( \theta \mid  y) = \phi_p( \theta;  \mu_\theta,  \Sigma_\theta) \frac{\bar \Phi_n\{\mathcal{C}; \Sigma_{z\theta}  \Sigma_\theta^{-1}( \theta -  \mu_\theta) +  \mu_z,  \Sigma_z - \Sigma_{z\theta} \Sigma_\theta^{-1}\Sigma_{z\theta}'\}}{ \bar\Phi_n(\mathcal{C};  \mu_z,  \Sigma_z)},
\end{equation}
and moment generating function $    M_{[ \theta \mid  z \in \mathcal{C}]}( s) = \exp\big( s'  \mu_\theta +   s' \Sigma_\theta s /2 \big) \bar\Phi_n(\mathcal{C}; \Sigma_{z\theta} s +  \mu_z,  \Sigma_z)/\bar\Phi_n(\mathcal{C};  \mu_z,  \Sigma_z)$, 
where $\phi_p(\cdot;  \mu,  \Sigma)$ denotes the Gaussian density function of a Gaussian random variable with mean $ \mu$ and covariance $ \Sigma$ and $\bar\Phi_n(\mathcal{C};  \mu,  \Sigma) = \int_\mathcal{C} \phi_n( x;  \mu,  \Sigma) \ d  x$. 
\end{theorem}
Crucially, there exists a more constructive representation:
\begin{theorem}\label{thm-mc}
    Let $ V_0 \sim N_n( 0,  \Sigma_z)$ and $ V_1 \sim N_p( 0,  \Sigma_\theta -  \Sigma_{z\theta}'\Sigma_z^{-1}  \Sigma_{z\theta})$ be independent with $ V_0^{(\mathcal{C} -  \mu_z)} \stackrel{d}{=}  [ V_0 \mid  V_0 \in \mathcal{C} -  \mu_z]$ for $\mathcal{C} = g(\mathcal{A}_y)$. Then $[ \theta \mid y] \stackrel{d}{=}  \mu_\theta +   V_1 + \Sigma_{z\theta}'\Sigma_z^{-1}  V_0^{(\mathcal{C} -  \mu_z)}$.
\end{theorem}
 Most important, Theorem~\ref{thm-mc} enables direct Monte Carlo simulation from the posterior, rather than iterative MCMC or other approximations. Specific examples are provided subsequently, often with useful simplifications (see Sections~\ref{sec-lm}~-~\ref{sec-predict}).  
 
 The primary computational bottleneck for evaluating the density  \eqref{density-slct-n}  is the integration of an $n$-dimensional multivariate normal density, whereas the Monte Carlo sampler in Theorem~\ref{thm-mc} is limited by the draw from the $n$-dimensional truncated normal. Our computing strategies leverage recent algorithmic developments. \cite{Botev2017} introduced a minimax exponential tilting strategy that accurately estimates the normalizing constant of a high-dimensional multivariate normal distribution and provides an efficient importance sampling algorithm for truncated multivariate normal variables with an accept-reject algorithm that achieves a high acceptance rate. These algorithms are efficient and accurate for moderate $n$, such as $n \approx 500$, and are implemented in the    \texttt{TruncatedNormal} package in \texttt{R}. For larger $n$, the Gibbs sampler in Algorithm~\ref{alg:gibbs} is recommended, even if the per-draw Monte Carlo efficiency is lower. 

\subsection{Simplifications for the linear model}\label{sec-lm}
Consider the linear model \eqref{lm} with a Gaussian prior $ \theta \sim N( \mu_\theta,  \Sigma_\theta)$ and let $X = (x_1,\ldots, x_n)'$ denote the $n\times p$ covariate matrix. The posterior distribution of the regression coefficients is available as a direct consequence of Theorem~\ref{thm-slct}:
\begin{lemma}\label{star-lm}
    For model \eqref{tr}--\eqref{lm} with $ \theta \sim N( \mu_\theta,  \Sigma_\theta)$, the posterior distribution is   $[ \theta \mid y] \sim \mbox{SLCT-N}_{n, p}( \mu_z =  X  \mu_\theta,  \mu_\theta,  \Sigma_z =  X  \Sigma_\theta  X' +  \sigma^2 I_n,  \Sigma_\theta,  \Sigma_{z\theta} =  X  \Sigma_\theta,  \mathcal{C} = g(\mathcal{A}_{ y}))$.
\end{lemma}
Lemma~\ref{star-lm} enables evaluation of the posterior density  (Theorem~\ref{thm-slct}), direct Monte Carlo simulation from the posterior (Theorem~\ref{thm-mc}), and computation of the posterior moments for the linear model. To illustrate, consider Zellner's $g$-prior for $\theta$: $\Sigma_\theta = \psi \sigma^2 ( X' X)^{-1}$ for $\psi >0$. Algorithm~\ref{alg:g-sim} provides direct simulation from the posterior distribution $\theta^* \sim p(\theta \mid y)$. Notably, this Monte  Carlo sampler only requires standard regression functionals of  $X$, which are a one-time computing cost across Monte Carlo draws. Validity of the algorithms is established in Appendix~\ref{sec-proofs}. 
\begin{algorithm}[h]
\SetAlgoLined  
\begin{enumerate}
    \item Simulate $V_0^* \sim N_n(0, \sigma^2\{ \psi  X ( X' X)^{-1}  X' + I_n\})$ truncated to $g(\mathcal{A}_y) - X\mu_\theta$.
    \item Simulate $V_1^* \sim N_p(0, \sigma^2 \psi(1+\psi)^{-1} (X'X)^{-1})$.
    \item Set  $\theta^* = \mu_\theta  +  V_1^* + \psi(1+\psi)^{-1}( X' X)^{-1} X' V_0^*$.
\end{enumerate} 
 \caption{Monte Carlo sampling for $\theta^* \sim p(\theta \mid y)$ under the $g$-prior.} \label{alg:g-sim}
\end{algorithm}

	Under the $g$-prior, the posterior expectation of the regression coefficients also simplifies:
    \begin{equation}
        \label{post-mean-g}
        \E( \theta \mid  y) =  \mu_\theta + \frac{\psi}{1+\psi} ( X' X)^{-1} X' \hat z
    \end{equation}
    which follows from Theorem~\ref{thm-mc} and Lemma~\ref{star-lm}, 
    where $\hat z  = \E\{z \mid z \in g(\mathcal{A}_y) - X\mu_\theta\} = \E(V_0^*)$ is the  expectation of $z$ truncated to $g(\mathcal{A}_y) - X\mu_\theta$. Using Algorithm~\ref{alg:g-sim}, $\E( \theta \mid  y)$ is easily estimable by replacing $\hat  z$ with the sample mean of Monte Carlo draws of $V_0^*$.     When the prior is centered at zero, $\mu_\theta=0$ and $\hat z = \E\{z \mid z \in g(\mathcal{A}_y\} = \E(z \mid y)$ is the posterior expectation of the latent data $z$ unconditional on $\theta$.  By comparison, the zero-centered $g$-prior with a Gaussian likelihood for $y$ produces a posterior expectation of the same form, but with the observed data $y$ in place of $\hat z$. This comparison is particularly insightful in the context of the rounded data scenario from  Example~\ref{ex-round}. In this case,  $z_i$ is the latent continuous version of the rounded data $y_i$, so $\hat z_i = \E(z_i \mid y)$ is a reasonable proxy for the unobservable continuous data in the posterior expectation \eqref{post-mean-g}.

More broadly, Lemma~\ref{star-lm} suggests a conjugate prior for model  \eqref{tr}--\eqref{lm}. The Gaussian distribution is a special case of the selection normal: $ \theta \sim \mbox{SLCT-N}_{1, p}(\mu_z = 0,  \mu_\theta, \Sigma_z = 1,  \Sigma_\theta,  \Sigma_{z\theta} =  0_p',  \mathcal{C} = \mathbb{R})$, where the moments and constraints on  $z$ are irrelevant as long as $ \Sigma_{z\theta} =  0$. By generalizing these additional arguments, we obtain a richer class of conjugate priors: 
\begin{lemma}
\label{lem:2}
    Under the prior $ \theta \sim \mbox{SLCT-N}_{n_0, p}( \mu_{z_0},  \mu_\theta,  \Sigma_{z_0},  \Sigma_\theta,  \Sigma_{z_0\theta},  \mathcal{C}_0)$ and the model \eqref{tr} and \eqref{lm}, the posterior is $[\theta \mid y] \sim \mbox{SLCT-N}_{n_0 + n, p}(
         \mu_{z_1}, 
         \mu_\theta, 
         \Sigma_{z_1},  
         \Sigma_\theta, 
          \Sigma_{z_1\theta},  
        \mathcal{C}_1)$ with
    \[
         \mu_{z_1} = \begin{pmatrix}  \mu_{z_0} \\  X  \mu_\theta \end{pmatrix}, \
          \Sigma_{z_1} = \begin{pmatrix}  \Sigma_{z_0} &  \Sigma_{z_0\theta}  X' \\  X  \Sigma_{z_0\theta}' &  X  \Sigma_\theta  X' +  \sigma^2 I_n \end{pmatrix}, \  
          \Sigma_{z_1\theta} = \begin{pmatrix}  \Sigma_{z_0\theta} \\  X  \Sigma_\theta \end{pmatrix},  \
        \mathcal{C}_1 = \mathcal{C}_0 \times g(\mathcal{A}_{ y}).
    \]
\end{lemma}
Lemma~\ref{lem:2} provides sequential updating based on multiple data sets, say $y = (y_{1},y_{2})$. When these data sets are conditionally independent and follow the same model \eqref{tr}--\eqref{lm}, the total posterior decomposes as $p(\theta  \mid y) \propto p(\theta \mid y_{1}) \ p(y_{2} \mid \theta)$. Under a Gaussian prior for $p(\theta)$, Lemma~\ref{star-lm} implies that the partial posterior $p(\theta  \mid y_{1})$  is selection normal, which serves as the prior for the  second data set update. Hence, Lemma~\ref{lem:2} updates  the partial posterior $p(\theta \mid y_{1})$ to the total posterior $p(\theta \mid y)$. In concurrent work, \cite{King2021} generalize Lemma~\ref{lem:2} for dynamic linear models   to provide analytic recursions for filtering and smoothing (selection normal) distributions.

\subsection{Posterior consistency under model misspecification}\label{sec-consist}
Continuing with fixed $g$, we seek to establish posterior consistency for $[\theta \mid y]$ under the  linear model \eqref{tr}--\eqref{lm}. The case of known $g$ completes the joint posterior consistency of Corollary~\ref{cor-joint-con}, but is also intrinsically interesting as a parametric linear model for discrete data. Crucially, the following result establishes strong posterior consistency \emph{without} assuming model correctness: 
\begin{theorem}\label{thm-theta-con}
Let $P_0$ be the true distribution of the data with $\{(x_i, y_i)\}_{i=1}^n \stackrel{iid}{\sim} P_0$ and let $P_\theta$ be the data-generating  model induced by \eqref{tr}--\eqref{lm} with $\theta \in \Theta \subseteq \mathbb{R}^p$.  Suppose there exists a unique $\theta_0 \in \mbox{int}(\Theta)$ such that $\theta_0 = \arg\min_{\theta \in \Theta} KL(P_0, P_\theta)$. 
If  $\vert \mathbb{E}_{P_0}\log p_\theta (y \mid x) \vert <\infty$ for all $\theta \in \Theta$ and the prior distribution satisfies $\Pi\{\mathcal{U}_\epsilon(\theta_0)\} > 0$ for all $\epsilon > 0$, then the posterior distribution satisfies $\Pi_n\{\mathcal{U}_\epsilon(\theta_0)\} \to 1$ almost surely $[P_0]$ for all $\epsilon > 0$, where  $\mathcal{U}_\epsilon(\theta_0) = \{\theta \in \Theta: \Vert \theta - \theta_0 \Vert_2 < \epsilon\}$.
\end{theorem}
Usually, $\Theta = \mathbb{R}^p$, so $\theta_0 \in \mbox{int}(\Theta)$ and any prior with density $\pi(\theta) >0 $ for all $\theta \in \mathbb{R}^p$ that is continuous in a neighborhood of $\theta_0$ will satisfy the necessary conditions for the prior distribution. The likelihood condition refers to  $p_\theta(y \mid x) = \bar\Phi\{g(\mathcal{A}_y); x'\theta, \sigma^2\}$, for which the expected [$P_0$] logarithm of this term must be finite. Although weaker conditions exist, an interpretable and sufficient set of conditions for $y \in \mathbb{Z}$ with $\mathcal{A}_y = [y, y+1)$ is 
$\mathbb{E}_{P_0}\vert  g^3(y)\vert < \infty$, $\mathbb{E}_{P_0}\Vert   g^2(y) x \Vert_2 < \infty$, and $\mathbb{E}_{P_0}\Vert  xx'\Vert_2 < \infty$; see the supplementary material for details. 
Alternatively, we may simply require that $KL(P_0, P_\theta) < \infty$ for all $\theta \in \Theta$, which establishes regularity of $P_\theta$ relative to $P_0$. 
Generalizations for latent models $[z \mid \theta, x]$ that deviate from the linear model \eqref{lm} are available; in fact, our proof only uses the linearity of \eqref{lm} to establish convexity of   $-\log p_\theta(y \mid x)$ in $\theta$. However, such generalizations must carefully consider the suitability of the prior and likelihood conditions. 
 
 Theorem~\ref{thm-theta-con} implies that  the posterior under \eqref{tr}--\eqref{lm} concentrates around the ``best" model-based approximation to the true data-generating process, regardless of whether that model is specified correctly. In concert,  Theorem~\ref{thm-joint-con} and Theorem~\ref{thm-theta-con} deliver posterior consistency for \emph{both} the unknown transformation and the linear coefficients in \eqref{lm} (Corollary~\ref{cor-joint-con}). 
 Since the linear term \eqref{lm} is  the only parametric structure in the proposed model, these results provide exceptional generality for the proposed approach.

\subsection{Prediction and nonlinear models}\label{sec-predict}
A key benefit of \eqref{tr}--\eqref{mod} is that the data-generating process and all predictive distributions are discrete and match the support of $y$. Consider the posterior predictive distribution at the $\tilde n \times p$ covariate matrix  $\tilde X = (\tilde x_1,\ldots, \tilde x_{\tilde n})'$. Because of the link in \eqref{tr}, the joint predictive distribution of $\tilde y = (\tilde y_1,\ldots, \tilde y_{\tilde n})'$ can be expressed via the latent variables, 
  $  p(\tilde y \mid y) = p\{\tilde z \in g(\mathcal{A}_{\tilde y}) \mid y\}  = \int p\{\tilde z \in g(\mathcal{A}_{\tilde y}) \mid \theta\} \ p(\theta \mid y) \ d\theta.
  $
We consider multiple computing strategies for (functionals of) the posterior predictive distribution. The simplest option is to leverage the Monte Carlo sampling algorithms for $p(\theta \mid y)$  (e.g., Theorem~\ref{thm-mc}) to obtain latent predictive samples from $[\tilde z \mid y]$,  and then map those to $\tilde y$ via  \eqref{tr}:
\begin{lemma}\label{lem-pred}
    A predictive draw $\tilde y^* \sim p(\tilde y \mid y)$ for \eqref{tr}--\eqref{lm} is generated as follows:
    sample $\theta^* \sim p(\theta \mid  y)$;
    sample  
    $\tilde z^* \sim N_{\tilde n}(\tilde X \theta^*, \sigma^2 I_{\tilde n})$; 
    set $\tilde y^* = h \circ g^{-1} (\tilde z^*)$.
\end{lemma}
These correctly-supported samples capture the joint dependence among predictive variables at multiple covariates $\tilde X$.  For the linear model \eqref{lm}, the posterior  sampling step is available via Monte Carlo sampling (e.g.,  Algorithm~\ref{alg:g-sim}). As  a result, Lemma~\ref{lem-pred} produces Monte Carlo samples from the joint posterior predictive distribution. Analytic computations for prediction are detailed in Appendix~\ref{sec-pred-an}.


Now suppose the goal is exclusively predictive inference and $p(\theta \mid y)$  is not required. 
To motivate this scenario, consider a nonlinear generalization of \eqref{lm}, 
$       z_i = b'(t_i)\theta + \epsilon_i,
$ with  $\epsilon_i  \sim N(0,\sigma^2) 
$, 
where $z_i = z(t_i)$ and   $b'(t) = (b_1(t), \ldots, b_p(t))$ are known basis functions, such as splines  or wavelets. \cite{canale2011bayesian} considered a similar model for count data, but 
required an MCMC algorithm for posterior inference. 
Let $X = (b(t_1), \ldots, b(t_n))'$ denote the basis matrix. Basis expansions  are usually accompanied by a  smoothness prior of the form $\theta \sim N_p(0, \psi \sigma^2  \Omega^{-})$, where  $\psi > 0$ controls the smoothness and $\Omega$ is a known positive semidefinite  penalty matrix. \cite{Scheipl2012} showed that $X$ and $\Omega$ can be reparametrized such that the prior distribution of $X\theta$ is unchanged, but $\Omega = I_p$ is the identity and $X'X$ is diagonal. Hence, we proceed without loss of generality under the assumption that $X'X = \mbox{diag}\{d_j\}$ and $\theta \sim N_p(0, \psi \sigma^2  I_p)$.

   Algorithm~\ref{alg:nl-pred} provides Monte Carlo simulation from the joint predictive distribution 
   $p\{\tilde y(t) \mid y\}$ across multiple points $t_{1'},\ldots, t_{\tilde n'}$ for a discrete and nonlinear regression model.  Most notably, the key terms are computable without any matrix inversions and avoid any excess sampling steps from the posterior distribution of $\theta$.  The most demanding step is  the draw of $V_0^*$ from the $n$-dimensional truncated normal, but fortunately this draw is shared among all choices of points $t_{1'},\ldots, t_{\tilde n'}$.  Notably, Algorithm~\ref{alg:nl-pred} may be inserted in step 2 of Algorithm~\ref{alg:joint} to instead provide a joint posterior prediction draw of $(\tilde y(t), g)$, which enables Monte Carlo predictive inference even when the transformation $g$ is unknown.
   \begin{algorithm}[h]
\SetAlgoLined  
\begin{enumerate}
    \item Simulate $V_0^* \sim N_n(0, \sigma^2 (\psi X X' + I_n))$ truncated to $g(\mathcal{A}_y)$.
    \item Simulate $\tilde  V_1^* \sim N_{\tilde n}(0,  \sigma^2 [\psi \tilde X \ \mbox{diag}\{1/(1 + \psi d_j)\} \ \tilde X' +   I_{\tilde n}])$.
    \item Compute  $\tilde z^* =  \tilde V_1^* + \psi \tilde X \ \mbox{diag}\{1 - \psi d_j/(1 + \psi d_j)\} \ \tilde X' V_0^*$.
    \item Set $\tilde y^* = h\circ g^{-1}(\tilde z^*)$.
\end{enumerate}
 \caption{Monte Carlo sampling for $\tilde y^* \sim p(\tilde y \mid y)$ at the observation points $t_{1'},\ldots, t_{\tilde n'}$ with $\tilde X =(b(t_{1'}), \ldots, b(t_{\tilde n'}))' $.} \label{alg:nl-pred}
\end{algorithm}

Algorithm~\ref{alg:nl-pred} is not limited to the nonlinear basis expansions model: in fact, it provides predictive simulation from  any linear model \eqref{tr}--\eqref{lm} with a prior of the form $\theta \sim N_p(0, \psi \sigma^2  \Omega^{-})$ for a known positive semidefinite matrix $\Omega$. Although the reparametrization of $X\theta$ changes the interpretation of $X$ and $\theta$, it does not change the data-generating process, so the predictive draws remain valid. Further, the reparametrization is a one-time computing cost, yet the computational simplifications are realized for every draw from Algorithm~\ref{alg:nl-pred}. 
A related sampling strategy for the $g$-prior is provided in Appendix~\ref{a-mc-pred}. 




\subsection{The sparse means problem for discrete data}\label{sec-sp}
The classical sparse means problem considers continuous data $[z_i \mid \theta_i] \sim N(\theta_i, \sigma^2)$ with the goal of identifying the nonzero means, $\{i: \theta_i \ne 0\}$ \citep{CvdV}. Sparsity is introduced via the spike-and-slab model 
\begin{equation}
    \label{ssmod}
    [\theta_i \mid \gamma_i]  \sim \gamma_i N(0, \sigma^2\psi)  + (1-\gamma_i) \delta_{\{0\}},
\end{equation}
where $\gamma_i \in \{0,1\}$ is the inclusion indicator with $p(\gamma_i=1 \mid \pi) = \pi$ for $i=1,\ldots,n$. Here, we instead consider discrete data $y_i \in \mathbb{Z}$ under \eqref{tr}--\eqref{lm}, such as the normal means model for rounded data (see 
Example~\ref{ex-round}). 
Conditional on the inclusion indicators, \eqref{ssmod} is a special case of \eqref{lm} with $X = I_n$ and $n=p$, where the prior for the active components $[\theta_i \mid \gamma_i = 1] \sim N(0, \sigma^2 \psi)$ is a special case of the Gaussian priors considered previously. This scenario is also useful for wavelet models and can be extended to other spike-and-slab priors.

Although it is possible to compute the marginal posterior $p(\gamma \mid y)$ directly---including more generally for model selection (Appendix~\ref{sec-model-sel})---the search over the model space requires consideration of $2^n$ models, which is computationally prohibitive even for moderate $n$. Instead, we may employ a stochastic search to sample from $[\gamma \mid y]$. Specifically, we apply a Gibbs sampler that cycles through $[\gamma_i \mid y, \gamma_{-i}]$ for $i=1,\ldots,n$, which observes that $p(\gamma_i = 1 \mid y, \gamma_{-i}) = \omega_i/(1+\omega_i)$ for the (full conditional) odds of inclusion 
    \begin{equation}\label{odds}
        \omega_i = \frac{ p(\gamma_i = 1 \mid y , \gamma_{-i})}{p(\gamma_i = 0 \mid y, \gamma_{-i})} 
        = \frac{ \pi }{1-\pi}\frac{ p(y \mid  \gamma_{-i}, \gamma_i = 1 )}{p( y \mid  \gamma_{-i}, \gamma_i = 0)}
    \end{equation}    
    and   $p(y \mid \gamma)= \bar\Phi_n\{\mathcal{C} = g(\mathcal{A}_{y});  0,  \Sigma_z^{(\gamma)}\} = \prod_{i=1}^n \int_{g(y_i - 0.5)}^{g(y_i + 0.5)} \phi_1\{x_i; 0, \sigma^2(1 + \psi \gamma_i)\} \ d  x_i.$  We may expand this Gibbs sampler to incorporate a prior on the inclusion probability, $\pi \sim \mbox{Beta}(a_\pi, b_\pi)$ \citep{Scott2010}, which adds a simple sampling step of the form$[\pi \mid \gamma] \sim \mbox{Beta}(a_\pi + \#\{\gamma_i = 1\}, b_\pi + \#\{\gamma_i = 0\})$ with an accompanying update for $\pi$ in \eqref{odds}.
    
    Inference on the regression coefficients $\theta$ is readily available by observing that $[\theta, \gamma \mid y] = [\theta \mid y, \gamma][\gamma \mid y]$. Given samples of $\gamma$ using the above procedure, we sample from $[\theta \mid y, \gamma]$ using the following fast  version of Algorithm~\ref{alg:g-sim}: for $\{i:\gamma_i = 1\}$,  sample $V_{0,i}^* \sim N\{0, \sigma^2(1 + \psi )\}$ truncated to $[g(y_i - 0.5), g(y_i + 0.5)]$;  sample $V_{1,i}^* \sim N\{0, \sigma^2 \psi(1+\psi)^{-1}\}$; and   set $\theta_i^* =  V_{1,i}^* + \psi(1+\psi)^{-1}  V_{0,i}^*$. We set $\theta_i^* = 0$ for all $\{i: \gamma_i = 0\}$. This sampler is parallelizable across $i=1,\ldots,n$ yet produces draws from the joint posterior of $\theta$. These draws also may be used to obtain predictive samples  by application of Lemma~\ref{lem-pred}. 

    This sampling algorithm blends computational scalability and Monte Carlo efficiency for a semiparametric sparse means model with discrete data. The stochastic search for $[\gamma \mid y]$ is scalable in $n$, while the remaining sampling blocks for $g$, $\theta$, and $\tilde y$ use direct Monte Carlo draws. 


\section{Simulation studies}\label{sims}

\subsection{Evaluating the computational performance} \label{sec-comp}
We assess the computational efficiency of the proposed algorithms relative to state-of-the-art alternatives. First, we consider the coefficient posterior $p(\theta \mid y, g)$ and compare the Monte Carlo sampler (Algorithm~\ref{alg:g-sim}) with the Gibbs sampler (Algorithm~\ref{alg:gibbs}), which  has been used previously for related discrete data models \citep{Albert1993,canale2011bayesian,canale2013nonparametric,Kowal2020a}. 
The two samplers are applied to the same model \eqref{tr}--\eqref{lm} with the rounding operator from Example~\ref{ex-round-count} with $y_{max} = \infty$, a fixed transformation \eqref{trans-cdf-approx}, and  a $g$-prior ($\psi = n$). 
For each $n \in \{100, 200, 500\}$ and  $p\in\{10,50\}$, we simulate $X$ and $y$ from the \texttt{NegBin} design in Section~\ref{sims-reg}. We run each algorithm for 1000 samples (after a burn-in of 1000 for the Gibbs sampler) and repeat for 25 replicates. We compute the median effective sample size across the regression coefficients $\{\theta_j\}_{j=1}^p$ and report the result as the percent of the total number of simulations (1000). We also record the raw computing time (using \texttt{R} on a MacBook Pro, 2.8 GHz Intel Core i7). 

The results are presented in Figure~\ref{fig:computing}. 
The proposed Monte Carlo sampler achieves maximal efficiency for all $n,p$. By comparison, the Gibbs sampler is below 60\% efficiency for $p=10$ and declines to about  15\% efficiency for $n=100, p = 50$. The raw computing times are comparable for $n \in \{100, 200\}$ regardless of $p$, while the proposed Monte Carlo sampler incurs some additional cost only for $n=500$. In general, the proposed Monte Carlo sampler demonstrates substantial efficiency gains with increasing $p$, especially when $n$ is moderate. Further, unlike MCMC, the proposed Monte Carlo sampler does not require convergence diagnostics, thinning, or a burn-in.


\begin{figure}[h]
\centering
\includegraphics[width=.45\textwidth]{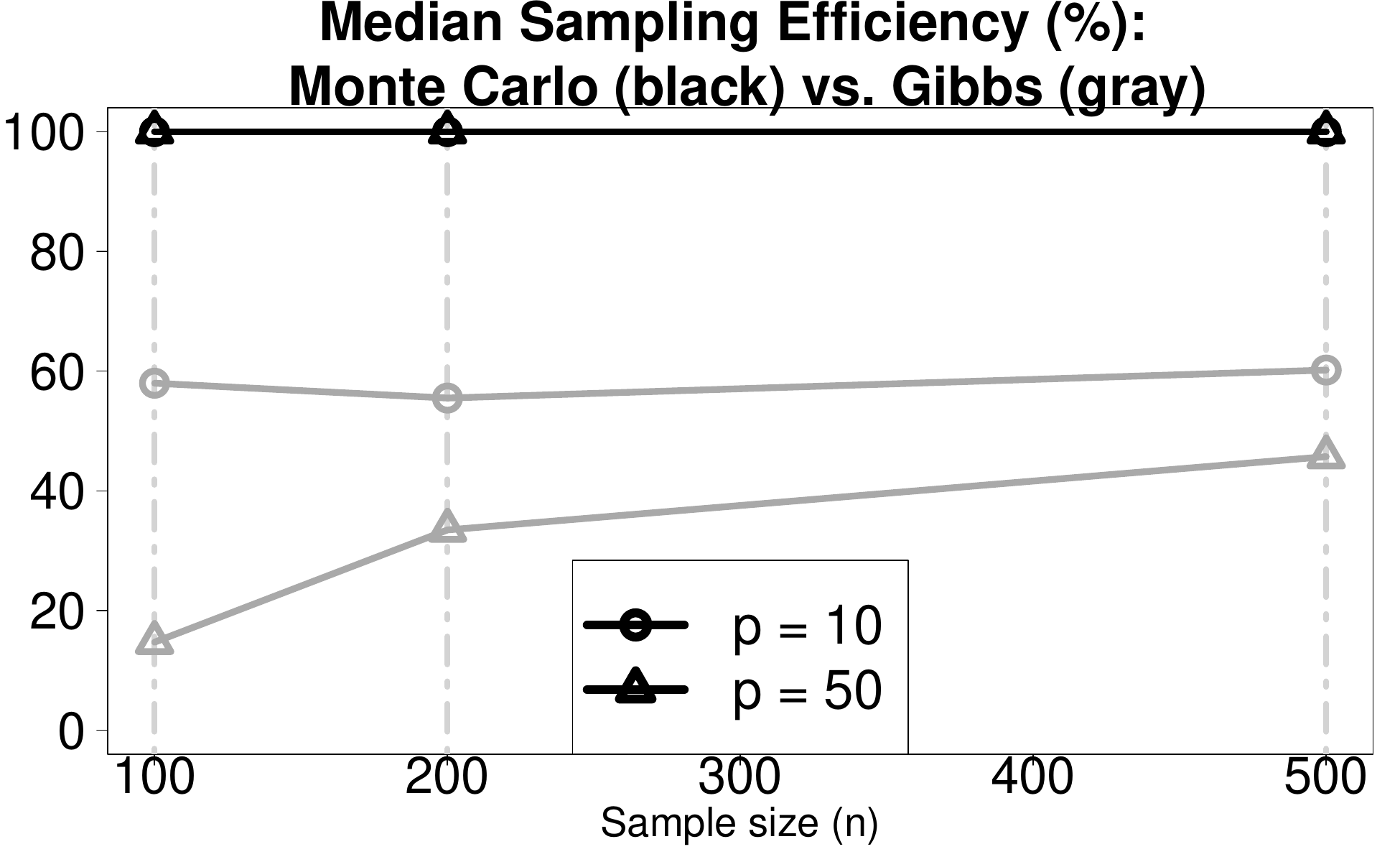}
\includegraphics[width=.45\textwidth]{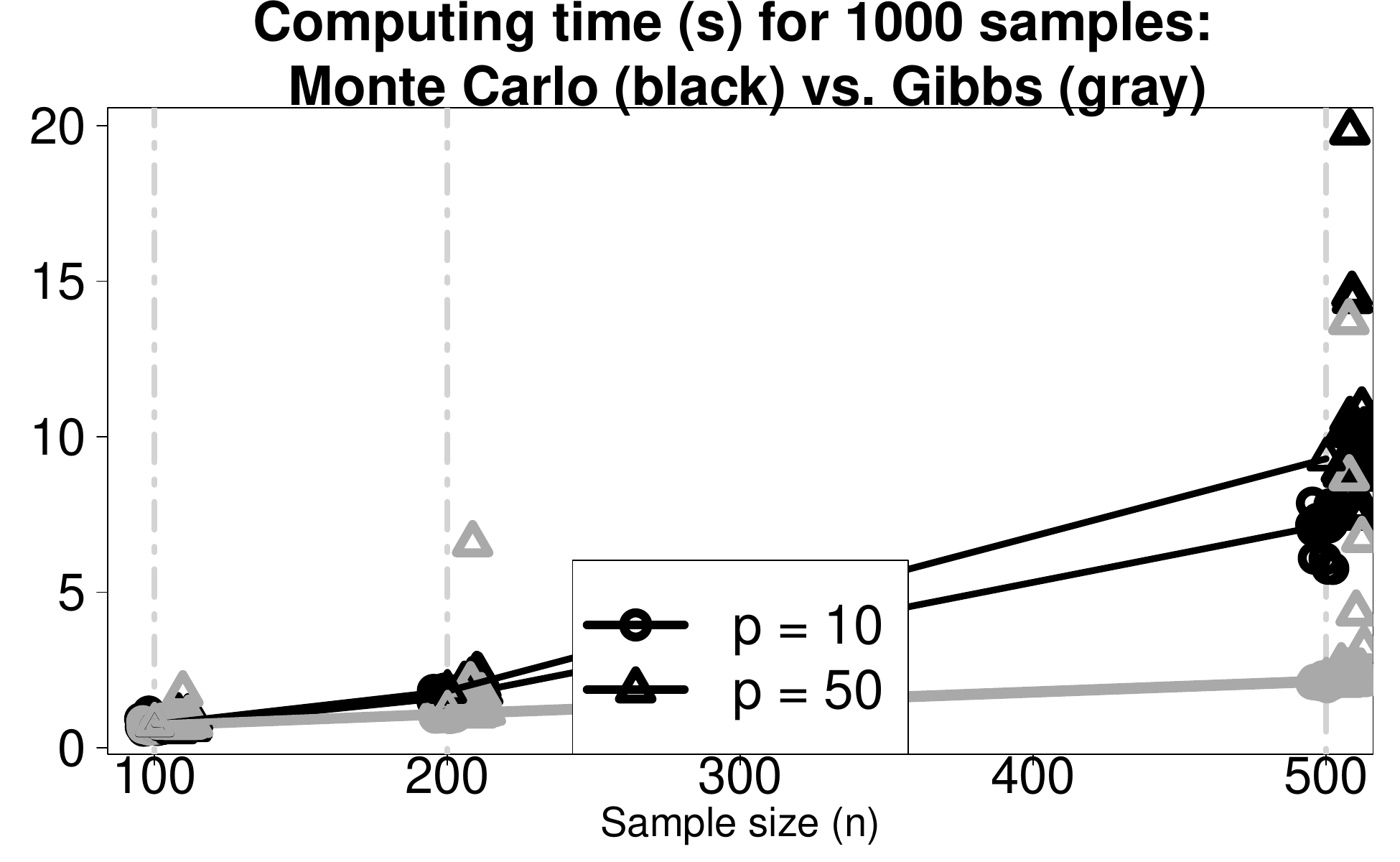}
\caption{\small Median effective sample size across $\{\theta_j\}_{j=1}^p$ reported as the percent of the total number of simulations (left) and raw computing time for 1000 samples (right). Results are computed for 25 synthetic data sets (jittered points, right) and summarized using medians (non-vertical lines) with $n \in \{100, 200, 500\}$ and $p\in\{10,50\}$. The proposed Monte Carlo sampler for $[\theta \mid y]$ demonstrates substantial efficiency gains over the Gibbs sampler for $[\theta, z \mid y]$ with increasing $p$, especially when $n$ is moderate. 
\label{fig:computing}}
\end{figure}


Next, we evaluate posterior sampling strategies for $p(g, \theta \mid y)$ when the transformation is unknown. To compare with the proposed BB sampler in Algorithm~\ref{alg:joint} (step 1), we include the adaptive Metropolis-Hastings sampler from \cite{Kowal2020a}, which uses a monotone I-spline model for $g$. For a fair comparison, we use the same Gibbs sampler for the parameters $\theta$  (Algorithm~\ref{alg:gibbs}). The results  are presented in Figure~\ref{fig:computing-g}. Most notably, the proposed sampler for $[g \mid y]$ achieves maximal efficiency regardless of $n,p$, while the adaptive Metropolis-Hastings sampler is woefully inefficient. The efficiency gains for $[\theta \mid y]$ are less significant---due to the inefficiencies introduced by the Gibbs sampler for $\theta$---but still decisively favor the proposed approach.


\begin{figure}[h]
\centering
\includegraphics[width=.45\textwidth]{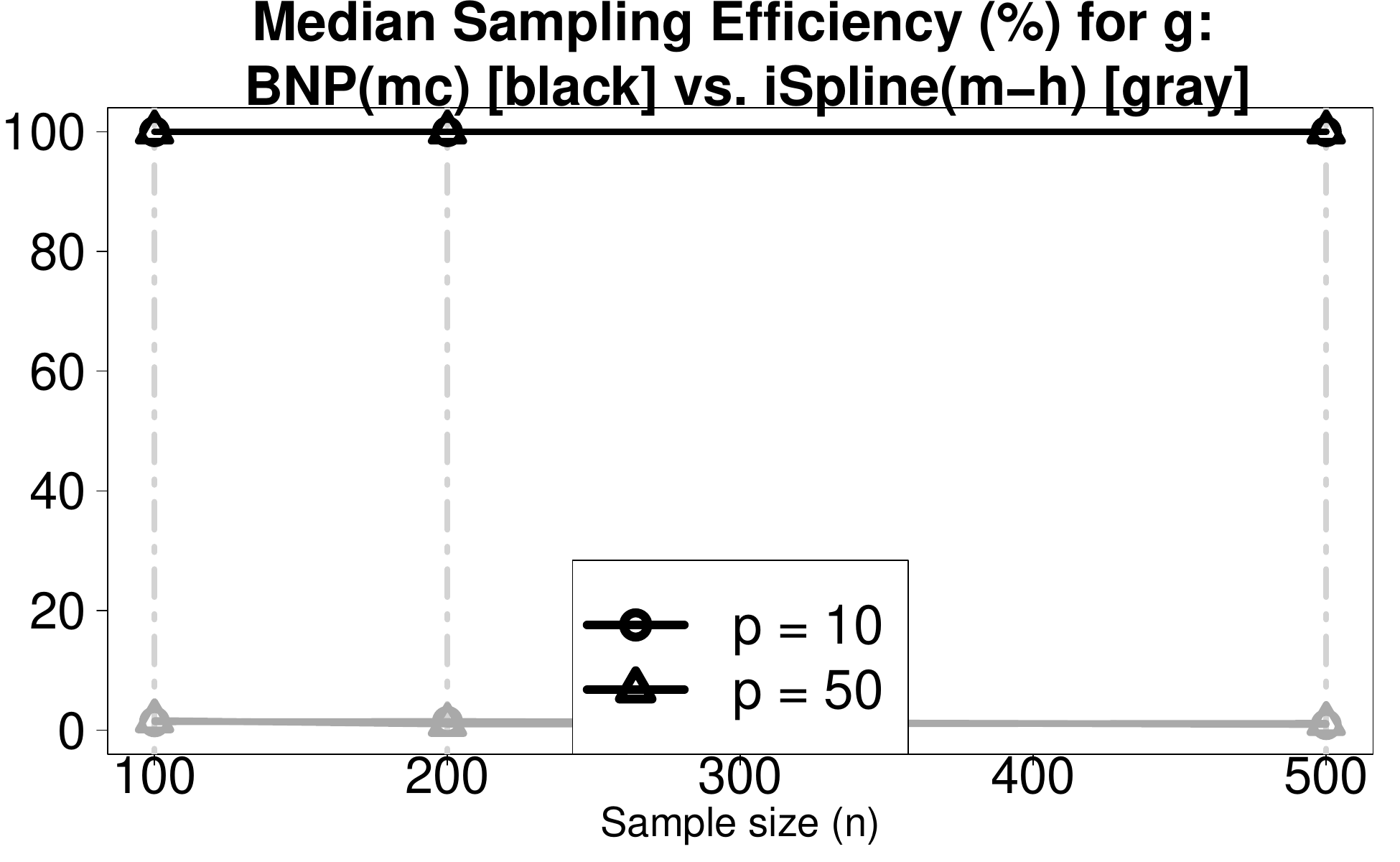}
\includegraphics[width=.45\textwidth]{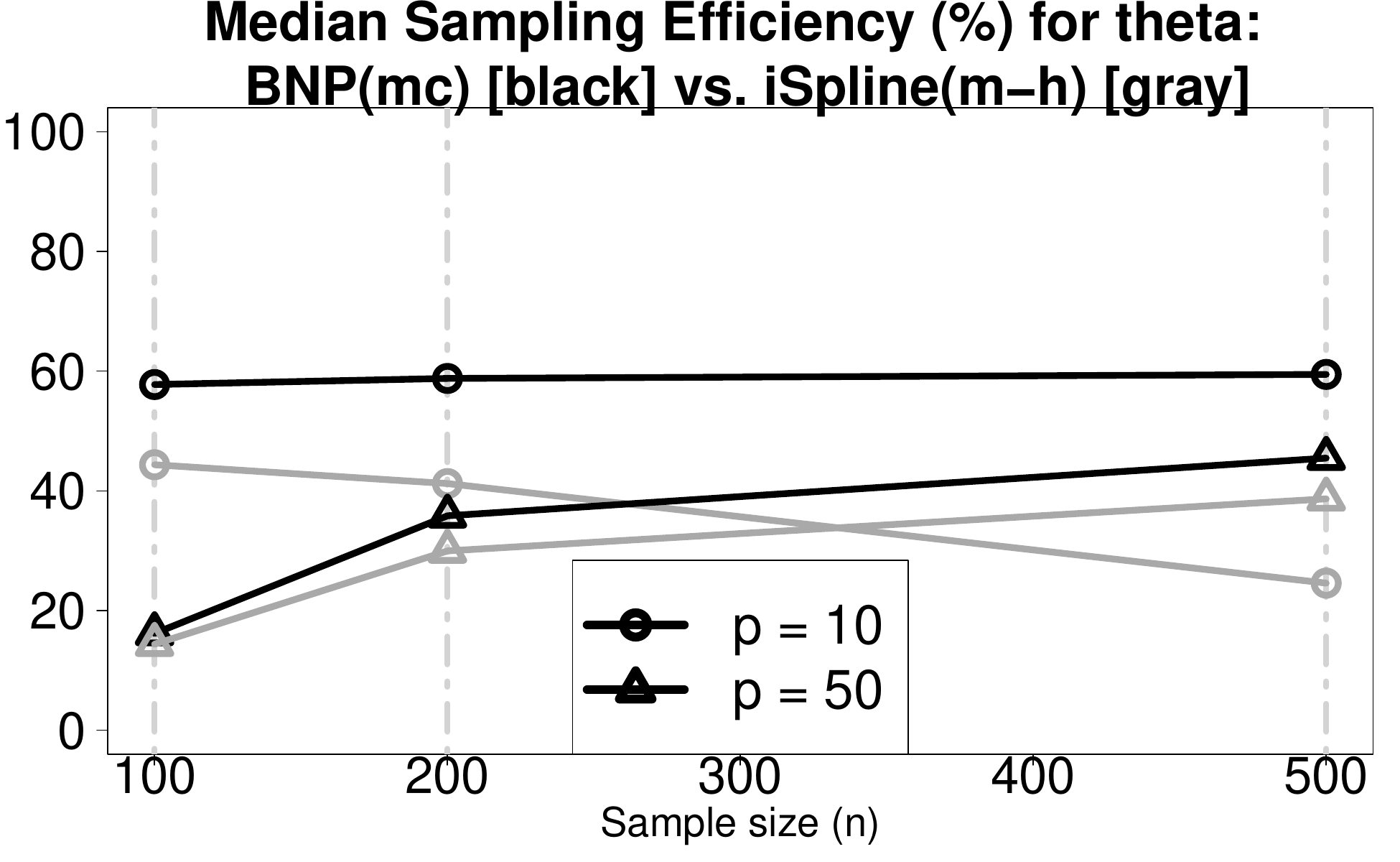}
\caption{\small Median effective sample size for $g$ (left) and $\theta$ (right) reported as the percent of the total number of simulations. Results are computed for 25 synthetic data sets and summarized using medians (non-vertical lines) with $n \in \{100, 200, 500\}$ and $p\in\{10,50\}$. The raw computing times are comparable and omitted for brevity. The efficiency gains  for $p(g \mid y)$ are especially noteworthy. 
\label{fig:computing-g}}
\end{figure}



\subsection{Linear regression with synthetic data}\label{sims-reg}
We evaluate the proposed semiparametric model for prediction and coefficient estimation (see Appendix~\ref{sec-sims-add}) using simulated data. We focus on two data-generating scenarios: \texttt{SemiPar}, which uses \eqref{tr}--\eqref{lm} and mimics the challenging features of the real data in Section~\ref{app}, and \texttt{NegBin}, which uses a Negative Binomial distribution as a more standard benchmark. Both versions feature a linear predictor $x_i'\theta^*$, where the $p=10$ covariates are marginal standard normal with $\mbox{Cor}(x_{ij}, x_{ij'}) = (0.75)^{\vert j - j'\vert}$ and the columns are randomly permuted and augmented with an intercept. The true coefficients are given by the intercept $\theta_0^* =\log(1.5)$, true signals $\theta_j^* = \log(1.25)$ for $j=1,\ldots, p/2 = 5$, and $\theta_j^*=0$ for the remaining coefficients. \texttt{SemiPar} generates data $y$ using \eqref{tr}--\eqref{lm} with the rounding operator from Example~\ref{ex-round-count} with $y_{max} = 30$, the transformation \eqref{trans-cdf-approx} computed from the real data in Section~\ref{app}, and $\sigma = 0.25$. Thus, the \texttt{SemiPar} data exhibit zero-inflation, heaping, and boundedness. \texttt{NegBin} generates each $y_i$ from a Negative Binomial distribution with expectation $\exp(x_i'\theta^*)$ and variance $\exp(x_i'\theta^*)\{1 + \exp(x_i'\theta^*)/10\}$, which exhibits moderate overdispersion. In each case, we consider $n \in\{100, 500\}$ and simulate 100 data sets. 

We implement the proposed semiparametric model  with a  $g$-prior ($\psi = n$), including the fully BNP model for $g$ (\texttt{BNP(mc)}) and the point approximation from \eqref{trans-cdf-approx} (\texttt{BNP(approx)}). For quicker computing across many simulations, we use Algorithm~\ref{alg:joint} and the Gibbs sampler in Algorithm~\ref{alg:gibbs} for 1000 iterations (after discarding a burn-in of 1000). This Gibbs sampler is conveniently blocked: $g$ is sampled from the marginal posterior $[g \mid y]$, while the Gibbs steps for $[\theta \mid y, g, z]$ and $[z \mid y, g, \theta]$ use $p$-dimensional and $n$-dimensional blocks, respectively. For comparison, we include the variation of this linear model from Section~\ref{sec-comp} with a monotone I-spline for $g$ (\texttt{iSpline}) and an adaptive Metropolis-Hastings  sampler \citep{Kowal2020a}.  
Finally, we include  Poisson and Negative Binomial regression models with the default specifications from \texttt{stan\_glm} (with \texttt{family = poisson} and \texttt{family = neg\_binomial\_2}, respectively) in the \texttt{rstanarm} package in \texttt{R}. 

The posterior predictive distributions  are evaluated on a testing data set $(X^{test}, y^{test})$ distributed identically as the training data $(X, y)$. First, we compute ranked probability scores  (Figure~\ref{fig:rps}). For the \texttt{SemiPar} data, the proposed BNP approaches dominate the Poisson and Negative Binomial competitors for both $n=100$ and $n=500$. The \texttt{NegBin} case is more competitive, yet surprisingly the proposed BNP approaches outperform the Poisson and Negative Binomial models. Further, both \texttt{BNP(mc)} and \texttt{BNP(approx)} improve upon the \texttt{iSpline} competitor, which suggests that the proposed sampling strategy in Algorithm~\ref{alg:joint} delivers empirical benefits for prediction.


\begin{figure}[h]
\centering
\includegraphics[width=.45\textwidth]{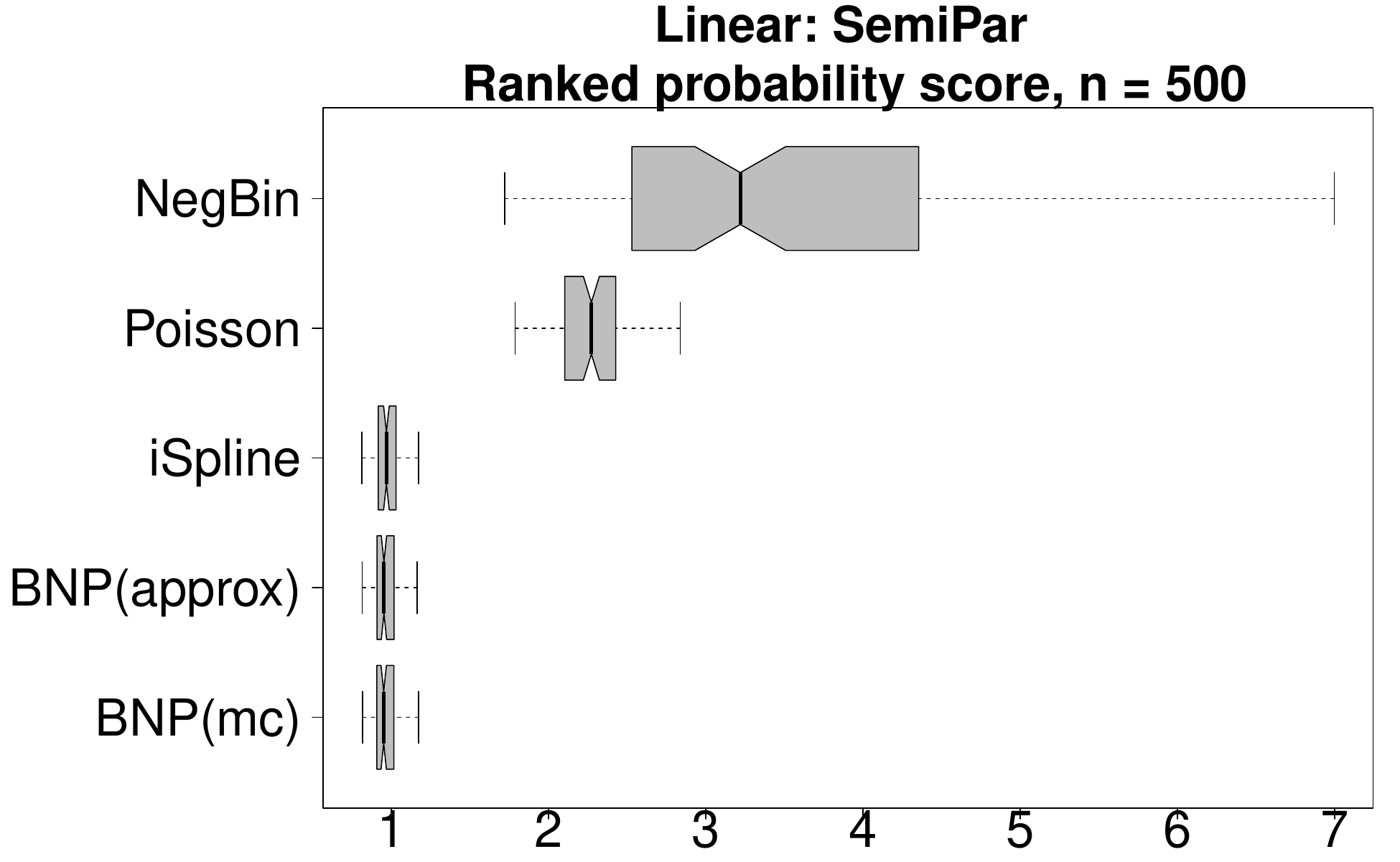}
\includegraphics[width=.45\textwidth]{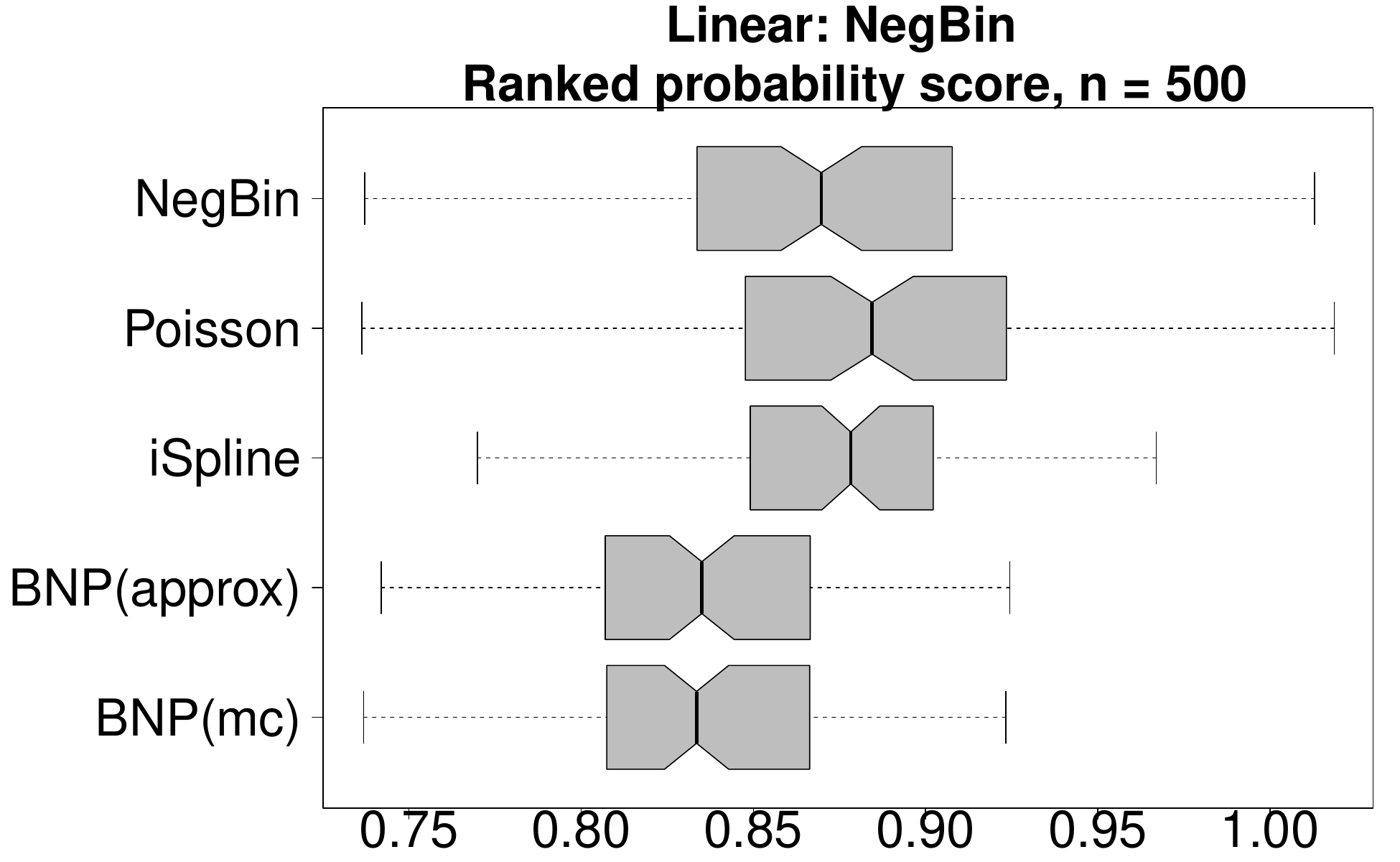}
\includegraphics[width=.45\textwidth]{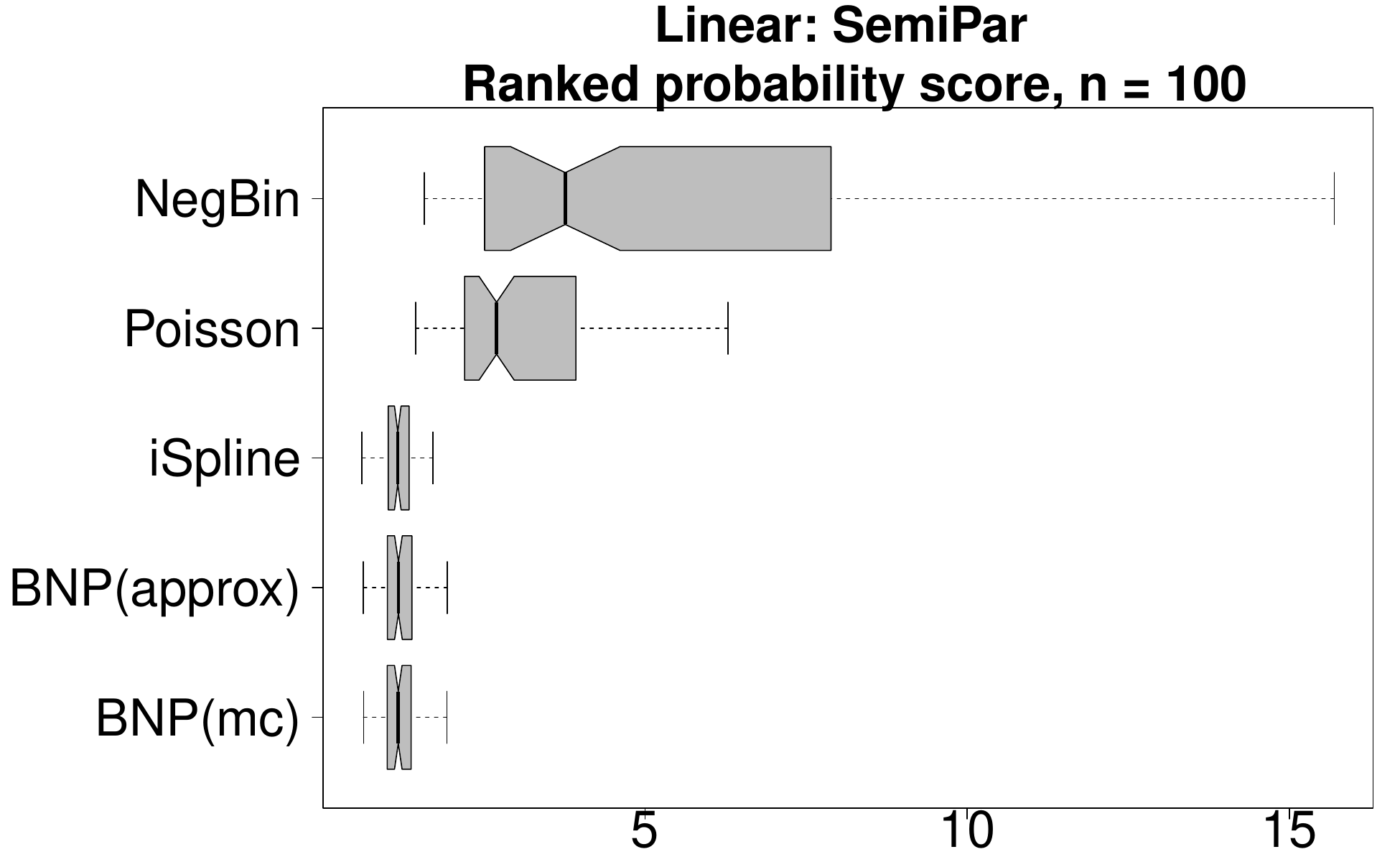}
\includegraphics[width=.45\textwidth]{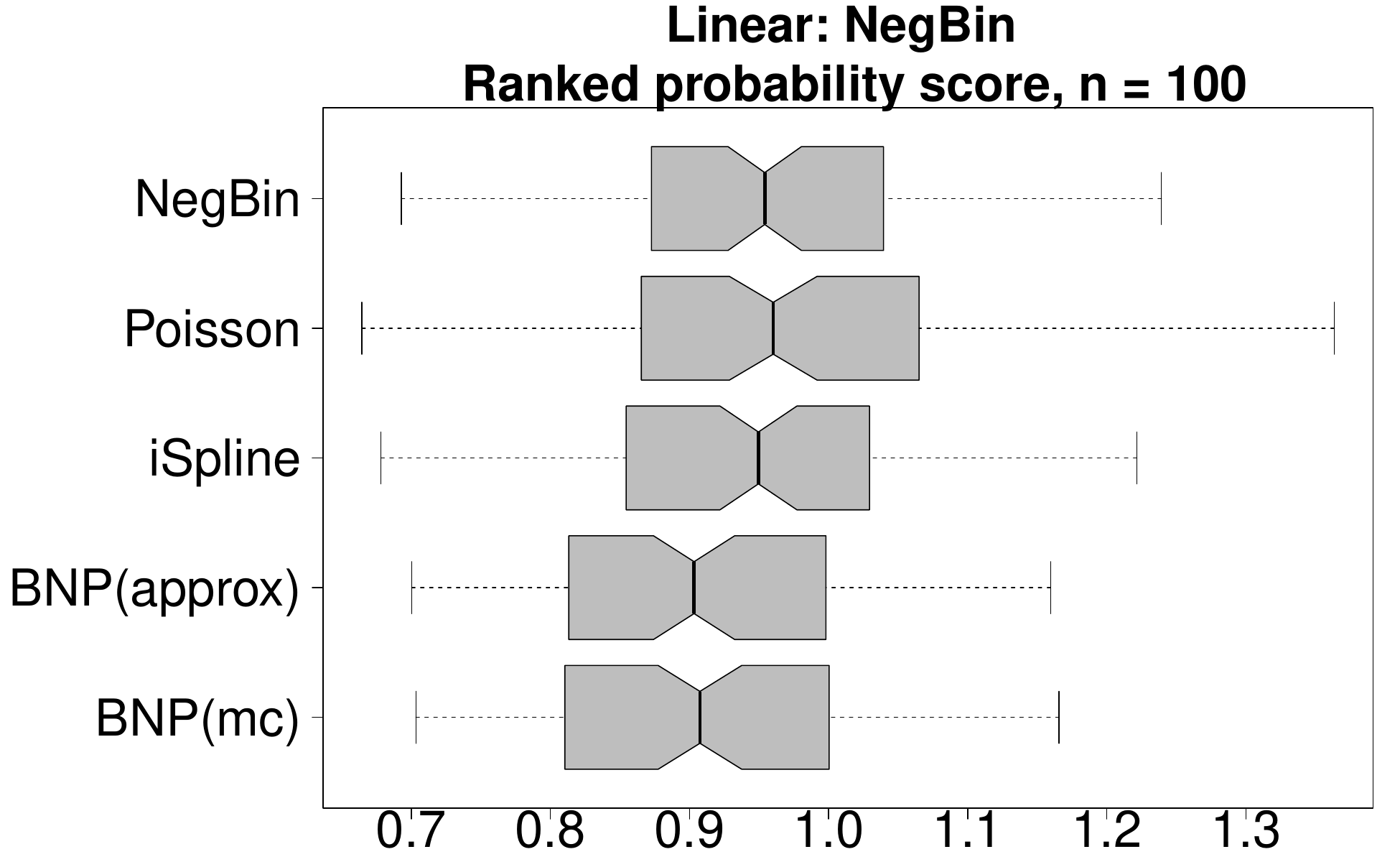}
\caption{\small Ranked probability scores for \texttt{SemiPar} (left) and \texttt{NegBin} (right) data, $n \in \{100,500\}$ (negatively oriented). Nonoverlapping notches indicate significant differences between medians. The proposed BNP approaches are much more accurate for \texttt{SemiPar} data yet still competitive for \texttt{NegBin} data.
\label{fig:rps}}
\end{figure}

Next, we evaluate 90\% prediction  intervals from each model in Figure~\ref{fig:miw}, which reports the mean interval widths and the empirical coverages.  Both \texttt{BNP(mc)} and \texttt{BNP(approx)} provide narrow interval estimates that achieve the nominal coverage. The Poisson model performs quite well for the \texttt{NegBin} design, but fails to maintain nominal coverage for the \texttt{SemiPar} data. Again, the proposed BNP approaches offer moderate to large gains relative to the \texttt{iSpline} alternative. 

\begin{figure}[h]
\centering
\includegraphics[width=.45\textwidth]{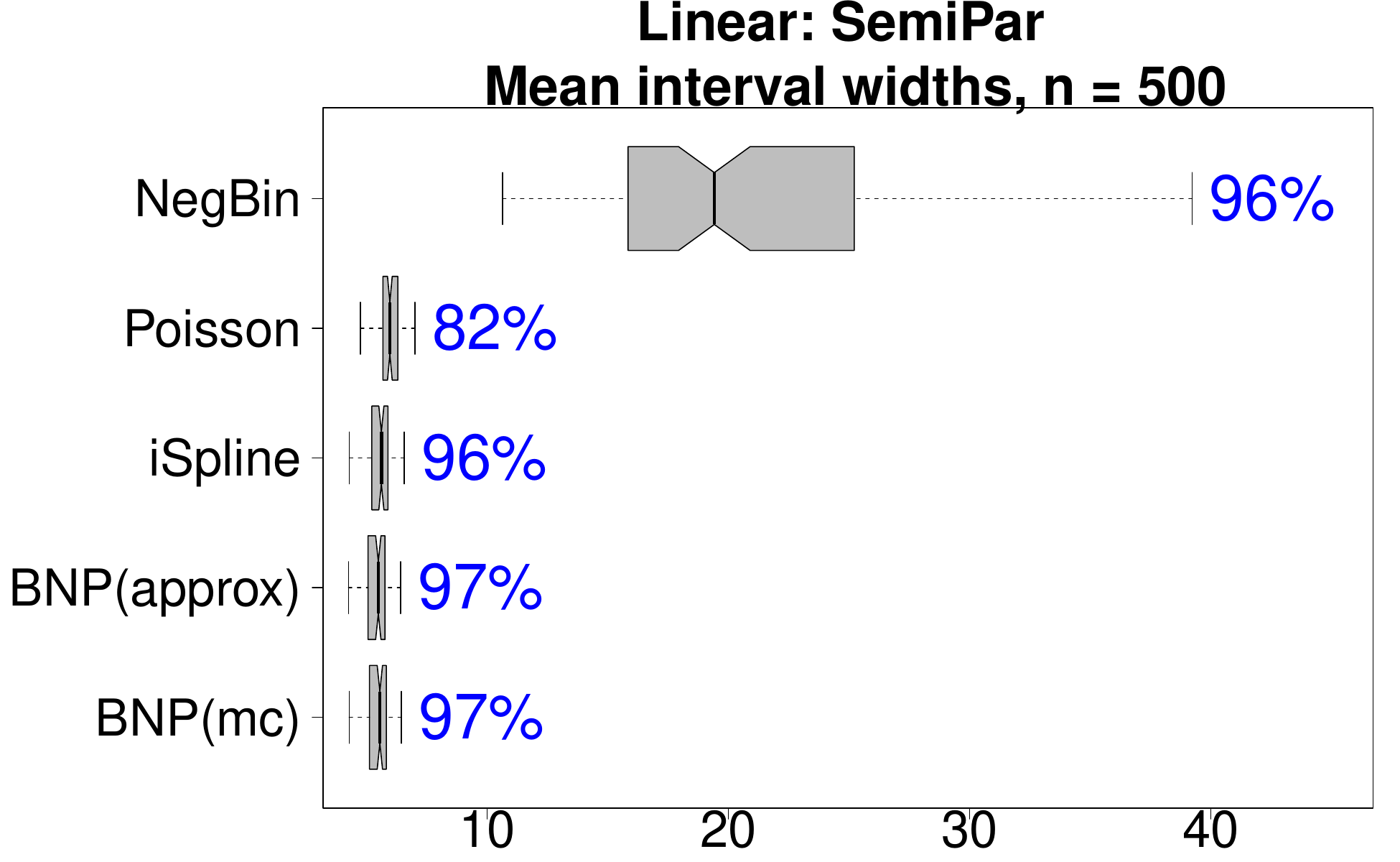}
\includegraphics[width=.45\textwidth]{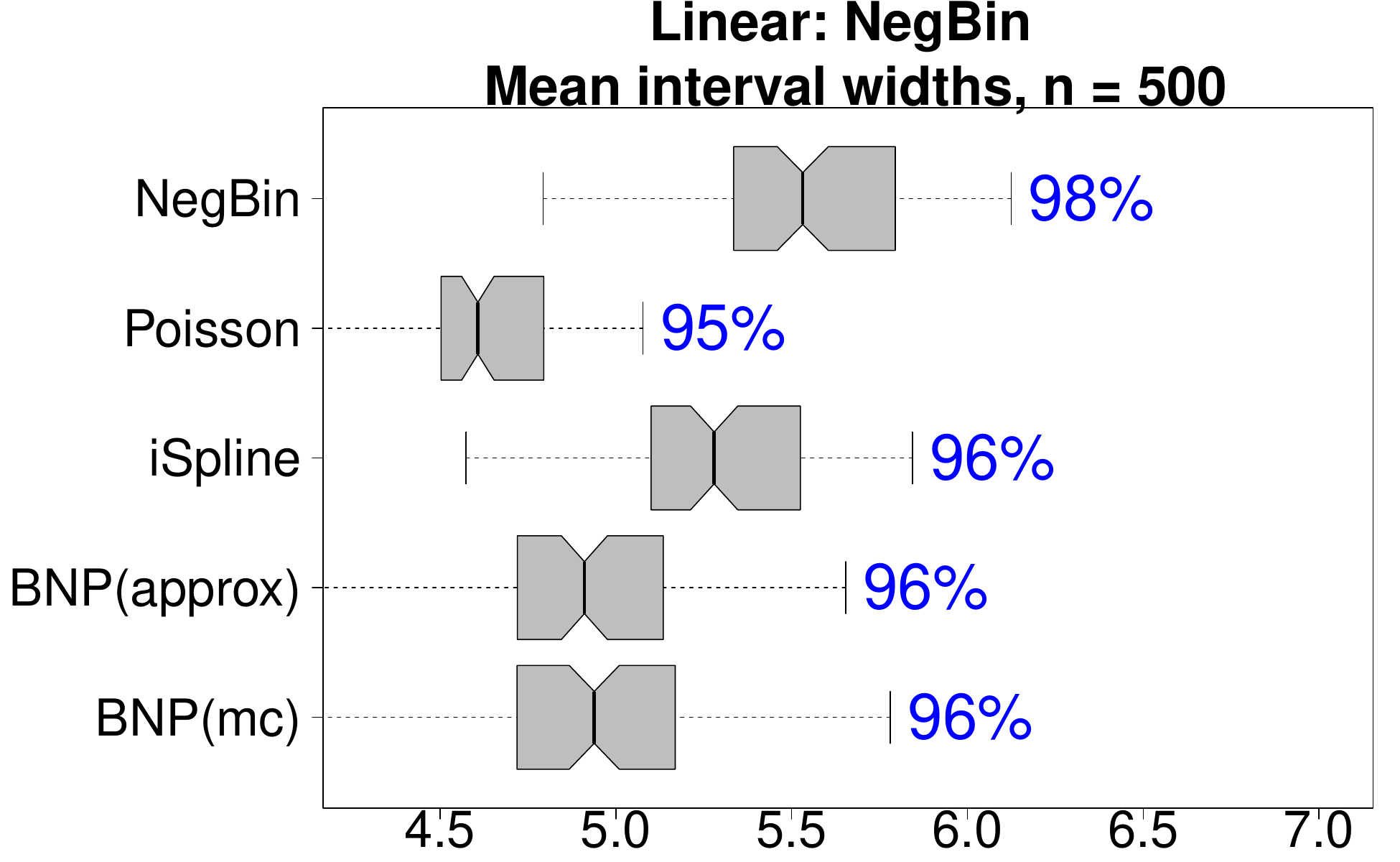}
\includegraphics[width=.45\textwidth]{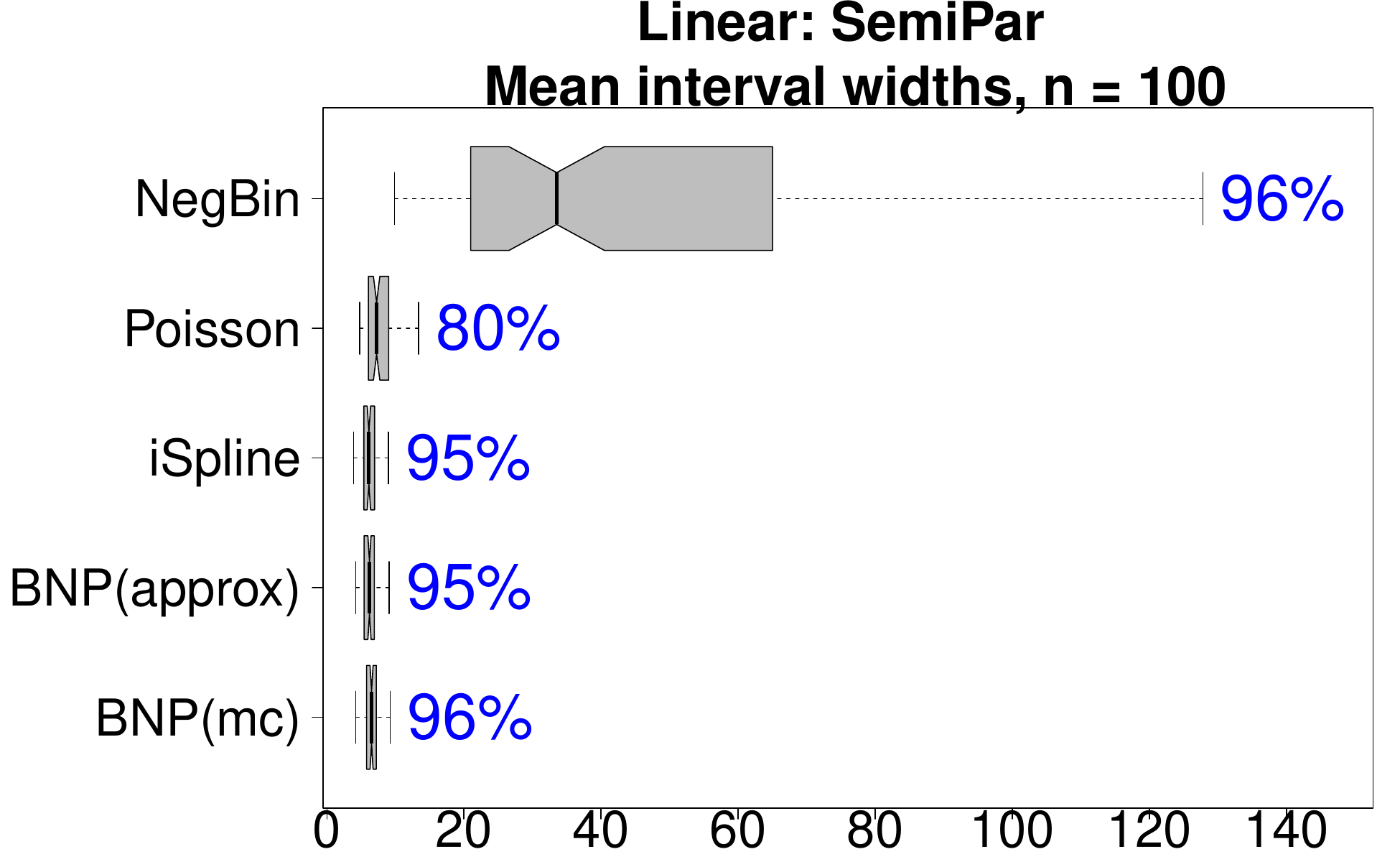}
\includegraphics[width=.45\textwidth]{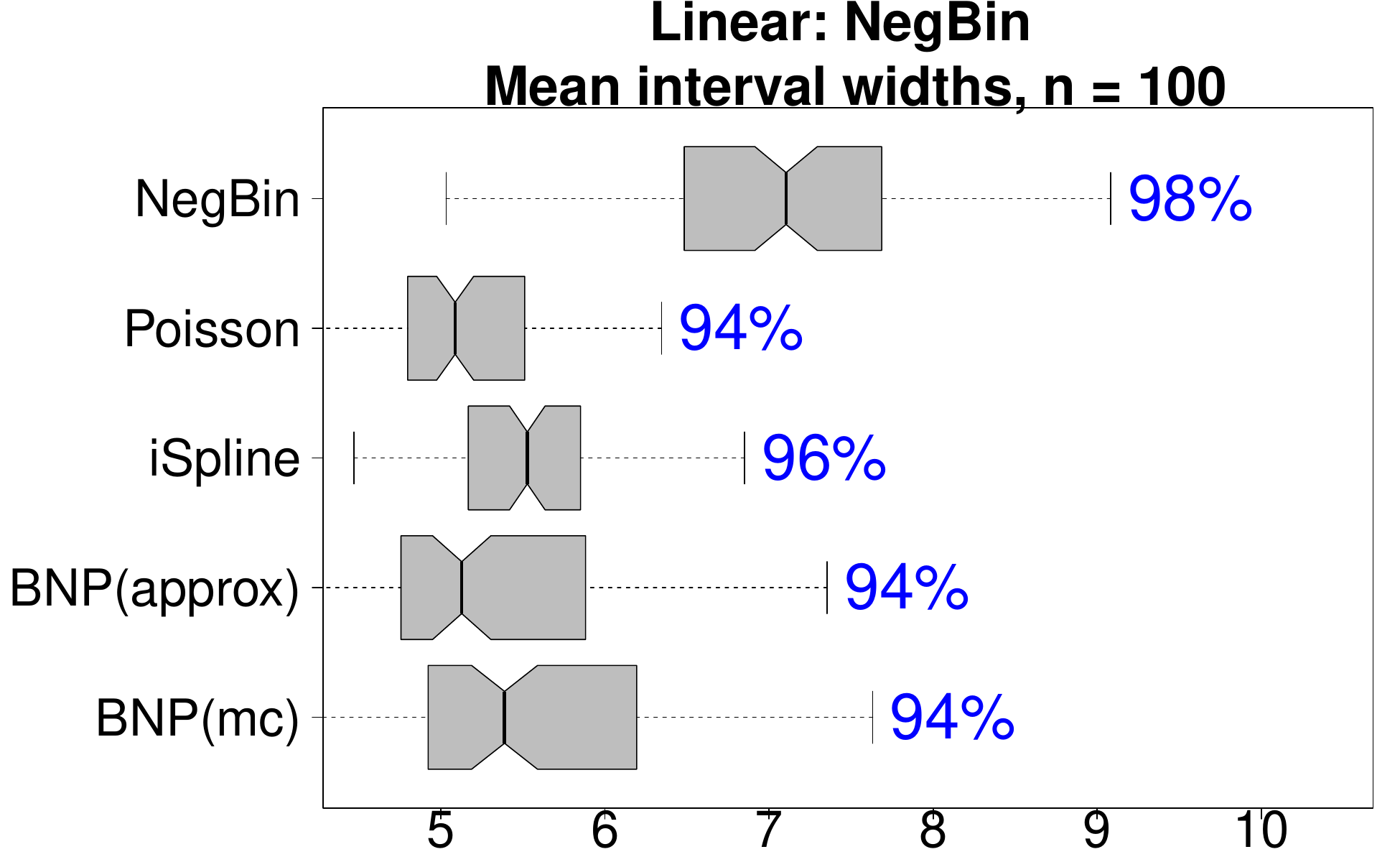}
\caption{\small Mean 90\% prediction interval widths (boxplots) with empirical coverage (annotations) on testing data with
\texttt{SemiPar} (left) and \texttt{NegBin} (right) data, $n \in \{100,500\}$ (negatively oriented). Nonoverlapping notches indicate significant differences between medians. 
The proposed BNP approaches provide narrow intervals with nominal coverage, with notable gains over the I-spline alternative under \texttt{NegBin}.
\label{fig:miw}}
\end{figure}

\subsection{Sparse means estimation for synthetic rounded data}\label{sims-sparse}
Lastly, we evaluate the selection capabilities of the sparse means model under rounding. The goal is to assess (i) whether an omission of the rounding operator impacts selection ability and (ii) whether the transformation is helpful for this task. Synthetic rounded data are generated as $y_i = \lfloor z_i \rfloor$ with $z_i = \theta_i + \epsilon_i$ and $\epsilon_i  \sim N(0, 1)$, where $\theta_i = \mu \gamma_i$ and $\mu$ denotes the signal strength with variable inclusion indicator $\gamma_i \in \{0,1\}$. We use a small proportion of signals, $\#\{\gamma_i = 1\}/n =0.10$, and a small signal $\mu = 2$. 
For all competing methods, we apply the same sparsity prior \eqref{ssmod} accompanied by  $\pi \sim \mbox{Beta}(1, 1)$  and 
$\sqrt{\psi} \sim \mbox{Uniform}(0, \sqrt{n})$.

We include two variations of the proposed approach, both with the rounding operator from Example~\ref{ex-round}: one uses the identity transformation (as in the data-generating process) and the other uses a modification of the BNP approximation  \eqref{trans-cdf-approx} with $\hat F_Z(t) = \Phi(t; 0, \sigma^2)$, which approximates  
 $F_Z(t) = n^{-1} \sum_{i=1}^n \Phi(t; 0, \sigma^2\psi \gamma_i + \sigma^2)$ when 
$\gamma_i = 0$ for nearly all $i$. 
In addition, we include a Gaussian model that omits the rounding operation entirely. To maintain focus on the modeling specifications, we apply identical Gibbs sampling strategies as outlined in Section~\ref{sec-sp}; the Gaussian case simply replaces \eqref{odds} with the relevant Gaussian likelihoods. The additional Gibbs sampling step for all models is  
$[\psi^{-1} \mid y, \theta, \gamma] \sim \mbox{Gamma}\{\#\{\gamma_i = 1\}/2 - 1/2, \sum_{i: \gamma_i=1} \theta_i^2/(2 \sigma^2)\}$ truncated to $[1/n, \infty)$. 
Lastly, we fix $\sigma$ at the MLE from a \texttt{kmeans} model with two clusters. 

The primary inferential target is $p(\gamma_i = 1 \mid y)$, which we use to select the nonzero means. The  receiving operator characteristic (ROC) curves for this selection mechanism are in Figure~\ref{fig:roc-auc}. Perhaps most surprising, the proposed \texttt{BNP(approx)} provides the best selection capabilities by a wide margin. Comparing the identity models with and without rounding, we find that the rounding operator offers only minor advantages for selection. Nonetheless, the best performer is from the class of models \eqref{tr}--\eqref{mod}, which is a coherent data-generating process for the discrete (rounded) data.

\begin{figure}[h]
\centering
\includegraphics[width=.49\textwidth]{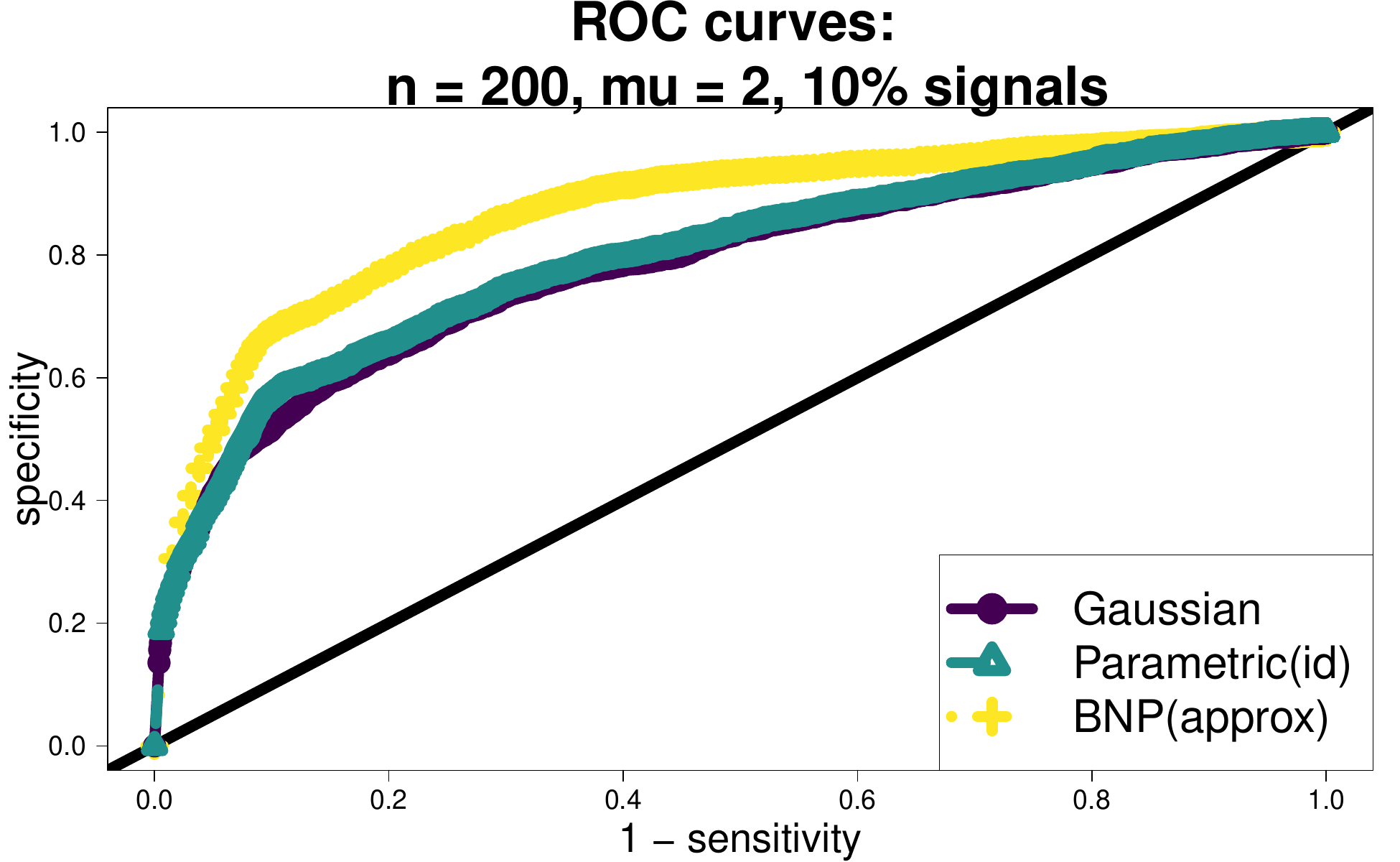}
\includegraphics[width=.49\textwidth]{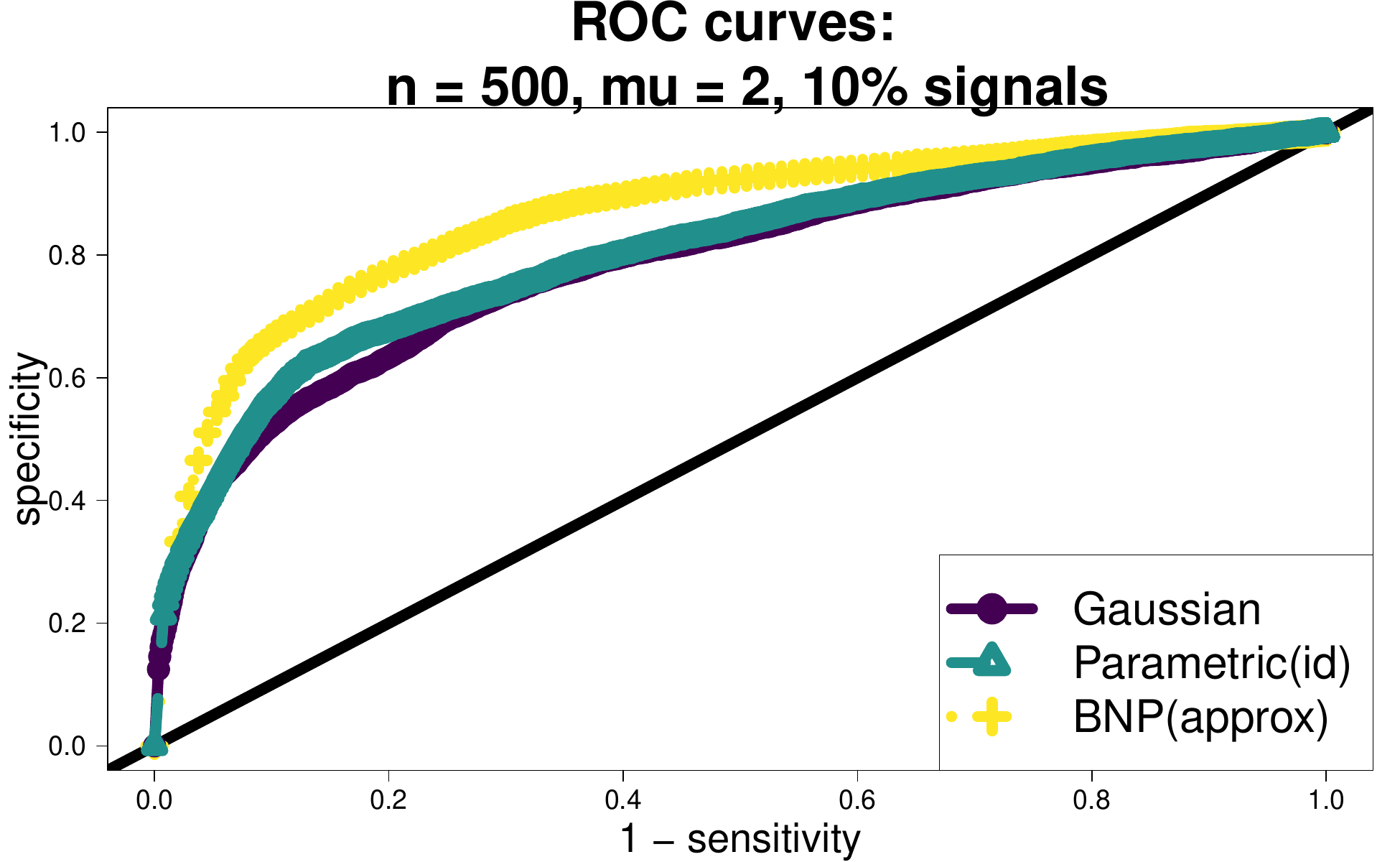}

\caption{\small ROC curves for $n\in\{200, 500\}$ for sparse means selection with rounded data. The rounding operator offers minor advantages for selection, while the proposed semiparametric approach with the approximation \eqref{trans-cdf-approx} dominates both competitors. The results persist for 5\% signals.
\label{fig:roc-auc}}
\end{figure}

\section{Application to self-reported mental health data}\label{app}
We apply the proposed methodology to self-reported mental health data from the National Health and Nutrition Examination Survey (NHANES). This data set 
contains responses $y$ to the question  ``For how many days during the past 30 days was your mental health not good?" and includes key demographic, socioeconomic, behavioral, and health-related covariates (see Appendix~\ref{sec-app-add}). The goal is to conduct a joint analysis of these important factors to determine their associations with self-reported mental health. However, these data present significant distributional challenges  (Figure~\ref{fig:dmhng}):  zero-inflation (57\% for $y=0$), boundedness and endpoint inflation (5.7\% for $y = y_{max} = 30$), and heaping (e.g., 3.3\% for $y=10$, but 0.4\% for $y=9$ and 0\% for $y=11$). 

We implement the proposed semiparametric model \eqref{tr}--\eqref{lm} with a  $g$-prior ($\psi = n$), including both the fully BNP model for $g$ and the approximation in \eqref{trans-cdf-approx}. The continuous covariates are centered and scaled prior to model-fitting.  Given the size of the data set ($n=2311$, $p = 24$), we use Algorithm~\ref{alg:joint} with the Gibbs sampler in Algorithm~\ref{alg:gibbs}. The samplers are efficient: for 6000 MCMC iterations (discarding 1000 as a burn-in), the sampler with the fully BNP model for $g$ required 5 minutes, while the sampler with approximation \eqref{trans-cdf-approx} required only 42 seconds  (using \texttt{R} on a MacBook Pro, 2.8 GHz Intel Core i7); traceplots suggested convergence and large effective sample sizes for $\theta$ (median ESS of $\{\theta_j\}_{j=1}^p$ exceeded 3000). In addition, we include a Negative Binomial regression model with the default specifications from \texttt{stan\_glm} (with \texttt{family = neg\_binomial\_2}) in the \texttt{R} package \texttt{rstanarm}. 

Posterior predictive diagnostics for these models are presented in Figure~\ref{fig:ecdf}. The posterior predictive functional is the ECDF of the data, which assesses whether these models can capture the challenging marginal distributional features in $y$. The fully BNP approach offers excellent performance, and suitably describes the heaping behavior and endpoint inflation. As expected, the approximation \eqref{trans-cdf-approx} conveys less uncertainty since the transformation is fixed at a point estimate, yet nonetheless maintains the same success as the fully BNP model. By comparison, the Negative Binomial model exhibits clear bias and fails to account for heaping.  We emphasize that the large jumps in the ECDF are \emph{not} simply due to the discreteness in the data, nor are they an artifact of the ECDF functional: these effects illustrate the heaping in the data, and close inspection of Figure~\ref{fig:ecdf} reveals their substantial size compared to the other much smaller increments.

\begin{figure}[h]
\centering
\includegraphics[width=.32\textwidth]{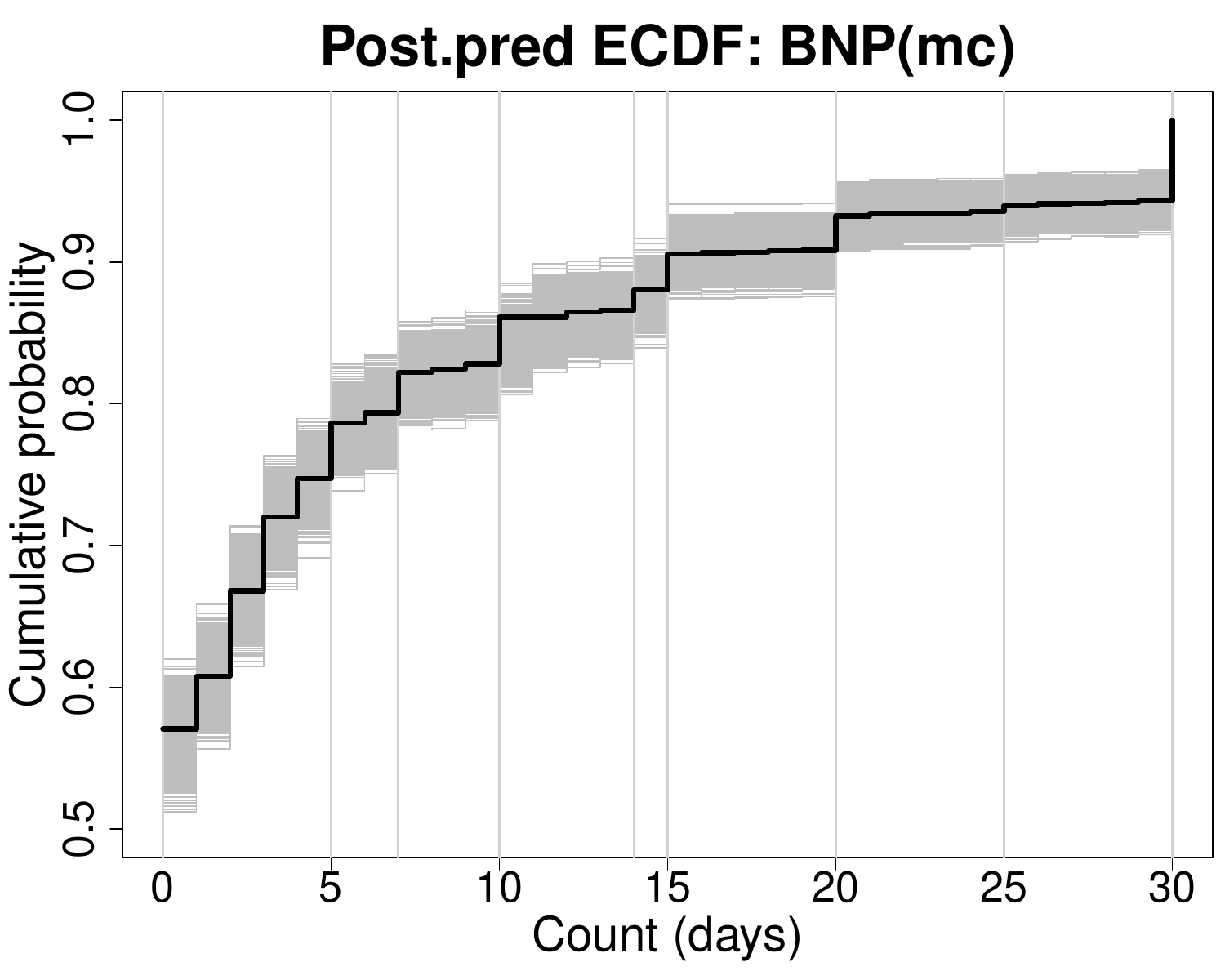}
\includegraphics[width=.32\textwidth]{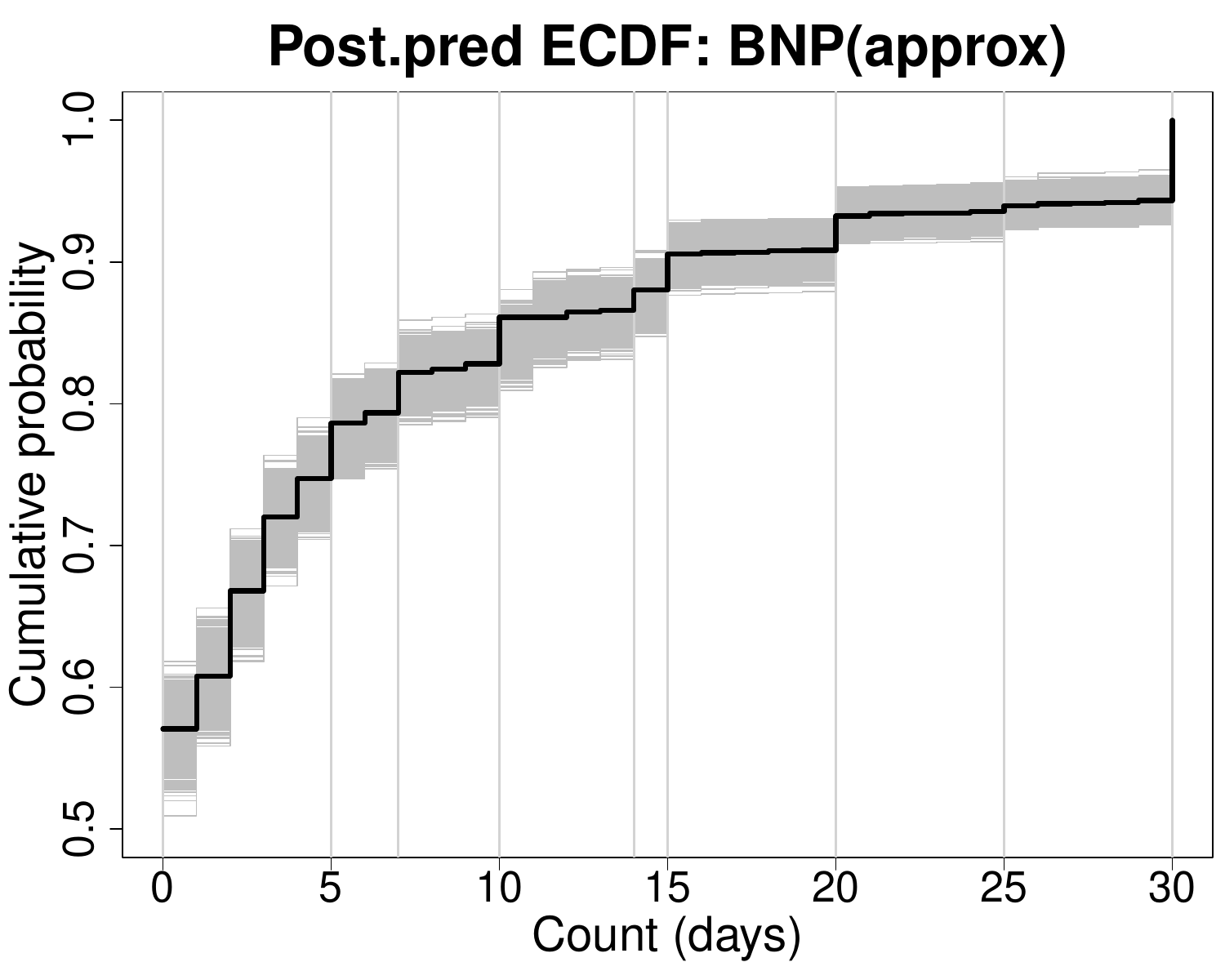}
\includegraphics[width=.32\textwidth]{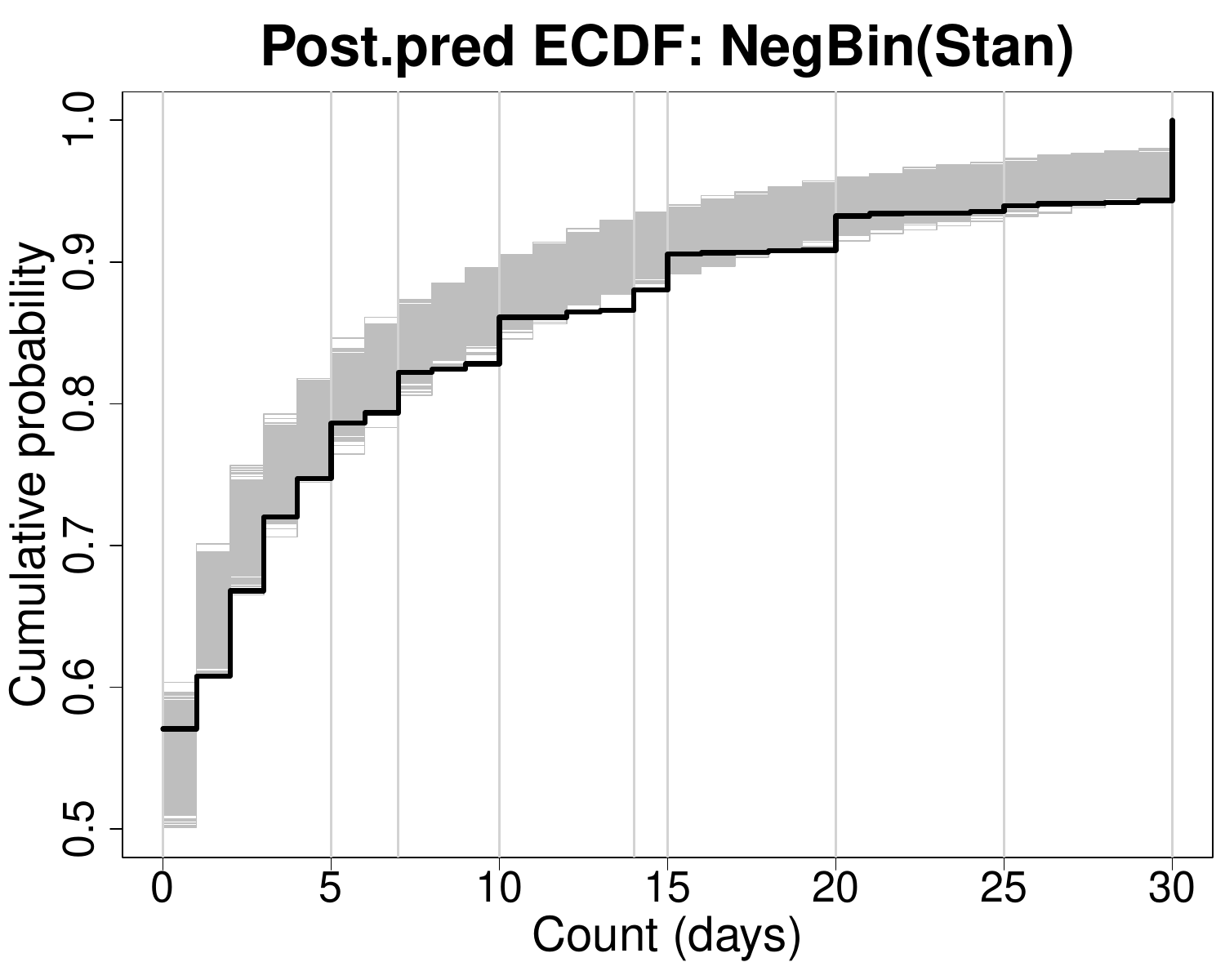}
\caption{\small Posterior predictive draws of the empirical CDF (gray) with the observed value (black) for the proposed BNP model (left), the proposed approximation \eqref{trans-cdf-approx} (center), and 
negative binomial regression (right). The proposed approaches suitably capture the key features in the self-reported mental health data, especially zero-inflation, heaping (vertical lines), and endpoint inflation, while the negative binomial model is biased and fails to account for heaping.
\label{fig:ecdf}}
\end{figure}

Finally, we present the posterior expectations and 90\% credible intervals for the regression coefficients in Figure~\ref{fig:coef}. The fully BNP version and the approximation \eqref{trans-cdf-approx} produce nearly identical results: more days with poor mental health are self-reported by females, low income individuals, drug (alcohol, nicotine, marijuana, and hard drugs) users, and individuals with diabetes or high blood pressure, while the opposite effect is observed for non-White, wealthier, uninsured, and older individuals. For clarity, this figure excludes the intercept and the primary sampling unit, which was included to help adjust for possible geographic clustering in the survey responses.

\begin{figure}[h]
\centering
\includegraphics[width=.8\textwidth]{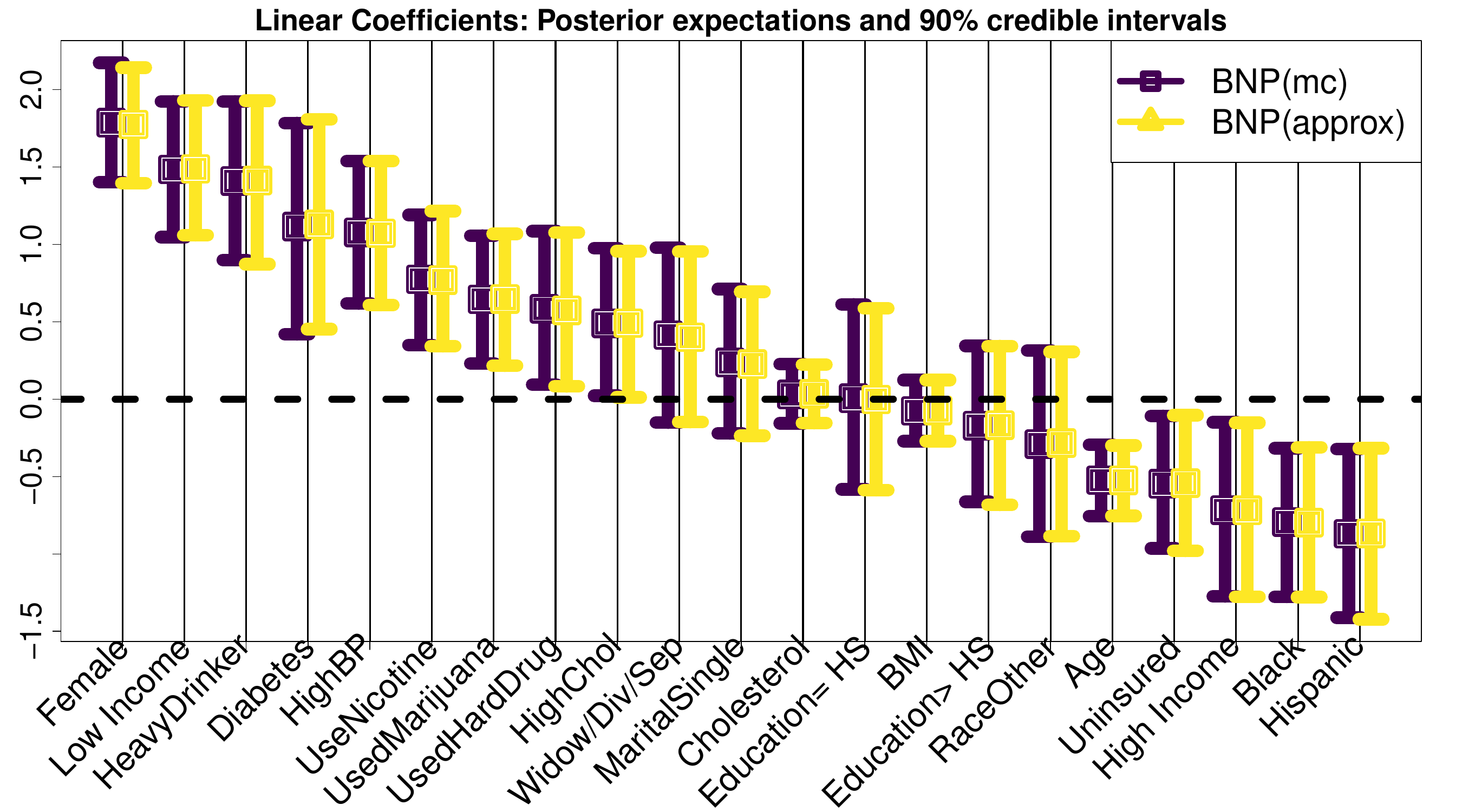}
\caption{\small Posterior expectations and 90\% credible intervals for $\theta$ under model \eqref{tr}--\eqref{lm} for the self-reported mental health  data. More poor mental health days are self-reported by females, low income individuals, drug (alcohol, nicotine, marijuana, and hard drugs) users, and individuals with diabetes or high blood pressure; the opposite effect is observed for non-White, wealthier, uninsured, and older individuals. 
\label{fig:coef}}
\end{figure}

\section{Conclusion}\label{disc}


Operating  within a broad class of semiparametric discrete data regression models, we introduced a novel collection of priors and algorithms that enabled consistent posterior inference and efficient Monte Carlo sampling. This framework incorporated a rounding operator to ensure the correct discrete support of the data-generating process along with a transformation to encourage greater representational ability. Most uniquely, the transformation was modeled nonparametrically using a smoothed Bayesian bootstrap procedure, which delivered posterior consistency under general conditions and admitted a straightforward and efficient Monte Carlo sampler for the marginal posterior. When paired with a latent linear regression model, the model featured conjugate priors, (pseudo) closed-form posteriors, and Monte Carlo sampling for posterior and predictive inference. Conditional on the transformation, the posterior distribution for the regression coefficients was strongly consistent---even under model misspecification. These methods demonstrated utility for linear regression, nonlinear regression via basis expansions, and selection for count and rounded data. The simulation studies highlighted significant advantages in computational efficiency, predictive performance, coefficient estimation, and  selection relative  to state-of-the-art competitors. These methods were applied to a self-reported mental health data and produced exceptional model performance, computational scalability, and insightful results. 

The central ideas in the proposed semiparametric modeling and computing framework are ripe for generalizations, including for continuous data, other latent data models, or alternative marginal (parametric or nonparametric) models for $F_Y$. Thus, these tools may be  useful for consistent and efficient Bayesian semiparametric inference in a wide array of settings.

\subsection*{Acknowledgements}
This material is based upon work supported by the National Science Foundation (SES-2214726).

\appendix

\section{Proofs and technical details}\label{sec-proofs}




\begin{proof}(Theorem~\ref{thm-Fz-con})
For weak convergence, it is sufficient to show that $\int u(z) \ d \tilde F_Z(z) \to \int u(z) \ d F_{Z,0}(z)$ as $n\to \infty$ for any continuous and bounded function $u$. By construction, $ \int u(z) \ d \tilde F_Z(z) = \int \int u(z) \ d \tilde F_{Z \mid X=x}(z) \  d \tilde F_X(x) = \int v(x)  \  d \tilde F_X(x)$, where $v(x) = \int u(z) \ d \tilde F_{Z \mid X=x}(z)$ is  continuous by assumption and bounded due to the boundedness of $u$. Because the BB posterior for $\tilde F_X$ converges weakly to $F_{X,0}$ \citep{Ghosal2017}, we equivalently have $\int v(x) \ d \tilde F_X(x) \to \int v(x) \ d F_{X,0}(x) = \int u(z) \ d F_{Z,0}(z)$ as $n\to \infty$, which proves the result. 
\end{proof}


To prove Theorem~\ref{thm-joint-con}, we first note several existing results.

\begin{theorem}[Pólya]\label{Pólya}
Let $\F(\R)$ be the set of all probability measures over $\R$.
Any sequence of distribution functions that converge weakly to a CDF also converges in the Kolmogorov-Smirnov metric, $d_{KS}(P, Q) = \sup_{x \in \R}\vert P((-\infty, x]) - Q((-\infty, x])\vert$ for $P,Q \in \F(\R)$.
As a consequence, any $d_{KS}$ neighborhood of an atomless $P_0 \in \F(\R)$ contains a weak neighborhood of $P_0$. 
\end{theorem}
Note that Pólya's theorem does not hold for discrete probability measures.  If a sequence of discrete distribution functions $F_n$ converges to $F_0$ weakly, that only implies pointwise, but not uniform, convergence of $F_n$ to $F$.

\begin{theorem}[Dini's second theorem (Pólya)] \label{dini}
Let $F_n:\R\to\R$ be a sequence of continuous and monotonically increasing functions. Then pointwise convergence of $F_n$ on an interval $[a, b]$ to a continuous function $F$ implies uniform convergence. 
\end{theorem}
The following lemma reveals the implication of the weak convergence on the generalized inverse function of CDFs that is useful in the proof of Theorem~\ref{thm-joint-con}.
\begin{lemma}  \label{inverse-convergence}
For any sequence $F_n$ of continuous and monotonically increasing CDFs that converge weakly to a continuous CDF $F$, the sequence of inverse CDFs $F_n^{-1}$ converges uniformly to $F^{-1}$ on every compact subset of $(0,1)$. 
\end{lemma}
\begin{proof}
Suppose $F_n$ is a sequence of continuous, monotonically increasing CDFs that converges weakly to a continuous function $F$.
By Lemma 21.2 of \cite{VanderVaart2000}, the weak convergence of $F_n$  to $F$ is equivalent to the weak convergence $F_n^{-1}$ to $F^{-1}$, where $F^{-1}\!:(0,1)\to\mathbb{R}$ with $F^{-1}(p) = \inf\{x: F(x) \ge p\}$ and similarly for $F_n^{-1}$. For each positive integer $n$,  $F_n^{-1}$ is continuous and monotonically increasing.  Consequently, $F_n^{-1}$ converges weakly to $F^{-1}$ if and only if $F_n^{-1}$ converges pointwise to $F^{-1}$. By applying Dini's second theorem, we obtain that $F_n^{-1}$ converges uniformly to $F^{-1}$ on every compact subset of $(0,1)$.  
\end{proof}

\begin{proof}(Theorem~\ref{thm-joint-con})
Define $\F(\R)$ as the set of CDFs of all probability distributions over $\R$ that have continuous CDFs and $\F(\Z)$ as the set of CDFs of all probability distributions over $\Z$.  The space $\F(\R)$ is complete and metrizable under   $\|\cdot\|_\infty$, where $\|F_1 - F_2\|_\infty = \sup_{x \in \R}\vert F_1(x) - F_2(x)\vert$ for $F_1, F_2\in \F(\R)$. 
The space $\F(\Z)$ is a subspace of $\M(\Z)$, i.e., the space of monotonically increasing functions over $\Z$, and both  $\F(\Z)$ and $\M(\Z)$ are complete and metrizable under the weak topology. Let $d_w$ denote one such metric.


Let $F_{Z,0}  \in \F(\R), F_{Y,0} \in \F(\Z)$ be the true distribution of $z, y$, respectively.   
For simplicity, consider $a_j = j$, so the transformation \eqref{trans-cdf} may be expressed as $
    \tilde g_0(j + 1) = F_{Z,0}^{-1}  \{F_{Y,0}(j)\}$ for $j \in \mathbb{Z}$. 
Define the mapping $\rho_n: (\F(\R), \|\cdot\|_\infty) \times (\F(\Z), d_w) \to ( \M, d_w)$ by $\rho_n(F_Z, F_Y) = F_Z^{-1} \{u_n F_Y + \ell_n\}$, where $u_n = n/(n+1)$ and $\ell_n = 1/(2n + 1)$, and let   $\tilde g_{0, n}: \Z \to \R$ by $\tilde g_{0, n}(j+1) =\rho_n(F_{Z,0},F_{Y, 0})(j+1)$. Since $F_{Z,0}$ is continuous and $u_n \to 1$, $\ell_n \to 0$, it follows that $ \tilde g_{0, n}(j) $ converges to $\tilde g_0(j) = F_{Z,0}^{-1}\{ F_{Y,0} (j)\}$ for every $j \in \Z$.  

For given $n$, we want to show that $\rho_n$ is a continuous mapping at $(F_{Z,0}, F_{Y, 0})$. Take a converging sequence $(F_{Z, n}, F_{Y,n}) \to (F_{Z,0}, F_{Y, 0})$ from the domain of $\rho_n$. The existence of such converging sequences is guaranteed because $\F(\R)$ and $\F(\Z)$ are complete under the assumed topology.  Each CDF $F_{Z, n}$ is continuous and monotonically increasing. 
Since $F_{Z,n}$ converges to $F_{Z,0}$ uniformly by construction, 
the sequence of inverses $F_{Z,n}^{-1}$ converges to $F_{Z,0}^{-1}$ uniformly on $[ \ell_n, u_n]$ by Lemma \ref{inverse-convergence}. 
Then for any integer $j$, 
\begin{align*}
    \big\vert \rho_n(F_{Z, n}, F_{Y,n})(j + 1) - \tilde g_{0,n} (j+1)\big\vert 
    &= \big\vert F_{Z,n}^{-1} \{u_nF_{Y,n}(j) +  \ell_n\} - F_{Z,0}^{-1} \{u_n F_{Y,0}(j) +  \ell_n\} \big\vert  \\
    &\leq \big\vert F_{Z,n}^{-1} \{u_n F_{Y,n}(j) + \ell_n\} - F_{Z,0}^{-1} \{u_n F_{Y,n}(j ) + \ell_n\}\big\vert  + \\
    &\quad \quad \big\vert F_{Z,0}^{-1} \{u_n F_{Y,n}(j) + \ell_n\}- F_{Z,0}^{-1} \{u_n F_{Y,0}(j) + \ell_n\}\big\vert.
\end{align*}
The first term converges to zero by the uniform convergence of $F_{Z, n}^{-1}$. The second term converges to zero because $F_{Y,n}(j)$ converges to $F_{Y,0}(j)$ pointwise 
and $F_{Z,0}^{-1}$ is continuous.  Thus,  $\rho_n$ is continuous at  $(F_{Z,0}, F_{Y, 0})$.


As a result, for every open neighborhood $\U_w(\tilde g_{0,n})$, there exists a neighborhood $\U_{\infty}(F_{Z,0}) \times \U_w(F_{Y,0})$ for which the image under $\rho_n$ is contained in $\U_w(\tilde g_{0,n})$. 
By Pólya's theorem (Theorem \ref{Pólya}), $\U_{\infty}(F_{Z,0})$ contains a weak neighborhood $U_{w}(F_{Z,0})$. By Proposition $6.29$ of \cite{Ghosal2017}, the weak consistency of $\tilde F_Z$ (Theorem~\ref{thm-Fz-con}) and $\tilde F_Y$ \citep{Ghosal2017} implies the weak convergence of their product, i.e., $\Pi_n\{\U_{w}(F_{Z,0}) \times \U_w(F_{Y,0})\mid y\} \to 1$. Then,
the sequence of functions $\tilde g_{0,n}$ is eventually in every neighborhood $\U_w(\tilde g_0)$ because $\tilde g_{0,n}$ converges to $\tilde g_0$ pointwise. Therefore, for large enough $n$, $\Pi_n\{\U_w(\tilde g_0)\mid y\} \geq \Pi_n \{\U_w(\tilde g_{0,n})\mid y\}  \geq \Pi_n\{\U_{w}(F_Z) \times \U_w(F_Y) \mid y\}$ which implies that $\Pi_n\{\U_w(\tilde g_0) \mid y\} \to 1$. Hence, the posterior distribution $[\tilde g \mid y]$ is weakly consistent. 
\end{proof}

\begin{proof}(Theorem~\ref{thm-slct})
    Since Theorem~\ref{thm-slct0} implies that $p(\theta \mid y)$ is a selection distribution under model \eqref{tr}--\eqref{lm}, the joint Gaussian assumption for $(z, \theta)$ implies that $p(\theta \mid y)$ is selection normal. The remaining properties follow from \cite{ArellanoValle2006}. 
\end{proof}

\begin{proof}(Theorem~\ref{thm-mc})
    The result follows from the properties of the selection normal \citep{ArellanoValle2006} applied to model \eqref{tr} and the joint Gaussian distribution for $(z, \theta)$ via Theorem~\ref{thm-slct}.
\end{proof}

\begin{proof}(Lemma~\ref{star-lm}) The result follows by direct computation as a consequence of  Theorem~\ref{thm-slct} under the stated model \eqref{lm} and prior $ \theta \sim N( \mu_\theta,  \Sigma_\theta)$.
\end{proof}

\begin{proof}(Algorithm~\ref{alg:g-sim})
        Using Theorem~\ref{thm-mc}, it is sufficient to compute $ \Sigma_z$ and 
        $\Sigma_\theta -  \Sigma_{z\theta}'\Sigma_z^{-1}  \Sigma_{z\theta}$  under the linear model with the $g$-prior. By computing $\Sigma_z =\sigma^2\{ I_n + \psi  X ( X' X)^{-1}  X'\}$ and applying the Woodbury identity, it follows that  $ \Sigma_z^{-1} = \sigma^{-2}  \{ I_n-  \psi(1+\psi)^{-1} X ( X' X)^{-1} X'\}$, so $\Sigma_{z\theta}'\Sigma_z^{-1}  =\psi(1+\psi)^{-1} ( X' X)^{-1} X'$ and  $\Sigma_{z\theta}'\Sigma_z^{-1}\Sigma_{z\theta} = \sigma^2\psi^2(1+\psi)^{-1} ( X' X)^{-1}$. Observing that $\Sigma_\theta -  \Sigma_{z\theta}'\Sigma_z^{-1}  \Sigma_{z\theta} = \sigma^2 \psi(1+\psi)^{-1} (X'X)^{-1}$ shows the result. 
    \end{proof}

\begin{proof}(Lemma~\ref{lem:2})
    The selection normal prior is equivalently defined by $[ \theta \mid  z_0 \in \mathcal{C}_0]$ for $( z_0',  \theta')'$ jointly Gaussian with moments given in the prior. The posterior is constructed similarly: $[ \theta \mid  y] = [ \theta \mid  z_0 \in \mathcal{C}_0,  z \in g(\mathcal{A}_{ y})] = [ \theta \mid  z \in \mathcal{C}_0 \times g(\mathcal{A}_{ y})]$ where $ z_1 = ( z_0',  z')'$. It remains to identify the moments of $( z_1',  \theta')' = ( z_0',  z',  \theta')'$. For each individual block of $ z_0$, $ z$, and  $ \theta$ and the pairs ($ z_0,  \theta$) and ($ z,  \theta$), the moments are provided by the selection normal prior or the selection normal posterior in Lemma~\ref{star-lm}. All that remains is the cross-covariance $\mbox{Cov}( z_0,  z) = \mbox{Cov}( z_0,  X  \theta +  \epsilon) =   \Sigma_{z_0\theta} X'$. 
\end{proof}

\begin{proof}(Algorithm~\ref{alg:nl-pred})
    Computing the selection normal terms in the usual notation, we have $\Sigma_z = \sigma^2 (\psi X X' + I_n)$, $\Sigma_\theta =\sigma^2 (\psi \tilde X \tilde X' + I_{\tilde n}) $, and 
    $\Sigma_{z\theta} = \psi \sigma^2 X \tilde X'$. Let $D_\psi = (\psi^{-1} I_p + X'X)^{-1} =  \mbox{diag}\{\psi /(1 + \psi d_j)\}$, so $\Sigma_z^{-1}  = \sigma^{-2}(I_n - X D_\psi X')$ by the Woodbury identity. Next, we derive  $\Sigma_{z\theta}'\Sigma_z^{-1}  = \psi \tilde X \{X' - (X'X) D_\psi X'\}$ and  $\Sigma_{z\theta}'\Sigma_z^{-1}\Sigma_{z\theta} = \psi^2 \sigma^2 \tilde X \{(X'X) - (X'X)D_\psi (X'X)\} \tilde X'$, so 
    $\Sigma_\theta -  \Sigma_{z\theta}'\Sigma_z^{-1}  \Sigma_{z\theta} = \psi \sigma^2  \tilde X [ I_p -  \psi \{(X'X) - (X'X) D_\psi (X'X) \}] \tilde X' +   \sigma^2  I_{\tilde n}$. Since $X'X = \mbox{diag}(d_j)$, the term in the square brackets is diagonal with elements $1 - \psi \{d_j - d_j^2 \psi /(1  + \psi d_j)\} = 1/(1+  \psi d_j)$. Applying Algorithm~\ref{alg:g-pred}  shows the result. 
    \end{proof}


\begin{proof}(Theorem~\ref{thm-theta-con})
Consider the sequence of functions $
    f_n(\theta) = -n^{-1} \sum_{i=1}^n \log p_\theta(y_i \mid x_i)$ and let $f(\theta) = - \mathbb{E}_{P_0} \log p_\theta(y \mid x)$, where $p_\theta(y \mid x) = \bar \Phi\{g(\mathcal{A}_y); x'\theta, \sigma^2\}$ is the likelihood  and $\mathcal{A}_y$ is determined by $h$ in \eqref{tr}. 
The proof requires establishing the following conditions from Theorem 3 in \cite{Miller2021}, and in particular that they hold almost surely $[P_0]$:
\begin{enumerate}
    \item  $\Pi\{\mathcal{U}_\epsilon(\theta_0)\} > 0$ for all $\epsilon >0$, where   $\mathcal{U}_\epsilon(\theta_0) = \{\theta \in \Theta: \Vert \theta - \theta_0 \Vert_2 < \epsilon\}$; 
    \item $f_n  \to f$ pointwise on $\Theta$;
    
    \item $f_n$ is convex for each $n$;
    \item $\Theta \subseteq \mathbb{R}^p$;
    \item $\theta_0 \in \mbox{int}(\Theta)$; and \item $f(\theta) > f(\theta_0) $ for all $\theta \in \Theta\backslash \{\theta_0\}.$
\end{enumerate} 
Conditions 1, 4, and 5 are satisfied directly by the theorem assumptions. Condition 2 is satisfied because $n^{-1} \sum_{i=1}^n \log p_\theta(y_i \mid x_i) \to  \mathbb{E}_{P_0} \log p_\theta(y \mid x)$ almost surely [$P_0$] by the strong law of large numbers since $\{(x_i, y_i)\}_{i=1}^n \stackrel{iid}{\sim} P_0$, provided that $\mathbb{E}_{P_0} \vert \log p_\theta(y \mid x)\vert < \infty$. This condition is satisfied by the theorem statement and discussed further below. Condition 3 follows because $f_n(\theta) = -n^{-1} \sum_{i=1}^n \log \bar \Phi\{g(\mathcal{A}_y); x'\theta, \sigma^2\}$ and the Gaussian CDF $\Phi$ is log-concave; this is also demonstrated in \cite{Kowal2021b}. Finally, we establish condition 6: by definition, $\theta_0$ satisfies 
$\theta_0 = \arg\min_{\theta \in \Theta} KL(P_0, P_\theta) 
    = \arg\min_{\theta \in \Theta} -\mathbb{E}_{P_0} \log p_\theta(y \mid x) 
    = \arg\min_{\theta \in \Theta} f(\theta)$, and the uniqueness of $\theta_0$ implies that $f(\theta) > f(\theta_0)$ whenever $\theta \ne \theta_0$. By Theorem 3 of \cite{Miller2021}, the posterior distribution satisfies $\Pi_n\{\mathcal{U}_\epsilon(\theta_0)\} \to 1$ almost surely $[P_0]$ for all $\epsilon > 0$, i.e., the posterior is strongly consistent at $\theta_0$. 
\end{proof}


The conditions on the prior are rather mild and are discussed in the main paper. The conditions on the likelihood require further elaboration. Specifically, we require  existence of $\mathbb{E}_{P_0} \log p_\theta(y \mid x)$, where 
$p_\theta(y \mid x) = \bar \Phi\{g(\mathcal{A}_y); x'\theta, \sigma^2\}$. Since $p_\theta < 1$ by construction, $\log p_\theta < 0$ and it remains to establish conditions under which $-\mathbb{E}_{P_0} \log p_\theta(y \mid x) < \infty$. Focusing on the likelihood term, we have
\begin{align*}
-\log p_\theta(y \mid x) &= - \log \bar \Phi\{g(\mathcal{A}_y); x'\theta, \sigma^2\} = -\log  \int_{g(\mathcal{A}_y)} \phi(z; x'\theta, \sigma^2) \ dz \\
&\le  \int_{g(\mathcal{A}_y)} -\log \phi(z; x'\theta, \sigma^2) \ dz \\
& \le \frac{1}{\sigma^2} \int_{g(\mathcal{A}_y)} (z - x'\theta)^2 \ dz. 
\end{align*}
This latter step requires that $\vert g(\mathcal{A}_y)\vert < \infty$. This requirement is \emph{not} satisfied at the boundaries for finite support $\{0,\ldots, y_{max}\}$, i.e., when $\mathcal{A}_y = (a_y, a_{y+1}]$, we have $g(\mathcal{A}_0) = (-\infty, g(a_1)]$ and $g(\mathcal{A}_{y_{max}}) = (g(a_{y_{max}}), \infty)$; that case is discussed below. For now, assume that $\vert g(\mathcal{A}_y)\vert < \infty$ for all $y \in \mathbb{Z}$. Let $\Delta \psi(y) = \psi(a_{y+1}) - \psi(a_y)$ for any function $\psi$. Then we further simplify to
\[
-\log p_\theta(y \mid x)  \le \frac{1}{\sigma^2}\{
\Delta g^3(y)/3 - (x'\theta)\Delta g^2(y) + \theta' (xx')\theta  \Delta g(y) \}
\]
Thus, we require that this term is finite in expectation $[P_0]$. While many options exist, one simple approach is to require finite expectation for each term, which is further simplified without the $\Delta$ operator: $\mathbb{E}_{P_0}\vert  g^3(y)\vert < \infty$, $\mathbb{E}_{P_0}\Vert   g^2(y) x \Vert_2 < \infty$, and $\mathbb{E}_{P_0}\Vert  xx'\Vert_2 < \infty$, as noted in the main paper. 

For finite support $y \in \{0,\ldots, y_{max}\}$, we may consider the boundaries and the interior separately. The expectation $[P_0]$ may be separated into the expectations over $[y \mid x]$ and $[x]$, i.e.,
\begin{align*}
    -\mathbb{E}_{P_0} \log p_\theta(y \mid x) &= -\mathbb{E}_{P_0[x]} \mathbb{E}_{P_0[y \mid x]} \log p_\theta(y \mid x) \\
    &= -\mathbb{E}_{P_0[x]} \sum_{j=0}^{y_{max}} p_{0}(y=j \mid x) \log p_\theta(j \mid x) \\ 
    &=-\mathbb{E}_{P_0[x]} \{ p_{0}( y=0 \mid x) \log p_\theta(0 \mid x)+ p_{0}( y=y_{max} \mid x) \log p_\theta(y_{max} \mid x)\} \\
    &\quad \quad -  \mathbb{E}_{\tilde P_0} [\{1 - p_0(y=0 \mid x) - p_0(y = y_{max} \mid x)\}\log p_\theta(y  \mid x)] \\
    &\le -\mathbb{E}_{P_0[x]} (\log \Phi\{g(a_1); x'\theta, \sigma^2\}+ \log [1 - \Phi\{g(a_{y_{max}}); x'\theta, \sigma^2] ) \\
    & \quad \quad -  \mathbb{E}_{\tilde P_0}  \log p_\theta(y  \mid x)
\end{align*}
where $\tilde P_0 = P_0[x]\tilde P_0[y\mid x]$ and the latter term is $P_0[y \mid x]$ the truncated to $j \in \{1,\ldots, y_{max} - 1\}$, i.e., $\tilde p_0(y = j \mid x) = p_0(y = j \mid x)/\{1 - p_0(y=0 \mid x) - p_0(y = y_{max} \mid x)\}$ for $j \in \{1,\ldots, y_{max} - 1\}.$ Thus, the conditions for the $y \in \mathbb{Z}$ may be modified for the truncated distribution, e.g., $\mathbb{E}_{\tilde P_0}\vert  g^3(y)\vert < \infty$ and $\mathbb{E}_{\tilde P_0}\Vert   g^2(y) x \Vert_2 < \infty$, while $\mathbb{E}_{ P_0}\Vert  xx'\Vert_2 < \infty$ is unchanged. The remaining conditions only require expectations over $[x]$: $-\mathbb{E}_{P_0[x]} \log \Phi\{g(a_1); x'\theta, \sigma^2\} < \infty$ and $-\mathbb{E}_{P_0[x]} \log [1 - \Phi\{g(a_{y_{max}}); x'\theta, \sigma^2] < \infty$. 


\section{Analytic computations for prediction}\label{sec-pred-an}

When the transformation is fixed (or conditioned upon), the predictive distribution may be computed analytically without Monte Carlo sampling.  The predictive distribution is computable using marginal densities, $    p(\tilde y \mid y) = p(\tilde y, y)/p(y)$, where  $
    p(y) = \bar\Phi_n\{\mathcal{C} = g(\mathcal{A}_{y});  \mu_z,  \Sigma_z\}$ is the constant term in the denominator of \eqref{density-slct-n} and 
$\mu_z =  X  \mu_\theta$ and $\Sigma_z =  X  \Sigma_\theta  X' +  \sigma^2 I_n$ are determined by the linear model \eqref{lm}. The joint distribution $p(\tilde y, y)$ proceeds similarly via the latent data $(\tilde z, z)$: 
$p(\tilde y, y) = \bar\Phi_{n + \tilde n}\{\mathcal{C} = g(\mathcal{A}_{\tilde y}) \times g(\mathcal{A}_{y});  \mu_{\tilde z z},  \Sigma_{\tilde z z}\},$
where 
$$ \mu_{\tilde z z} = 
    \begin{pmatrix} 
        \tilde X \mu_\theta \\
        X \mu_\theta
    \end{pmatrix}, \quad 
    \Sigma_{\tilde z z} = 
    \begin{pmatrix}
        \tilde X  \Sigma_\theta  \tilde X' +  \sigma^2 I_{\tilde n} & \tilde X \Sigma_\theta X' \\
        X \Sigma_\theta \tilde X' & X  \Sigma_\theta  X' +  \sigma^2 I_n
    \end{pmatrix}
    $$
    for the linear model setting of Lemma~\ref{star-lm}. These direct probability computations of $p(\tilde y \mid y)$ are most useful when $\tilde n$ is small, such as $\tilde n =1$. The predictive probabilities also provide analytic  point predictions, such as
     $   \E(\tilde y \mid y) = \sum_{j \in {\rm supp}(y)} j \ p(\tilde y  =j \mid y).
    $ 
    

\section{Monte Carlo sampling for prediction}\label{a-mc-pred}
An alternative Monte Carlo sampler for the linear model circumvents the posterior sampling and instead draws directly from the latent posterior predictive distribution of $[\tilde z \mid y]$, which can be derived explicitly:
\begin{lemma}
\label{lem-pred-g}
    For the linear model \eqref{lm} with a Gaussian prior $ \theta \sim N( \mu_\theta,  \Sigma_\theta)$, the latent predictive distribution at $\tilde X$ is 
    $[ \tilde z \mid  y] \sim \mbox{SLCT-N}_{n, \tilde n}(X  \mu_\theta,  \tilde X\mu_\theta,   X  \Sigma_\theta  X' +  \sigma^2 I_n,  \tilde X  \Sigma_\theta  \tilde X' +  \sigma^2 I_{\tilde n},   X  \Sigma_\theta \tilde X',  \mathcal{C} = g(\mathcal{A}_{ y}))$.
\end{lemma}

    \begin{proof}(Lemma~\ref{lem-pred-g})
    As a preliminary, consider a linear combination of a selection normal variable  with a Gaussian innovation: 
\begin{lemma}\label{lin-combo}
    Suppose $\theta \sim \mbox{SLCT-N}_{n, p}( \mu_z,  \mu_\theta,  \Sigma_z,  \Sigma_\theta,  \Sigma_{z\theta},  \mathcal{C})$ and let $ \theta^* =  A \theta +  a$ where $ A$ is a fixed $q \times p$ matrix and $ a \sim N_q( \mu_a,  \Sigma_a)$ is independent of $ \theta$. Then $ \theta^* \sim \mbox{SLCT-N}_{n, q}( \mu_z,  \mu_{\theta^*} = A \mu_\theta +  \mu_a,  \Sigma_z,  \Sigma_{\theta^*}=  A  \Sigma_\theta  A' + \Sigma_a,  \Sigma_{z\theta^*} =  \Sigma_{z\theta}  A',  \mathcal{C})$. 
\end{lemma}
\begin{proof}(Lemma~\ref{lin-combo})
    First consider the distribution of $ A  \theta$. Since the moments of the joint distribution $( z,  \theta)$ are given by assumption, it follows that $( z,  A  \theta)$ is jointly Gaussian  and  $[A  \theta \mid z \in \mathcal{C}] \sim \mbox{SLCT-N}_{n, q}( \mu_z,  A \mu_\theta,  \Sigma_z,  A  \Sigma_{\theta} A',  \Sigma_{z\theta}  A',  \mathcal{C})$. Using the moment generating functions and noting independence between $ A \theta$ and $ a$, it follows that $M_{ \theta^*}( s) = M_{  A  \theta+ a}( s)= M_{  A  \theta}( s) M_{ a}( s)$ where $ M_{  A  \theta}$ is in Theorem~\ref{thm-slct} and $M_{ a}( s) = \exp( s'  \mu_a +  s'  \Sigma_a  s/2)$. The result of this product is the moment generating function of the stated selection normal distribution.
\end{proof}
    Returning to Lemma~\ref{lem-pred-g}, 
    the latent predictive variable $[\tilde z  \mid \theta]$ in  $  p(\tilde y \mid y) = \int p\{\tilde z \in g(\mathcal{A}_{\tilde y}) \mid \theta\} \ p(\theta \mid y) \ d\theta$ can be represented as $\tilde z = \tilde X \theta + \tilde \epsilon$, where $\tilde \epsilon \sim N_n(0, \sigma^2 I_n) $ is independent of $\{\epsilon_i\}_{i=1}^n$ from \eqref{lm}. This integrand requires marginalization over $\theta \sim p(\theta  \mid y)$, which is selection normal according to Lemma~\ref{star-lm}. Hence we obtain the latent predictive distribution $[\tilde z \mid y]$ via the marginalization from Lemma~\ref{lin-combo}. 
    \end{proof}

Direct Monte Carlo simulation proceeds using Theorem~\ref{thm-mc} applied to the latent predictive selection normal distribution.  Specifically, consider the $g$-prior $\Sigma_\theta = \psi \sigma^2 ( X' X)^{-1}$ for $\psi >0$. Under model \eqref{tr}--\eqref{lm}, a joint predictive draw $\tilde y^* \sim p(\tilde y \mid y)$ at the $\tilde n \times p$ covariate matrix  $\tilde X$ is obtained using Algorithm~\ref{alg:g-pred}. 
    \begin{algorithm}[h]
\SetAlgoLined  
\begin{enumerate}
    \item Simulate $V_0^* \sim N_n(0, \sigma^2\{ \psi  X ( X' X)^{-1}  X' + I_n\})$ truncated to $g(\mathcal{A}_y) - X\mu_\theta$.
    \item Simulate $\tilde  V_1^* \sim N_{\tilde n}(0, \sigma^2\{ \psi(1+\psi)^{-1} \tilde X(X'X)^{-1}\tilde X' +   I_{\tilde n}\})$.
    \item Compute  $\tilde z^* = \tilde X \mu_\theta +  \tilde V_1^* + \psi(1+\psi)^{-1} \tilde X ( X' X)^{-1} X' V_0^*$.
    \item Set $\tilde y^* = h\circ g^{-1}(\tilde z^*)$.
\end{enumerate} 
 \caption{Monte Carlo sampling for $\tilde y^* \sim p(\tilde y \mid y)$  under the $g$-prior.} \label{alg:g-pred}
\end{algorithm}
\begin{proof}(Algorithm~\ref{alg:g-pred})
    First, $\Sigma_z =\sigma^2\{ I_n + \psi  X ( X' X)^{-1}  X'\}$ is unchanged from Algorithm~\ref{alg:g-sim}. Similarly, using the usual notation, we  have
     $\Sigma_{z\theta}'\Sigma_z^{-1}  =\psi(1+\psi)^{-1} \tilde X ( X' X)^{-1} X'$ and $\Sigma_{z\theta}'\Sigma_z^{-1}\Sigma_{z\theta} = \sigma^2\psi^2(1+\psi)^{-1}  \tilde  X ( X' X)^{-1}\tilde  X'$. Computing $\Sigma_\theta -  \Sigma_{z\theta}'\Sigma_z^{-1}  \Sigma_{z\theta} =  \sigma^2\{ \psi(1+\psi)^{-1} \tilde X(X'X)^{-1}\tilde X' +   I_{\tilde n}\}$ shows the result for $\tilde z^*$. Since the link in \eqref{tr} is deterministic, the  result for $\tilde y^*$ follows. 
    \end{proof}

Algorithm~\ref{alg:g-pred} provides direct Monte Carlo sampling from the joint posterior predictive distribution using standard regression functionals of the covariate matrix $X$ and the predictive covariate matrix $\tilde X$.  By comparison, the samplers in Algorithm~\ref{alg:g-sim} and Algorithm~\ref{alg:g-pred} share the same simulation step for $V_0^*$, but the second step differs: Algorithm~\ref{alg:g-sim} requires a draw from a $p$-dimensional multivariate normal distribution, while Algorithm~\ref{alg:g-pred} draws from a $\tilde n$-dimensional multivariate normal distribution. As a result,  Algorithm~\ref{alg:g-pred} may provide computational savings relative to Lemma~\ref{lem-pred} when $\tilde n$ is small,  $p$ is large, and  the posterior samples of $\theta$ are not needed. 

\section{Model selection}\label{sec-model-sel}
Suppose we have candidate models $\mathcal{M}_1,\ldots, \mathcal{M}_K$, where each $\mathcal{M}_k$ is of the form \eqref{tr}--\eqref{lm}. Each candidate model differs in a fixed attribute: examples include the transformation $g$, the hyperparameters in the model (such as $\psi$), and linear vs.\! nonlinear models. This setting also applies to variable selection by replacing each covariate matrix $X$ with $X_\gamma$, where $\gamma \in \{0,1\}^p$ is a variable inclusion indicator and each $\gamma$ corresponds to a model $\mathcal{M}_k$. We focus on two goals: model selection and prediction via model-averaging. 

Given model probabilities $p(\mathcal{M}_k)$ for $k=1,\ldots,K$, the posterior probability of each model is  
\begin{equation} \label{mprob}
    p(\mathcal{M}_k \mid y) = \frac{p(y \mid \mathcal{M}_k) p(\mathcal{M}_k)}{\sum_{h = 1}^K p(y \mid \mathcal{M}_h) p(\mathcal{M}_h)}
\end{equation}
 where $p(y \mid \mathcal{M}_k)$ denotes the marginal likelihood under $\mathcal{M}_k$. When the prior on $\theta$ and the latent continuous data model for $z$ are both Gaussian  for each $\mathcal{M}_k$, we derive the posterior probabilities directly:
 \begin{lemma}\label{lem-mprob}
     For each model $\{\mathcal{M}_k\}_{k=1}^K$ of the form \eqref{tr}--\eqref{lm} with transformation $g^{(k)}$, suppose that $(\theta^{(k)}, z^{(k)})$ are jointly Gaussian with $z^{(k)} \sim N_n(\mu_z^{(k)},  \Sigma_z^{(k)})$. Then the posterior model probabilities are 
     \begin{equation} \label{mprob-n}
    p(\mathcal{M}_k \mid y) = \frac{\bar\Phi_n\{\mathcal{C} = g^{(k)}(\mathcal{A}_{y});  \mu_z^{(k)},  \Sigma_z^{(k)}\}p(\mathcal{M}_k)}{\sum_{h = 1}^K \bar\Phi_n\{\mathcal{C} = g^{(h)}(\mathcal{A}_{y});  \mu_z^{(h)},  \Sigma_z^{(h)}\} p(\mathcal{M}_h)}.
\end{equation}
 \end{lemma}
The result follows from the marginal likelihood $p(y)$ (see the denominator of \eqref{density-slct-n})  applied to each model $\mathcal{M}_k$. 
Lemma~\ref{lem-mprob} is naturally informative for model selection, for example by selecting $\mathcal{M}_k$ for which $p(\mathcal{M}_k \mid y)$ is largest. More broadly, the posterior probabilities in \eqref{mprob-n} enable model uncertainty in posterior and predictive inference. Focusing on the latter setting, consider the goal of prediction that averages over the uncertainty about the models $\{\mathcal{M}_k\}_{k=1}^K$. Using the  predictive simulators, which now depend on the specific model $\mathcal{M}_k$, we propose the following Monte Carlo sampler: for $s=1,\ldots, S$,
    sample $k^s \sim \mbox{Categorical}(\pi_1, \ldots, \pi_K)$ with $\pi_k  = p(\mathcal{M}_k \mid y)$ from \eqref{mprob-n};
    then sample $\tilde y^s \sim p(\tilde y \mid y, \mathcal{M} = \mathcal{M}_{k^s})$ from Lemma~\ref{lem-pred} or Algorithm~\ref{alg:nl-pred}. 
This sampler produces joint and independent draws of $\{(\mathcal{M}^s, \tilde y^s)\}_{s=1}^S \sim p(\mathcal{M}, \tilde y \mid y)$. By computing functionals of $\{\tilde y^s\}$ only, we effectively marginalize over the model uncertainty within $\{\mathcal{M}_k\}_{k=1}^K$ yet retain a discrete and joint predictive distribution at $\tilde X$.

The sparse means model with prior \eqref{ssmod} offers certain computational simplifications. In particular, inference on the inclusion indicators $[\gamma \mid y]$ proceeds via Lemma~\ref{lem-mprob}. Since the marginal latent mean is zero and the latent covariance is diagonal, $\Sigma_z^{(\gamma)} = \sigma^2 \mbox{diag}\{1 + \psi \gamma_i\}_{i=1}^n$, the key term in \eqref{mprob-n} simplifies considerably: 
\begin{equation}
    \label{mprob-n-simp}
    \bar\Phi_n\{\mathcal{C} = g(\mathcal{A}_{y});  0,  \Sigma_z^{(\gamma)}\} = \prod_{i=1}^n \int_{g(a_{y_i})}^{g(a_{y_i+1})} \phi_1\{x_i; 0, \sigma^2(1 + \psi \gamma_i)\} \ d  x_i.
\end{equation}
 In place of an $n$-dimensional integral of a multivariate normal density, we obtain a product of $n$ univariate normal integrals, which is a massive advantage for computing with large $n$.

\section{Additional simulation results}\label{sec-sims-add}
To evaluate the regression coefficients  in the simulation study from Section~\ref{sims-reg}, we compute the Pearson correlation between the true and estimated (posterior expectation) regression coefficients, excluding the intercept. The linear coefficients for the semiparametric and Negative Binomial regressions are not on the same scale, so the correlation helps provide a reasonable metric for accuracy. The results are in Figure~\ref{fig:cor}. Under the \texttt{NegBin} scenario, all methods perform about the same. Yet under the \texttt{SemiPar} scenario, the Poisson estimates are substantially less accurate, while the Negative Binomial estimates are  moderately less accurate. As in the case of prediction, the proposed approximations to the fully BNP procedure do not appear to inhibit estimation, and thus are valuable options for improved computational efficiency and numerical stability. 

\begin{figure}[h]
\centering
\includegraphics[width=.45\textwidth]{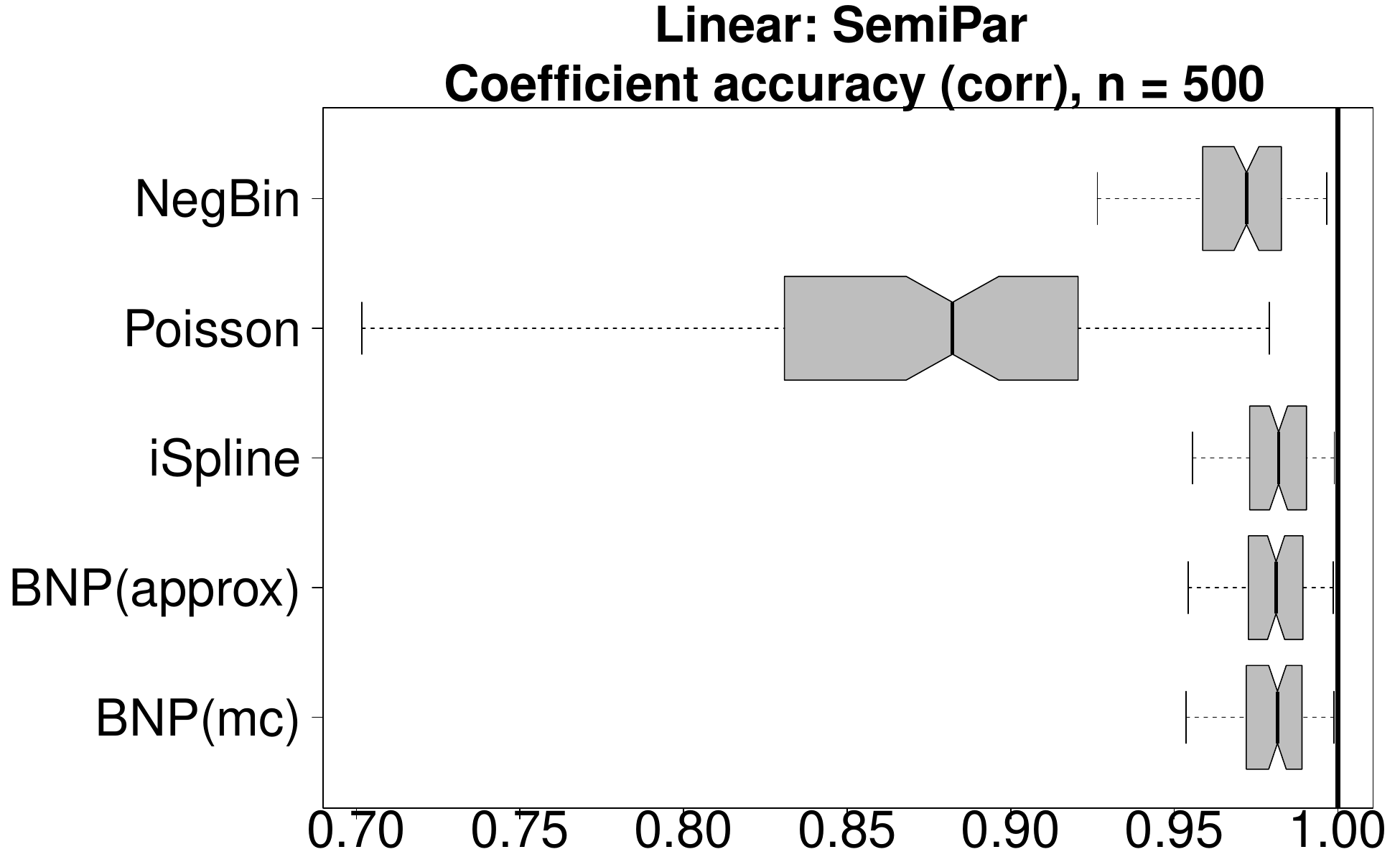}
\includegraphics[width=.45\textwidth]{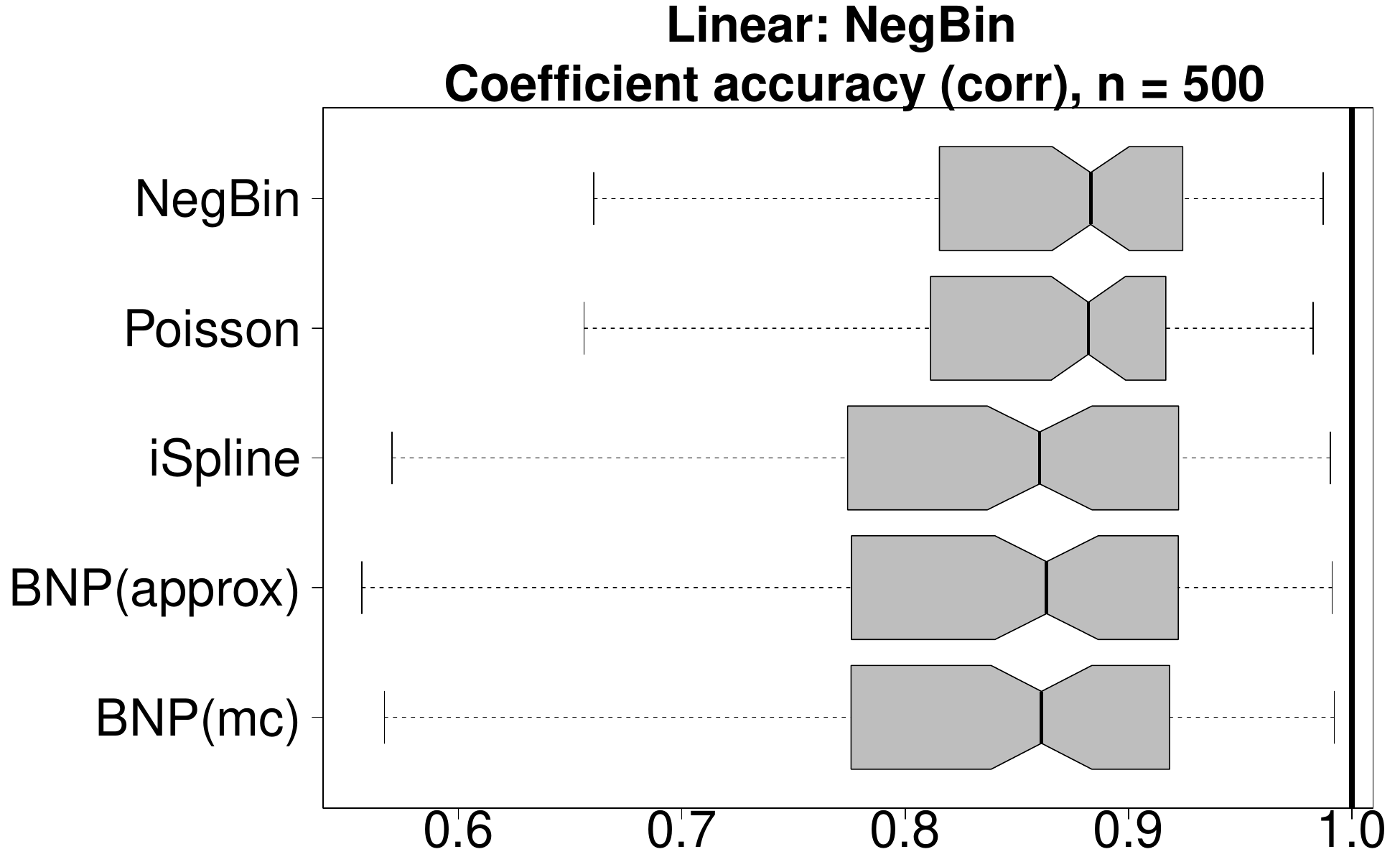}
\includegraphics[width=.45\textwidth]{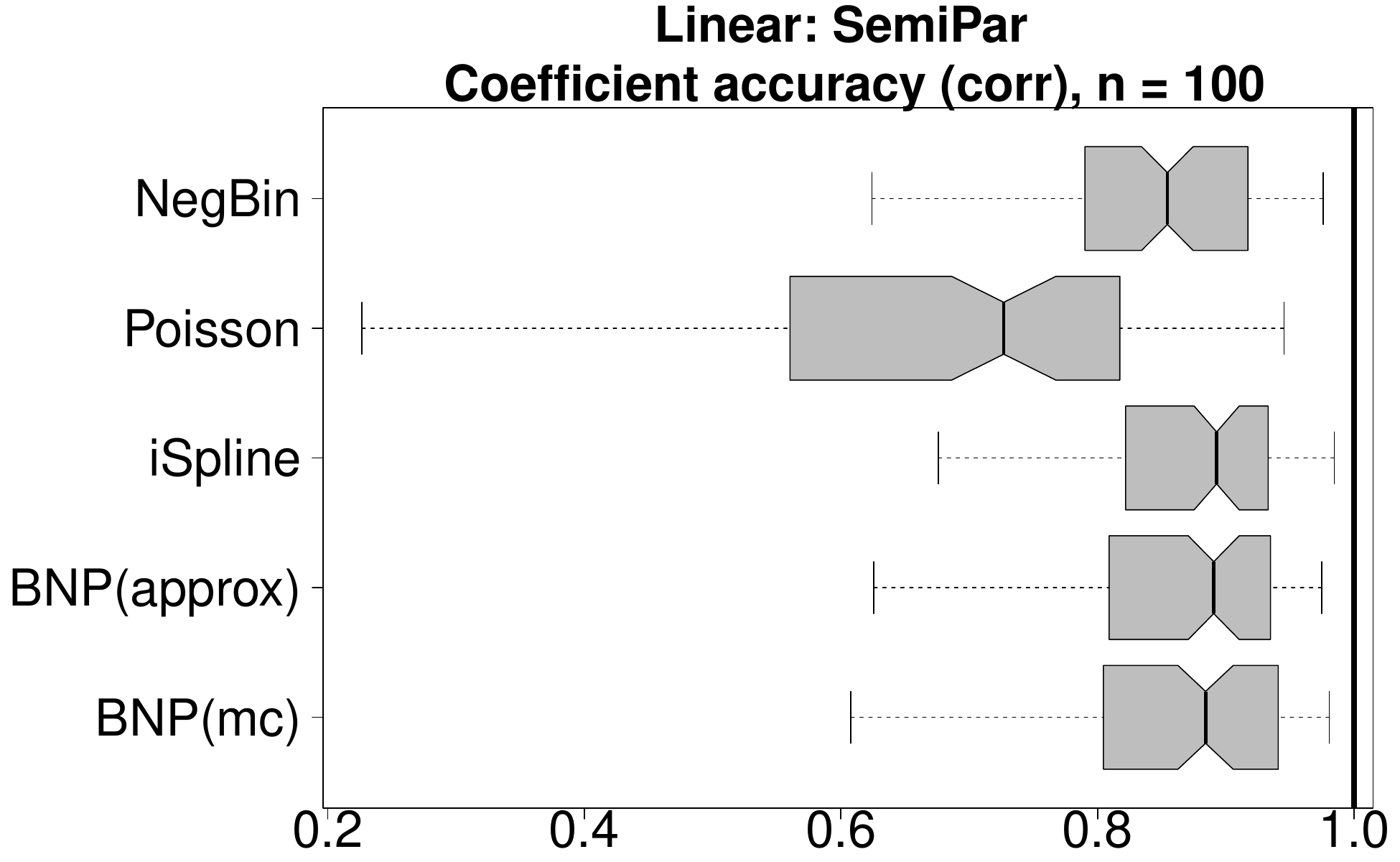}
\includegraphics[width=.45\textwidth]{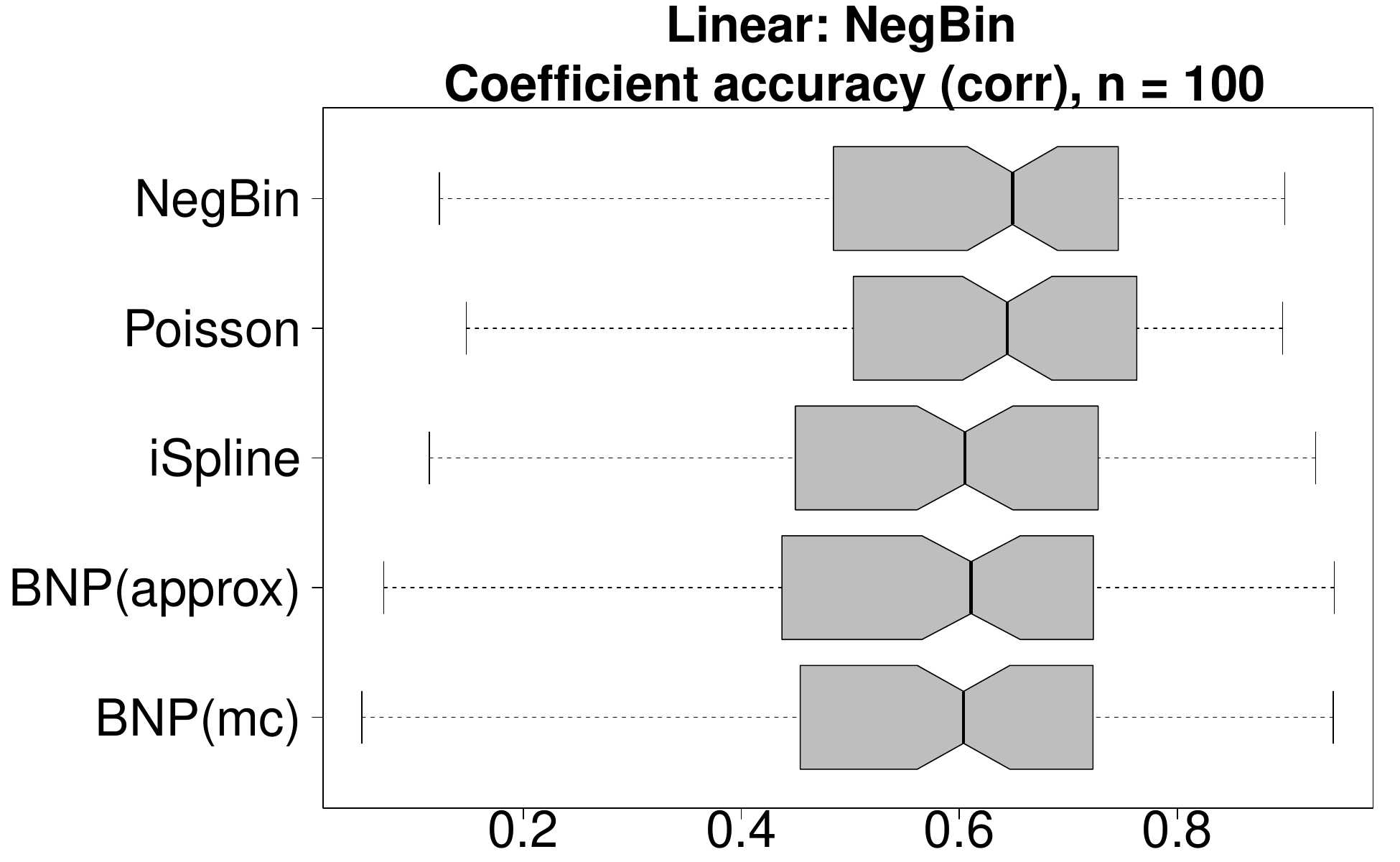}
\caption{\small Correlation between the estimated and true regression coefficients for \texttt{SemiPar} (left) and \texttt{NegBin} (right) data, $n \in \{100,500\}$ (positively oriented). The proposed BNP approaches are much more accurate for \texttt{SemiPar} data yet still competitive for \texttt{NegBin} data. 
\label{fig:cor}}
\end{figure}


\section{Additional details about the data}\label{sec-app-add}
The variables from the self-reported mental data are presented and summarized in Table~\ref{tab:data}, which includes active links to the NHANES survey questions and additional details about each variable.

\begin{table}[h]
\centering
\begin{tabular}{m{6cm} m{7cm}} 
\textbf{Variable} &  \textbf{Values} \\
\hline
   \rowcolor{shade}
 \multicolumn{2}{l}{\textbf{Response variable:}} \\
\href{https://wwwn.cdc.gov/Nchs/Nhanes/2011-2012/HSQ_G.htm#HSQ480}{DaysMentHlthNotGood} &  $\{0, 1, \ldots, 30\}$ \\
 \hline
  \rowcolor{shade}
  \multicolumn{2}{l}{\textbf{Demographic and socioeconomic variables:}} \\
 \href{https://wwwn.cdc.gov/nchs/nhanes/2011-2012/demo_g.htm#RIAGENDR}{Gender} &  {\emph{Male} ($54\%$), Female ($46\%$)} \\
 \rowcolor{shade}
 \href{https://wwwn.cdc.gov/nchs/nhanes/2011-2012/demo_g.htm#RIDAGEYR}{Age (years)}  &   $\{20, \ldots, 59\}$  \\
 \href{https://wwwn.cdc.gov/nchs/nhanes/2011-2012/demo_g.htm#RIDRETH1}{Race$^*$} &  {\emph{White} ($40\%$), Black ($25\%$), Hispanic ($20\%$), Other ($15\%$)} \\
  \rowcolor{shade}
 \href{https://wwwn.cdc.gov/nchs/nhanes/2011-2012/demo_g.htm#DMDMARTL}{Marital Status$^*$} &{\text{\emph{Married/Partner} ($58\%$)}, \text{Widowed/Divorced/Separated ($14\%$)}, \text{Single ($28\%$)}}\\
 \href{https://wwwn.cdc.gov/nchs/nhanes/2011-2012/demo_g.htm#INDFMIN2}{Family Income Level$^*$} &  {Low (26\%), \emph{Middle (57\%)}, High (17\%)}\\
  \rowcolor{shade}
 \href{https://wwwn.cdc.gov/nchs/nhanes/2011-2012/demo_g.htm#DMDHREDU}{Education Level$^*$} & {\emph{$<$ HS} ($19\%$), = HS ($19\%$), $>$ HS ($62\%$)}\\
\href{https://wwwn.cdc.gov/Nchs/Nhanes/2011-2012/HIQ_G.htm#HIQ011}{Uninsured$^*$}&{Yes ($30\%$), \emph{No}  ($70\%$)} \\
 \hline
   \rowcolor{shade}
 \multicolumn{2}{l}{\textbf{Alcohol and drug use variables:}} \\
 \href{https://wwwn.cdc.gov/Nchs/Nhanes/2011-2012/ALQ_G.htm#ALQ151}{HeavyDrinker} & {Yes ($16\%$), \emph{No} ($84\%$)}\\
   \rowcolor{shade}
 \href{https://wwwn.cdc.gov/Nchs/Nhanes/2011-2012/SMQRTU_G.htm#SMQ680}{UseNicotine} & {Yes ($31\%$), \emph{No}  ($69\%$)}\\
 \href{https://wwwn.cdc.gov/Nchs/Nhanes/2011-2012/DUQ_G.htm#DUQ200}{UsedMarijuana} & {Yes ($60\%$), \emph{No}  ($40\%$)}\\
   \rowcolor{shade}
 \href{https://wwwn.cdc.gov/Nchs/Nhanes/2011-2012/DUQ_G.htm#DUQ240}{UsedHardDrug} & {Yes ($20\%$), \emph{No}  ($80\%$)}\\
\hline
\multicolumn{2}{l}{\textbf{Health-related variables:}} \\
  \rowcolor{shade}
\href{https://wwwn.cdc.gov/nchs/nhanes/2011-2012/BMX_G.htm#BMXBMI}{Body Mass Index (BMI, kg/$m^2$)} & $[13.6, 69.0]$  \\
\href{https://wwwn.cdc.gov/nchs/nhanes/2011-2012/TCHOL_G.htm#LBXTC}{Cholesterol (total, mg/dL)} & $[59.0, 69.0]$ \\
  \rowcolor{shade}
\href{http://data.nber.org/nhanes/2011-2012/BPQ_G.htm}{HasHighBP (BPQ020 at link) }& {Yes ($25\%$), \emph{No}  ($75\%$)} \\
\href{http://data.nber.org/nhanes/2011-2012/BPQ_G.htm}{HasHighChol (BPQ080 at link) }& {Yes ($25\%$), \emph{No}  ($75\%$)} \\
  \rowcolor{shade}
\href{https://wwwn.cdc.gov/Nchs/Nhanes/2011-2012/DIQ_G.htm#DIQ010}{HasDiabetes$^*$}&{Yes ($9\%$), \emph{No}  ($91\%$)} \\
 \hline
\multicolumn{2}{l}{\textbf{Survey design variables:}} \\
\rowcolor{shade}
\href{https://wwwn.cdc.gov/Nchs/Nhanes/2011-2012/DEMO_G.htm#WTMEC2YR}{Sampling Weights}& $[0, 222579.8]$ \\ [1ex] 
\href{https://wwwn.cdc.gov/Nchs/Nhanes/2011-2012/DEMO_G.htm#SDMVPSU}{Primary Sampling Units}&  \emph{1} (48\%), 2 (44\%), 3 (8\%)
\\ [1ex]
\end{tabular}
\caption{\small Variables in the analysis dataset with hyperlinks to the online NHANES descriptions. The baseline categories for one-hot/dummy encoding are italicized. The continuous variables (Age, BMI, and Cholesterol) are centered and scaled prior to model fitting.  Annotated variables ($*$) include minor modifications (e.g., collapsed categories) from the original NHANES variables; see the online Data Dictionary in \cite{Kowal2021b} for additional details. \label{tab:data}}
\end{table}

\clearpage 
\bibliographystyle{apalike}
\bibliography{refs.bib}

\end{document}